\documentclass{imsart}

\usepackage[colorlinks,citecolor=blue,urlcolor=blue]{hyperref}

\usepackage{amsmath}
\usepackage{graphicx}
\usepackage{enumerate}
\usepackage{natbib}
\usepackage{url}
\usepackage{algorithm}
\usepackage{algorithmic}
\usepackage{amssymb}
\usepackage{xargs}
\usepackage{ifthen}
\usepackage{bbm}
\usepackage{aliascnt}
\usepackage{stmaryrd}
\usepackage{mathtools}
\mathtoolsset{showonlyrefs}
\usepackage{float}
\usepackage{ushort}
\usepackage{amsthm}
\usepackage{secdot}
\usepackage{subcaption}
\usepackage{xcolor}
\usepackage{comment}
\usepackage{booktabs}
\usepackage{mathtools}

\newtheorem{theorem}{Theorem}[section]
\newaliascnt{proposition}{theorem}
\newtheorem{proposition}[proposition]{Proposition}
\aliascntresetthe{proposition}

\newaliascnt{lemma}{theorem}
\newtheorem{lemma}[lemma]{Lemma}
\aliascntresetthe{lemma}

\newaliascnt{corollary}{theorem}
\newtheorem{corollary}[corollary]{Corollary}
\aliascntresetthe{corollary}

\newaliascnt{definition}{theorem}

\aliascntresetthe{definition}

\newaliascnt{example}{theorem}

\aliascntresetthe{example}

\newaliascnt{remark}{theorem}
\newtheorem{remark}[remark]{Remark}
\aliascntresetthe{remark}

\newtheorem{assumption}{Assumption}

\newcounter{hypA}

\newcounter{hypB}

\def\equationautorefname~#1\null{%
  Equation~(#1)\null
}

\DeclarePairedDelimiter\floor{\lfloor}{\rfloor}
\DeclarePairedDelimiter\ceil{\lceil}{\rceil}
\newcommand{\1}[1]{\mathbbm{1}_{#1}}

\newcommandx\A[2][1=]{
\ifthenelse{\equal{#1}{}}
{\hspace{-1mm}(\textbf{A\ref{#2}})\hspace{-1mm}}
{\hspace{-1mm}(\textbf{A\ref{#1}--\ref{#2}})\hspace{-1mm}}
}
\newcommand{\af}[1]{h_{#1}} 
\newcommand{\alg}[1]{\mathcal{#1}} %Algebra
\newcommand{\am}[1]{\vartheta_{#1}} %auxiliary weights
\newcommand{\amm}[1]{\boldsymbol{\vartheta}_{#1}}
\newcommand{\amup}{\bar{\gamma}}

\newcommand{\addf}[1]{\termletter_{#1}} %Additive functional
\newcommand{\addfext}[2]{\bar{h}_{#1 | #2}} %Additive functional extension
\newcommand{\adds}[1]{\af{#1}} %Sum of additive functionals
\newcommand{\apf}{\mathsf{APF}}
\newcommandx{\arr}[2][1=]{
\ifthenelse{\equal{#1}{}}
{\upsilon_{\N}^{#2}}
{(\upsilon_{\N}^{#2})^{#1}}
}
\newcommandx{\arrterm}[3][1=]{
\ifthenelse{\equal{#1}{}}
{\tilde{\upsilon}_{\N}(#3,#2)}
{\tilde{\upsilon}_{\N}^{#1}(#3,#2)}
}
\newcommand{\arru}[1]{\tilde{\upsilon}_{\N}^{#1}}
\newcommandx{\asvar}[4][1=]{
\ifthenelse{\equal{#1}{}}
{\sigma_{#2} \langle #3, #4 \rangle(\af{#2})}
{\sigma_{#2}^2 \langle #3, #4 \rangle(\af{#2})}
}
\newcommandx{\asvarFFBSm}[4][1=]{
\ifthenelse{\equal{#1}{}}
{\tilde{\sigma}_{#2} \langle #3, #4 \rangle(\af{#1})}
{\tilde{\sigma}_{#2}^2 \langle #3, #4 \rangle(\af{#2})}
}
\newcommandx{\asvarstd}[2][1=]{
\ifthenelse{\equal{#1}{}}
{\sigma_{#2}(\af{#2})}
{\sigma_{#2}^2(\af{#2})}
}
\newcommandx{\asvarFFBSmstd}[2][1=]{
\ifthenelse{\equal{#1}{}}
{\tilde{\sigma}_{#2}(\af{#2})}
{\tilde{\sigma}_{#2}^2(\af{#2})}
}

\newcommandx\B[2][1=]{
\ifthenelse{\equal{#1}{}}
{\hspace{-1mm}(\textbf{B\ref{#2}})\hspace{-1mm}}
{\hspace{-1mm}(\textbf{B\ref{#1}--\ref{#2}})\hspace{-1mm}}
}
\newcommandx{\BF}[3][1=]{
\ifthenelse{\equal{#1}{}}
{\kernel{D}_{#2, #3}}
{\kernel{D}_{#2, #3}^{#1}}
}
\newcommandx{\BFcent}[3][1=]{
\ifthenelse{\equal{#1}{}}
{\tilde{\kernel{D}}_{#2, #3}}
{\tilde{\kernel{D}}_{#2, #3}^{#1}}
}
\newcommandx{\BFm}[2]{
	\boldsymbol{\mathcal{D}}_{#1, #2}
}
\newcommandx{\BFcentm}[2]{
	\tilde{\boldsymbol{\mathcal{D}}}_{#1, #2}
}
\newcommand{\bi}[3]{J_{#1}^{(#2, #3)}}
\newcommandx{\bk}[2][1=]{ 
\ifthenelse{\equal{#1}{}}
{\overleftarrow{\kernel{Q}}_{#2}}
{\overleftarrow{\kernel{Q}}_{#2}^{#1}}
}
\newcommand{\bkm}[1]{\kernel{B}_{#1}} %backward kernel
\newcommand{\bkmm}[1]{\boldsymbol{\mathcal{B}}_{#1}}
\newcommand{\bigadds}[1]{\boldsymbol{h}_{#1}}
\newcommand{\bigaddf}[1]{\tilde{\boldsymbol{h}}_{#1}}

\newcommand{\bmf}[1]{\set{F}(#1)} %Bounded measurable function
\newcommand{\borel}[1]{\mathcal{B}(\set{#1})} %Borel sigma-algebra

\newcommand{\catdist}{\mathsf{Cat}}
\newcommandx{\cexp}[3][1=]{
\ifthenelse{\equal{#1}{}}
{\mathbb{E}\left[ #2 \mid #3 \right]} % conditional expectation
{\mathbb{E}[ #2 \mid #3 ]}
}

\newcommand{\convd}{\overset{\mathcal{D}}{\longrightarrow}}
\newcommand{\convp}{\overset{\prob}{\longrightarrow}}

\newcommand{\dd}{\Delta}

\newcommand{\E}{\mathbb{E}}
\newcommand{\enoch}[2]{E_{#1}^{#2}} % Enoch ancestor indices
\newcommand{\eg}{\emph{e.g.}}
\newcommand{\epart}[2]{\xi_{#1}^{#2}}
\newcommand{\epartm}[2]{\boldsymbol{\xi}_{#1}^{#2}}
\newcommand{\eparttil}[2]{\tilde{\xi}_{#1}^{#2}}
\newcommand{\epartbar}[2]{\bar{\xi}_{#1}^{#2}}
\newcommand{\eqdef}{\coloneqq} %Equal by definition
\newcommand{\eqdist}{\overset{\mathcal{D}}{=}}
\newcommand{\eqsp}{}
\newcommand{\ess}{\mathsf{ESS}}

\newcommand{\g}[1]{g_{#1}} %g density

\newcommandx{\genfd}[1][1=]{
\ifthenelse{\equal{#1}{}}
{\mathcal{F}}
{\mathcal{F}_{\N}}
}

\newcommand{\hbd}{| \tilde{h} |_{\infty}}
\newcommand{\hk}{\kernel{P}} %Hidden kernel
\newcommand{\hkm}[1]{\boldsymbol{\mathcal{P}}_{\! #1}} %Hidden kernel
\newcommand{\hklow}{\ushort{\varepsilon}}
\newcommand{\hkup}{\bar{\varepsilon}}
\newcommand{\hd}{p} %Hidden density
\newcommand{\hdm}[1]{\boldsymbol{p}_{#1}} 

\newcommand{\I}[2]{I_{#1}^{#2}} %resamplings
\newcommand{\ie}{\emph{i.e.}}
\newcommand{\intvect}[2]{\llbracket #1, #2 \rrbracket}

\newcommand{\kernel}[1]{\mathbf{#1}}
\newcommand{\kletter}{M}
\newcommandx{\K}[1][1=]{
	\ifthenelse{\equal{#1}{}}{{\kletter}}{{M^{#1}}}
}
\newcommand{\ld}[1]{\ell_{#1}} %unnormalized "\ell" density

\newcommand{\ldm}[1]{\boldsymbol{\ell}_{#1}}

\newcommandx{\lebfun}[1][1=]{
\ifthenelse{\equal{#1}{}}
{\lebfunletter}
{\lebfunletter_{#1}}
}

\newcommand{\lebfunletter}{\varphi}
\newcommand{\lk}[1]{\kernel{L}_{#1}} %L kernel
\newcommand{\lkm}[1]{\boldsymbol{\mathcal{L}}_{#1}}

\newcommand{\backprob}[2]{\mathbf{\Lambda}_{#1}^{#2}}

\newcommand{\maxd}{d}
\newcommand{\mdup}{\bar{\delta}}
\newcommand{\meas}[1]{\mathsf{M}(#1)}
\newcommand{\mr}{\varrho}

\newcommand{\N}{N}

\newcommand{\nset}{\mathbb{N}}
\newcommand{\nsetpos}{\mathbb{N}^*}

\newcommand{\ordo}{\mathcal{O}}
\newcommandx{\oscn}[2][1=]{
\ifthenelse{\equal{#1}{}}{\operatorname{osc}(#2)}{\operatorname{osc}^{#1}(#2)}
}

\newcommand{\partfilt}[1]{\mathcal{G}_{#1}^{\N}}
\newcommand{\partfiltbar}[1]{\mathcal{F}_{#1}^\N }

\newcommandx\post[2][1=]{
\ifthenelse{\equal{#1}{}}
	{\phi_{#2}}
	{\phi_{#2}^\N}
}
\newcommandx\postm[2][1=]{
	\ifthenelse{\equal{#1}{}}
	{\boldsymbol{\phi}_{#2}}
	{\boldsymbol{\phi}_{#2}^\N}
}

\newcommand{\prob}{\mathbb{P}} %Probability measure

\newcommand{\probmeas}[1]{\mathsf{M}_1(#1)}
\newcommand{\proj}{\boldsymbol{\Pi}}

\newcommand{\refM}[1]{\mu_{#1}} %Reference measure
\newcommandx\res[2][1=]{
	\ifthenelse{\equal{#1}{}}
	{\rho_{#2}}
	{\rho_{#2}^{#1}}
}

\newcommand{\rset}{\mathbb{R}}

\newcommand{\rsetpos}{\mathbb{R}_+^*}

\newcommand{\set}[1]{\mathsf{#1}} %Set

\newcommand{\sqmeas}[1]{\gamma_{#1}}
\newcommand{\sqc}[1]{\eta_{#1}}
\newcommand{\supn}[1]{\|#1\|_{\infty}} %Sup-norm

\newcommand{\termletter}{\tilde{h}}
\newcommand{\tensprod}{\varotimes}
\newcommand{\testfsymb}{f}
\newcommandx{\testf}[1][1=]{  %Test function
\ifthenelse{\equal{#1}{}}{\testfsymb}{\testfsymb_{#1}}
}
\newcommand{\testfpsymb}{\tilde{f}}
\newcommandx{\testfp}[1][1=]{  %Test function
\ifthenelse{\equal{#1}{}}{\testfpsymb}{\testfpsymb_{#1}}
}
\newcommandx{\testfbar}[1][1=]{  %Test function
\ifthenelse{\equal{#1}{}}{\bar{f}}{\bar{f}_{#1}}
}
\newcommand{\tk}[1]{\boldsymbol{\mathcal{T}}_{\!#1}}

\newcommand{\tstatletter}{\kernel{T}}
\newcommandx\tstat[2][1=]{
\ifthenelse{\equal{#1}{}}
	{\tstatletter_{#2}}
	{\tau_{#2}^{#1}}
}
\newcommandx\tstatm[2][1=]{
\ifthenelse{\equal{#1}{}}
	{\boldsymbol{\tstatletter}_{#2}}
	{\boldsymbol{\tau}_{#2}^{#1}}
}

\newcommand{\unorm}[1]{\varphi_{#1}}
\newcommand{\wgt}[2]{\omega_{#1}^{#2}}
\newcommand{\wgtm}[2]{\boldsymbol{\omega}_{#1}^{#2}}
\newcommand{\wgtsum}[1]{\Omega_{#1}}
\newcommand{\wgtsumm}[1]{\boldsymbol{\Omega}_{#1}}
\newcommand{\wgttil}[2]{\tilde{\omega}_{#1}^{#2}}
\newcommand{\wgtsumtil}[1]{\tilde{\Omega}_{#1}}
\newcommand{\wgtbar}[2]{\bar{\omega}_{#1}^{#2}}
\newcommand{\wgtsumbar}[1]{\bar{\Omega}_{#1}}

\newcommand{\Xinit}{\chi}
\newcommand{\Xinitm}{\boldsymbol{\chi}}
\newcommand{\Xpalg}[1]{\boldsymbol{\mathcal{X}}_{\!#1}}
\newcommand{\Xp}[1]{\boldsymbol{\mathsf{X}}_{#1}}

\newcommand{\xvec}[1]{\boldsymbol{x}_{#1}}

\newcommand{\Mletter}{\varepsilon}
\newcommandx\M[2][1=]{  % backward draws
	\ifthenelse{\equal{#1}{}}
	{\Mletter_{#2}}
	{\Mletter_{#2}^{#1}}
}
\DeclareMathOperator*{\argmax}{arg\text{ }max}

\startlocaldefs
\endlocaldefs

\allowdisplaybreaks
\begin{document}
	\noindent {\footnotesize This is an original manuscript of an article published by Taylor \& Francis in the Journal of the American Statistical Association (JASA) on 10 October 2022, available online: \url{https://www.tandfonline.com/doi/full/10.1080/01621459.2022.2118602}.}
\begin{frontmatter}
	\title{Fast and numerically stable particle-based online additive smoothing: the AdaSmooth algorithm}
	%\title{A sample article title with some additional note\thanksref{T1}}
	\runtitle{Fast and stable online additive smoothing: the AdaSmooth algorithm}
	%\thankstext{T1}{A sample of additional note to the title.}
	
	\begin{aug}
		\author[A]{\fnms{Alessandro} \snm{Mastrototaro}\ead[label=e1]{alemas@kth.se}},
		\author[A]{\fnms{Jimmy} \snm{Olsson}\ead[label=e2]{jimmyol@kth.se}}
		\and
		\author[B]{\fnms{Johan} \snm{Alenl\"{o}v}\ead[label=e3]{johan.alenlov@liu.se}}
		
		%%%%%%%%%%%%%%%%%%%%%%%%%%%%%%%%%%%%%%%%%%%%%%
		%% Addresses                                %%
		%%%%%%%%%%%%%%%%%%%%%%%%%%%%%%%%%%%%%%%%%%%%%%
		
		\address[A]{Department of Mathematics, KTH Royal Institute of Technology, Stockholm, \printead{e1,e2}}
		\address[B]{Department of Computer and Information Science, Link\"{o}ping University, Link\"{o}ping, Sweden, \printead{e3}}
	\end{aug}
	
	\begin{abstract}
		We present a novel sequential Monte Carlo approach to online smoothing of additive functionals in a very general class of path-space models. Hitherto, the solutions proposed in the literature suffer from either long-term numerical instability due to particle-path degeneracy or, in the case that degeneracy is remedied by particle approximation of the so-called backward kernel, high computational demands. In order to balance optimally computational speed against numerical stability, we propose to furnish a (fast) naive particle smoother, propagating recursively a sample of particles and associated smoothing statistics, with an adaptive backward-sampling-based updating rule which allows the number of (costly) backward samples to be kept at a minimum. This yields a new, function-specific additive smoothing algorithm, AdaSmooth, which is computationally fast, numerically stable and easy to implement. The algorithm is provided with rigorous theoretical results guaranteeing its consistency, asymptotic normality and long-term stability as well as numerical results demonstrating empirically the clear superiority of AdaSmooth to existing algorithms.
	\end{abstract}
	
	
\end{frontmatter}

%%%%%%%%%%%%%%%%%%%%%%%%%%%%%%%%%%%%%%%%%%%%%%
%%%% Main text entry area:

\section{INTRODUCTION}\label{sec:intro}
\subsection{Background} \label{sec:background}
 
We consider a general path-space model comprising general measurable spaces $(\set{X}_n,\alg{X}_n)_{n\in\nset}$ and unnormalized transition densities $(\ld{n})_{n \in \nset}$, where for every $n \in \nset$, $\ld{n}$ is a nonnegative measurable function on $\set{X}_n \times \set{X}_{n+1}$ such that $\sup_{x_n\in\set{X}_n}\int\ld{n}(x_n,x_{n+1})\,dx_{n+1}<\infty$, with $dx_{n+1}$ being some reference measure on $\alg{X}_{n+1}$.  In addition, we let $\Xinit$ be some possibly unnormalized density function on $\set{X}_0$. The transition densities $(\ld{n})_{n\in\nset}$, which are assumed to be tractable, induce multivariate probability densities
\begin{equation}\label{eq:smooth}
	\post{0:n}(x_{0:n})\propto\Xinit(x_0)\prod_{m=0}^{n-1}\ld{m}(x_m,x_{m+1}),\quad n\in\nset,
\end{equation}
where $x_{0:n} \eqdef (x_0, \dots, x_n)$ (being our generic notation for vectors) denotes an element in the Cartesian product $\set{X}_0 \times \cdots \times \set{X}_n$. The aim of the present paper is the development of \emph{sequential Monte Carlo} (SMC) \emph{methods} approximating \emph{online} (in a sense that will be specified below) the expectations 
\begin{equation}\label{eq:post}
	\post{0:n} \adds{n} \eqdef \int \adds{n}(x_{0:n})\post{0:n}(x_{0:n})\,dx_{0:n},\quad n\in \nset,
\end{equation}
for given \emph{additive state functionals} $(\adds{n})_{n \in \nset}$ such that \eqref{eq:post} is well defined. Starting with some measurable function $\adds{0}$ on $\set{X}_0$,  these functionals are defined recursively as 
\begin{equation} \label{eq:adds}
\adds{n+1}(x_{0:n+1})=\adds{n}(x_{0:n})+\addf{n}(x_{n:n+1}), 
\end{equation}
where $\addf{n}$ is some measurable function on $\set{X}_n \times \set{X}_{n+1}$. 

Our model framework, which was also considered by \citet{gloaguen:lecorff:olsson:2021}, has great generality. It covers, \eg, the \emph{Feynman--Kac models}, for which the transition densities can be decomposed as
\begin{equation}\label{eq:feynkac}
	\ld{n}(x_n,x_{n+1}) = q_n(x_n, x_{n+1}) \g{n+1}(x_{n+1}),
\end{equation}
where $\g{n+1}$ is some tractable potential function and $q_n$ some Markov transition density. These models are widely used in, \eg, statistics, physics, biology, and signal processing, and we refer to \citet{delmoral:2004} for a comprehensive treatment. Closely related to Feynman--Kac models are \emph{hidden Markov models} (HMMs), which constitute a modeling tool of significant importance in a variety of scientific and engineering disciplines \citep[see][]{cappe:moulines:ryden:2005}. A fully dominated HMM consists of a bivariate Markov chain $(X_n, Y_n)_{n \in \mathbb{N}}$ evolving on some product measurable space $(\set{X} \times \set{Y}, \alg{X} \varotimes \alg{Y})$ according to Markov transition densities in the form $q(x_n, x_{n + 1}) g(x_{n + 1},y_{n + 1})$, $(x_n, x_{n + 1}, y_{n + 1}) \in\set{X} \times \set{X} \times \set{Y}$, where $q$ and $g$ are themselves Markov transition densities (which may depend on $n$ in the general case) on $\set{X}\times\set{X}$ and $\set{X} \times \set{Y}$, respectively. The chain is initialized according to $\chi(x_0) g(x_0, y_0)$ for some density $\chi$ on $\set{X}$. In this model, only the marginal process $(Y_n)_{n \in \nset}$ is observed, whereas $(X_n)_{n \in \nset}$ is latent. The construction implies \citep[see][Section~2.2, for details]{cappe:moulines:ryden:2005} that (i) the marginal \emph{state process} $(X_n)_{n \in \nset}$ is itself a Markov chain with transition densities $q$ and that (ii) conditionally to the state process, the observations $(Y_n)_{n \in\nset}$ are independent with marginal densities given by $g(X_n, y_n)$, $n \in \nset$. Now, let $(y_n)_{n \in \nset}$ be a \emph{fixed} sequence of observations and define, for every $n \in \nset$, the transition density $\ld{n}(x_n, x_{n + 1}) =q(x_n, x_{n + 1}) \g{n + 1}(x_{n + 1})$, $(x_n, x_{n + 1})\in \set{X} \times \set{X}$ (with the dependence on $y_{n+1}$ being implicit in the notation); then, with these definitions, each density \eqref{eq:smooth} corresponds to the \emph{joint-smoothing distribution at time $n$}, \ie, the conditional density of $X_{0:n}$ given $Y_{0:n} = y_{0:n}$. In the HMM literature,  the computation of \eqref{eq:smooth} is referred as \emph{joint smoothing}, and in the absence of alternative terminology we adopt this term to the more general context considered in the present paper. Moreover, the above-described problem of computing online the expectations $(\post{0:n} \adds{n})_{n \in \nset}$ will be referred to as \emph{online additive smoothing}. 

Additive smoothing is of crucial importance in many applications in statistics and engineering. It is a key ingredient of most approaches to parameter learning in HMMs, \eg, when computing log-likelihood gradients (\emph{score functions}) via the \emph{Fisher identity} or the \emph{intermediate quantity} of the \emph{expectation-maximization} (EM) \emph{algorithm} \citep[see, \eg,][Chapters 10--11]{cappe:moulines:ryden:2005}. Scenarios of streaming data or limited computing resources call for online versions---such as the \emph{recursive maximum likelihood} \citep{legland:mevel:1997} and \emph{online EM} \citep{mongillo:deneve:2008,cappe:2009} \emph{methods}---of these approaches, which rely entirely on the possibility of computing incrementally expectations of form 
\eqref{eq:post}. 

However, as the transition densities $(\ld{n})_{n \in \nset}$ are typically complicated, the densities \eqref{eq:smooth} are known only up to normalizing constants in the general case, \ie, for models outside the classes of finite state-space models or models with a linear Gaussian structure. SMC methods---or, \emph{particle methods}---constitute a class of powerful genetic-type algorithms sampling recursively from sequences of distributions, defined on spaces of increasing dimension and known only up to normalizing constants, by means of sequential importance sampling and resampling techniques; see \citet{chopin:papaspiliopoulos:2020} for a recent introduction to this methodology and \cite{kantas:doucet:singh:chopin:2015} for a survey of its application to parameter inference in general state-space HMMs. In the following we provide an overview of the most popular approaches to SMC-based additive smoothing. Focus is entirely on \emph{online} algorithms, by which we mean algorithms such that (1) the sequence $(\post{0:n} \adds{n})_{n \in \nset}$ is approximated incrementally in a single sweep of the data and (2) the computational cost of each incremental update as well as the total storage demand is uniformly bounded in $n$. 

\subsection{Previous work}\label{sec:prevwork}
In the following all random variables are assumed to be well defined on a common probability space $(\varOmega, \mathcal{F}, \prob)$. We aim to approximate the sequence $(\post{0:n}\adds{n})_{n\in\nset}$ by propagating recursively a random sample $(\epart{0:n}{i}, \wgt{n}{i})_{i=1}^\N$ of particles (the $\epart{0:n}{i}$) and associated weights (the $\wgt{n}{i}$). Here $\N$ is the Monte Carlo sample size. For each $n$, the sample forms an empirical probability measure $\post[\N]{0:n} \eqdef \wgtsum{n}^{-1}\sum_{i=1}^{\N}\wgt{n}{i}\delta_{\epart{0:n}{i}}$, where $\wgtsum{n} \eqdef \sum_{i=1}^{\N}\wgt{n}{i}$ and $\delta_{\epart{0:n}{i}}$ is the Dirac measure located at $\epart{0:n}{i}$, which allows $\post{0:n}\adds{n}$ to be approximated by $\post[\N]{0:n}\adds{n} = \wgtsum{n}^{-1}\sum_{i=1}^{\N}\wgt{n}{i}\adds{n}(\epart{0:n}{i})$. 

Algorithm \ref{algo:sisr} describes how the particle sample is updated recursively in the \emph{auxiliary particle filter} (APF) introduced by \citet{pitt:shephard:1999} \citep[generalizing the \emph{bootstrap particle filter} proposed by][]{gordon:salmond:smith:1993} and here furnished with adaptive multinomial resampling. Using the APF requires a few algorithmic parameters to be set. The \emph{mutation step} (Line~\ref{line:mutation}) is determined by proposal transition density $\hd_n$ on $\set{X}_n \times \set{X}_{n+1}$ such that $\hd_n(x_n, \cdot)$ dominates $\ld{n}(x_n, \cdot)$ for all $x_n \in \set{X}_n$. As a part of the \emph{selection step} (Line~\ref{line:selection}), each particle weight is multiplied by some adjustment multiplier function $\am{n}$ allowing information concerning the density $\ld{n}$ to be taken into account when selecting the particles. At time zero the particle sample is initialized by standard importance sampling, \ie, by drawing independent particles $(\epart{0}{i})_{i=1}^\N$ from some proposal density $\nu$ and assigning each particle the weight $\wgt{0}{i} \eqdef \Xinit(\epart{0}{i}) / \nu(\epart{0}{i})$. Selection is absolutely essential to counteract weight degeneracy, and hence to stabilize numerically the estimator \citep[see, \eg,][Section~7.3]{cappe:moulines:ryden:2005}, but should not be applied unnecessarily; thus, we introduce a sequence of binary-valued random variables $(\res[\N]{n})_{n\in\nset}$ indicating whether resampling should be triggered or not. The sequence $(\res[\N]{n})_{n \in \nset}$ is assumed to be adapted to the filtration $(\partfiltbar{n})_{n \in \nset}$ generated by the particle filter, where $\partfiltbar{n} \eqdef \sigma((\epart{0}{i})_{i=1}^\N, (\epart{m}{i}, \I{m}{i})_{i=1}^\N, m \in \intvect{1}{n})$. Thus, these indicators may depend on the values of the importance weights, implying an adaptive resampling schedule, or, alternatively, on $n$ only, implying a deterministic ditto. In the first case, weight skewness is most commonly assessed using the \emph{effective sample size} (ESS, \citealt{liu:1996}) defined by $\ess_n \eqdef 1 / \sum_{i=1}^{\N}(\wgt{n}{i}/\wgtsum{n})^2$, which provides an estimator of the number of active particles at time $n$, taking on the values $1$ and $N$ in the cases of maximal (all the weights are equal to zero except one) and minimal (all weights are equal and non-zero) skewness, respectively. Using the ESS, one may let $\res[\N]{n}=\1{\{ \ess_n<\alpha\N\}}$, where $\alpha\in(0,1)$ is a design parameter, and this will be our primary choice.

\begin{algorithm}[htb]
	\caption{Adaptive APF.}
	\begin{algorithmic}[1]\label{algo:sisr}
		\REQUIRE $(\epart{0:n}{i},\wgt{n}{i})_{i=1}^\N$.
		\FOR{$i=1\rightarrow\N$}
			\IF {$\res[\N]{n}=1$}
				\STATE draw $\I{n+1}{i}\sim \catdist((\wgt{n}{\ell}\am{n}(\epart{n}{\ell}))_{\ell=1}^\N)$;\label{line:selection}
			\ELSE	
				\STATE set $\I{n+1}{i}\leftarrow i$;\label{line:noselection}
			\ENDIF
			\STATE draw $\epart{n+1}{i}\sim \hd_n(\epart{n}{\I{n+1}{i}},\cdot)$;\label{line:mutation}
			\STATE set $\epart{0:n+1}{i}\leftarrow(\epart{0:n}{\I{n+1}{i}},\epart{n+1}{i})$;
			\STATE weight $\wgt{n+1}{i}\leftarrow  \dfrac{\ld{n}(\epart{n}{\I{n+1}{i}}, \epart{n+1}{i})}{\hd_n(\epart{n}{\I{n+1}{i}}, \epart{n+1}{i})(\am{n}(\epart{n}{\I{n+1}{i}}))^{\res[\N]{n}}}(\wgt{n}{i})^{1-\res[\N]{n}}$;
		\ENDFOR
		\RETURN $(\epart{0:n+1}{i},\wgt{n+1}{i},\I{n+1}{i})_{i=1}^\N$.
	\end{algorithmic}
\end{algorithm}

In the case of additive functionals, $\post[\N]{0:n}\adds{n}$ can be updated incrementally and without storing the particle paths. Indeed, assuming that we have, at time $n$, computed the statistics $\tstat[i]{n} \eqdef \adds{n}(\epart{0:n}{i})$, $i \in \intvect{1}{\N}$, we can, after having executed Algorithm~\ref{algo:sisr}, easily update the same according to
\begin{equation}\label{eq:updatepoor}
	\tstat[i]{n+1}=\tstat[\I{n+1}{i}]{n}+\addf{n}(\epart{n}{\I{n+1}{i}},\epart{n+1}{i}). 
\end{equation}
The procedure is initialized by letting $\tstat[i]{0}\leftarrow\adds{0}(\epart{0}{i})$. Besides allowing for completely recursive and computationally fast updates, this technique has constant memory requirements; in order to perform \eqref{eq:updatepoor} and then compute the estimator $\post[\N]{0:n+1}\adds{n+1}=\wgtsum{n+1}^{-1}\sum_{i=1}^{\N} \wgt{n+1}{i}\tstat[i]{n+1}$, we only need access to $(\epart{n:n+1}{i}, \wgt{n+1}{i},\allowbreak \I{n+1}{i}, \tstat[i]{n})_{i=1}^\N$ rather than the whole particle paths, whose dimension increases indefinitely with time. 

Despite its ease of use and low computational requirements, the procedure described above is impractical due to the well known \emph{particle-path degeneracy phenomenon} caused by the resampling operation. More precisely, every time selection is performed, some particles will be propagated from the same parent; thus, by tracing the genealogical history of the particles we eventually encounter, assuming $ n $ is sufficiently large, a common ancestor for all the particles. In the case of multinomial resampling and under standard strong mixing assumptions on the model, \citet{koskela2020} showed that the expected number of generations back to the most recent common ancestor is $\ordo(\N)$. %, tightening the $\ordo(\N\log\N)$ estimate presented by \citet{jacob:murray:rubenthaler:2015}. 
 This result suggests that as $n$ grows, all particle paths will largely coincide, affecting greatly the reliability of the approximation and yielding a variance that grows quadratically with $n$; see, \eg, \citet{poyiadjis:doucet:singh:2011} for a discussion. An adaptive strategy based on, say, the ESS would still not prevent this particle-path depletion; in fact, such an approach is only able to defer an inevitable destiny, without ensuring stability for large $n$. In the light of these shortcomings, we will, following the terminology of \citet{douc:moulines:stoffer:2014}, refer to this approach as the \emph{poor man's smoother}. 

An alternative approach, addressing the particle-path degeneracy, is the \emph{fixed-lag smoothing} technique, proposed by \citet{kitagawa:sato:2001} and developed further by \citet{olsson:cappe:douc:moulines:2006}. The method obtains long-term stability at the cost of a bias that depends on the ergodicity properties of the model. The bias is controlled by a lag parameter, which should be neither too small, leading to significant bias, nor too large, leading to increased particle-path collapse and hence increase of the variance. Thus, designing a good lag is non-trivial in general. 

Another line of research aims to circumvent the particle-path degeneracy phenomenon using \emph{backward-sampling techniques}. Assume for a moment a Feynman--Kac model of type \eqref{eq:feynkac} and that the particle cloud is propagated using the standard bootstrap particle filter, corresponding to the parameterization $\res[\N]{n}=1$, $\am{n} \equiv1$ and $\hd_n \equiv q_n$ of Algorithm \ref{algo:sisr} \citep[see][Section 2.2, for the generalization to our setting]{gloaguen:lecorff:olsson:2021}. In this case, it is easily seen that the conditional probability $\backprob{n}{\N}(i,j)$ that $\I{n+1}{i}=j$ given $\epart{n+1}{i}$ and $(\epart{n}{\ell})_{\ell=1}^\N$, or, in other words, the probability that $\epart{n}{j}$ is the parent of $\epart{n+1}{i}$, is
\begin{equation}\label{eq:backprob}
	\backprob{n}{\N}(i,j)\propto \wgt{n}{j} q_n(\epart{n}{j},\epart{n+1}{i})\propto\wgt{n}{j}\ld{n}(\epart{n}{j},\epart{n+1}{i}).
\end{equation}
In the case of additive smoothing, \citet{delmoral:doucet:singh:2010} use the conditional backward probabilities \eqref{eq:backprob} to Rao-Blackwellize the update \eqref{eq:updatepoor}, yielding the alternative update 
\begin{equation}\label{eq:updateffbsm}
	\tstat[i]{n+1}=\sum_{j=1}^{\N}\backprob{n}{\N}(i,j)(\tstat[j]{n}+\addf{n}(\epart{n}{j},\epart{n+1}{i})).
\end{equation}
It is easily seen that this approach is nothing but a forward-only implementation of the so-called \emph{forward-filtering backward-smoothing} (FFBSm) algorithm \citep[see, \eg,][]{doucet:godsill:andrieu:2000}. 
%(The FFBSm algorithm, which is a batch-mode algorithm in its basic form, was originally derived in differently from the so-called \emph{backward decomposition} of the joint-smoothing distribution.) 
 Importantly, the Rao-Blackwellized update \eqref{eq:updateffbsm} avoids genealogical tracing and, as a consequence, the path-degeneracy problem. Still, a significant drawback of this approach is its $\ordo(\N^2)$ complexity, which is due to the fact that each update \eqref{eq:updateffbsm} involves the calculation of two sums of $\N$ terms (including the normalizing constant of $\backprob{n}{\N}(i, \cdot)$). 

In order to reduce the computational complexity of forward-only FFBSm, \citet{olsson:westerborn:2017} propose to replace the update \eqref{eq:updateffbsm} by a Monte Carlo estimate based on 
$\K \ll \N$ conditionally independent draws $(\bi{n+1}{i}{j})_{j=1}^{\K}$ from \eqref{eq:backprob}, leading to the update
\begin{equation} \label{eq:updateparis}
	\tstat[i]{n+1}=\frac{1}{\K}\sum_{j=1}^{\K}(\tstat[\bi{n+1}{i}{j}]{n}+\addf{n}(\epart{n}{\bi{n+1}{i}{j}},\epart{n+1}{i})).
\end{equation}
By adopting an accept-reject technique developed by \citet{douc:garivier:moulines:olsson:2009}, applicable whenever $q$ is uniformly bounded, the computational complexity of the resulting algorithm, referred to as the \emph{particle-based, rapid incremental smoother} (PaRIS), can be shown to be $\ordo(\K \N)$. The rejection-sampling approach was originally introduced for the \emph{forward-filtering backward-simulation} (FFBSi) \emph{algorithm} \citep{godsill:doucet:west:2004}, a batch-mode smoother that avoids the computational overload of FFBSm by means of additional simulation, and the PaRIS can in some sense be viewed as an online version of FFBSi. Importantly, \citet{olsson:westerborn:2017} establish that the PaRIS is asymptotically consistent (as $\N$ tends to infinity) and numerically stable for any fixed $\K \geq 2$, while $\K = 1$ leads to a particle-path degeneracy phenomenon reminiscent of that of the poor man's smoother. In fact, letting $\K \geq 2$ in the PaRIS yields an estimator with a linear variance growth in $n$, which is the optimal rate for a Monte Carlo approximation of additive functions on the path space, since some variance is inevitably added at each step. Even though the accept-reject approach implies an average $\ordo(\K \N)$ complexity, which is a significant improvement compared to forward-only FFBSm, backward sampling is still the computational bottleneck of the PaRIS. Indeed, in most applications the computational time of the PaRIS exceeds that of the poor man's smoother by at least one order of magnitude. 

%Finally, so far we discussed the forward-only FFBSm and PaRIS algorithms under the assumption that the model is of Feynman-Kac type. However, these algorithms were cast into the more general setting of Section~\ref{sec:background} by \citet[Section 2.2]{gloaguen:lecorff:olsson:2021}. 

\subsection{Our contribution}
In the next section we propose a novel additive smoothing algorithm which can be viewed as a golden mean between computational speed and stability. If the PaRIS may be viewed as a hybrid between the forward-only FFBSm and the FFBSi, our novel algorithm can rather be viewed as a hybrid between the adaptive poor man's smoother and the PaRIS. The main idea is to avoid, by adaptation, unnecessary selection in order to reduce the particle-path degeneracy in the poor man's smoother, while interleaving, possibly adaptively, the evolution of the particles with regular backward-sampling operations in order to repopulate, when needed, the support of the estimator. In this way we are able to keep the number of backward-sampling operations at a minimum, yielding an algorithm that is, as demonstrated by our simulations, at least one order of magnitude faster than the PaRIS, but with a fully comparable variance. Moreover, besides proving the consistency and asymptotic normality (as $N$ tends to infinity) of the estimators produced by the algorithm, we also establish the long-term numerical stability of the algorithm by showing that the asymptotic variance grows at most linearly with $n$. 

The rest of the paper is organized as follows. In Section~\ref{sec:algo} we present our novel algorithm and Section~\ref{sec:theory} is devoted to the theoretical analysis of the same. Besides benchmarking the proposed algorithm against existing online smoothers, the purpose of the simulation study in Section~\ref{sec:simul} is also to formulate guidelines on how to set its algorithmic parameters. In Section~\ref{sec:discuss} we conclude the paper.  The paper is furnished with an Appendix, Sections~\ref{sec:introsupp}--\ref{sec:propbound}, providing the proofs of the theoretical results in Section~\ref{sec:theory}, which tend to be quite technical and call for a more advanced notational machinery.

\section{A NOVEL ADAPTIVE SMOOTHER}\label{sec:algo}
In the previous section we introduced the $(\partfiltbar{n})_{n \in \nset}$-adapted sequence $(\res[\N]{n})_{n \in \nset}$ regulating the adaptive selection schedule of the APF. We now introduce another binary-valued random sequence $(\M[\N]{n})_{n\in\nset}$, where each $\M[\N]{n}$ is measurable with respect to the $\sigma$-field $\partfiltbar{n} \vee \sigma((\I{n+1}{i})_{i=1}^\N, (\res[\N]{m})_{m=0}^n)$ and such that $\M[\N]{n}=0$ whenever $\res[\N]{n}=0$. While the sequence $(\res[\N]{n})_{n \in \nset}$ determines the resampling times (corresponding to times $n$ for which $\res[\N]{n} = 1$), the sequence $(\M[\N]{n})_{n\in\nset}$ determines the times for which backward sampling is triggered ($\M[\N]{n} = 1$). By construction, the backward-sampling times form a subset of the resampling times. Loosely speaking, our approach is basically a poor man's smoother that regularly executes PaRIS-like updating steps according to the schedule determined by $(\M[\N]{n})_{n\in\nset}$. As before, the algorithm is propagating a weighted sample $(\epart{n}{i},\tstat[i]{n},\wgt{n}{i})_{i=1}^\N$ of particles and associated smoothing statistics. Whenever $\M[\N]{n}=0$, the smoothing statistics $(\tstat[i]{n})_{i=1}^\N$ are updated according to the equation \eqref{eq:updatepoor}; when instead $\M[\N]{n}=1$, implying that resampling has been applied, the statistics are updated by means of a superposition of an update \eqref{eq:updatepoor} and a PaRIS-like update. More specifically, after selection and mutation, each draw $\epart{n+1}{i}$ is linked to a randomly selected ancestor $\epart{n}{J_{n + 1}^i} $ and associated statistic $\tstat[J_{n + 1}^i]{n}$ in the previous generation, where $J_{n + 1}^i$ is drawn from $\backprob{n}{\N}(i, \cdot)$; after this, the smoothing statistic is updated according to the equation
\begin{equation}\label{eq:newupdate}
	\tstat[i]{n+1}=\frac{1}{2}\left(\tstat[\I{n+1}{i}]{n}+\addf{n}(\epart{n}{\I{n+1}{i}},\epart{n+1}{i})+\tstat[J_{n+1}^i]{n}+\addf{n}(\epart{n}{J_{n+1}^i},\epart{n+1}{i})\right).
\end{equation} 
As shown by \citet[Section~2.2]{gloaguen:lecorff:olsson:2021}, the index $J_{n + 1}^i$ can, using rejection sampling, be generated without calculation of the normalizing constant of $\backprob{n}{\N}(i, \cdot)$, at least under the mild assumption that the exists some positive function $c_n$ on $\set{X}_{n + 1}$ such that $\ld{n}(x_n, x_{n + 1}) \leq c_n(x_{n + 1})$ for all $(x_n, x_{n + 1}) \in \set{X}_n \times \set{X}_{n + 1}$. In that case, $J_{n + 1}^i$ can be simulated by generating, until acceptance, a candidate $J^*$ from $\catdist((\wgt{n}{i})_{i = 1}^\N)$ and accepting the same with probability $\ld{n}(\epart{n}{J^*}, \epart{n + 1}{i}) / c_n(\epart{n + 1}{i})$. This can be shown to yield an overall $\ordo(\N)$ computational complexity \citep[see][for details]{gloaguen:lecorff:olsson:2021,douc:garivier:moulines:olsson:2009}.

Algorithm \ref{algo:new}, which we have called \emph{AdaSmooth} to emphasize its adaptive nature, summarizes all these steps. Clearly, as AdaSmooth operates completely online, without any need of storing the full particle paths, it is enough to input the last particle components and associated weights, $(\epart{n}{i}, \wgt{n}{i})_{i=1}^\N$, into the APF (rather than the whole paths) and let it output only the updated ditto along with the associated ancestor indices, $(\epart{n + 1}{i}, \I{n+1}{i}, \wgt{n + 1}{i})_{i=1}^\N$; this operation is expressed compactly as $(\epart{n+1}{i}, \I{n+1}{i}, \wgt{n+1}{i})_{i=1}^\N\leftarrow\apf((\epart{n}{i},\wgt{n}{i})_{i=1}^\N)$ in Algorithm~\ref{algo:new}. 
\begin{algorithm}[htb]
\caption{AdaSmooth}
\begin{algorithmic}[1]\label{algo:new}
	\REQUIRE %Particle sample with associated statistics 
	$(\epart{n}{i}, \tstat[i]{n}, \wgt{n}{i})_{i=1}^\N$%.
	\STATE run $(\epart{n+1}{i}, \I{n+1}{i}, \wgt{n+1}{i})_{i=1}^\N\leftarrow\apf((\epart{n}{i},\wgt{n}{i})_{i=1}^\N)$;\label{linealgo:apf}
	\FOR{$i=1\rightarrow\N$}
		\IF{$\M[\N]{n}=1$}
			\STATE draw $J_{n+1}^{i}\sim \catdist((\wgt{n}{j}\ld{n}(\epart{n}{j},\epart{n+1}{i}))_{j=1}^\N)$;
			\STATE set $\tstat[i]{n+1}\leftarrow2^{-1}\big(\tstat[\I{n+1}{i}]{n}+\addf{n}(\epart{n}{\I{n+1}{i}},\epart{n+1}{i})+\tstat[J_{n+1}^{i}]{n}+\addf{n}(\epart{n}{J_{n+1}^{i}},\epart{n+1}{i})\big)$;
		\ELSE
			\STATE set $\tstat[i]{n+1}\leftarrow\tstat[\I{n+1}{i}]{n}+\addf{n}(\epart{n}{\I{n+1}{i}},\epart{n+1}{i})$;\label{line:forw}
		\ENDIF
	\ENDFOR
	\RETURN $(\epart{n+1}{i}, \tstat[i]{n+1}, \wgt{n+1}{i})_{i=1}^\N$.
\end{algorithmic}
\end{algorithm}
As explained above, backward sampling is used in Algorithm \ref{algo:new} as a means of guaranteeing the stochastic stability of the resulting estimators, and here the sequences $(\res[\N]{n})_{n \in \nset}$ and $(\M[\N]{n})_{n \in \nset}$ play a critical role; in Section~\ref{sec:theory} we will discuss the convergence (as $\N$ increases) and stability properties of Algorithm~\ref{algo:new}, by starting to analyze the case where these sequences are deterministic and then extending the analysis to adaptive policies. 

As we mentioned in Section~\ref{sec:intro}, a common approach is to let $\res[\N]{n}=\1{\{\ess_n<\alpha\N\}}$ for all $n \in \nset$, \ie, to resample only when the ESS, estimating of the number of active particles, falls below a given threshold $\alpha \N$ for some prescribed $\alpha\in (0,1)$. Similarly, the sequence $(\M[\N]{n})_{n \in \nset}$ regulating the backward-sampling schedule should be based on some criterion assessing the degeneracy of the particle paths. Since backward sampling is expensive, our goal is to allow $\M[\N]{n}$ to be zero as often as possible without jeopardizing the stability of the estimator. One way to do this is to monitor the number of distinct trajectories by keeping track of the ancestors of the current particles $(\epart{n}{i})_{i=1}^\N$ at the last time point $ n_0< n $ for which $ \M[\N]{n_0}=1 $. At time $n_0$ the trajectories were recombined through the updating rule \eqref{eq:newupdate} into rejuvenated statistics $(\tstat[i]{n_0+1})_{i=1}^\N$. Thus, even if backward sampling affects only the smoothing statistics and not the underlying particle system in AdaSmooth, we may forget about the particles' history before $n_0$ and imagine that a new particle genealogy is started at time $n_0+1$. We may then re-proceed without backward sampling until the number of distinct ancestors at time $n_0+1$ is too small. More precisely, whenever this number falls below some given threshold, we set $\M[\N]{n}=1$ and let the current particles be the ancestors of a new genealogy; otherwise we set $\M[\N]{n}=0$. 
In order to keep track of the ancestors at time $n_0 + 1$, we make use of the \emph{Enoch indices}  
\citep[the concept is borrowed from][]{olsson:douc:2017} at the same time point, defined recursively through 
\begin{equation}
	\enoch{n}{i}\eqdef\begin{cases}
		i\quad &\text{for }n=n_0 + 1,\\
		\enoch{n - 1}{\I{n}{i}}&\text{for }n > n_0 + 1,
	\end{cases}
\end{equation} 
for $i \in \intvect{1}{\N}$. With this definition, $\enoch{n}{i}$ is the index of the time $n_0 + 1$ ancestor of the particle $\epart{n}{i}$. Using the Enoch indices, a new genealogy is initialized by letting $\enoch{0}{i}=i$ for all $i\in\intvect{1}{\N}$; after this, the indices are updated recursively according to $\enoch{n+1}{i}=\enoch{n}{\I{n+1}{i}}$, and once the number of distinct elements among $(\enoch{n+1}{i})_{i=1}^\N$ falls below a threshold $\beta\N$, for some prescribed $\beta\in(0,1)$, we set $\M[\N]{n}=1$ and reinitialize $\enoch{n+1}{i}=i$ for all $i\in\intvect{1}{\N}$. We summarize this adaptive policy for determining the sequence $(\M[\N]{n})_{n \in \nset}$ in Algorithm \ref{algo:backdraws}. The parameter $\beta$ determines the fraction of distinct Enoch indices below which we decide to activate backward sampling. Clearly, Algorithm~\ref{algo:backdraws} is not a stand-alone routine and has to be embedded in Algorithm~\ref{algo:new}, immediately after Line~\ref{linealgo:apf}. Having established also a criterion handling path degeneracy, we have now obtained a fully adaptive version of Algorithm~\ref{algo:new}.

\begin{algorithm}
	\caption{Generation of adaptive backward-sampling schedule $(\M[\N]{n})_{n \in \nset}$.}
	\begin{algorithmic}[1]\label{algo:backdraws}
		\REQUIRE 
		$(\enoch{n}{i})_{i=1}^\N$, $(\I{n+1}{i})_{i=1}^\N$, $\beta\in(0,1)$
		\STATE set $\enoch{n+1}{i}\leftarrow\enoch{n}{\I{n+1}{i}}$ \textbf{for} $i=1\rightarrow\N$;
		\IF{$\res[\N]{n}=1$ and $\lvert(\enoch{n+1}{i})_{i=1}^\N\rvert<\beta\N$}
			\STATE set $\M[\N]{n}\leftarrow1$;
			\STATE set $\enoch{n+1}{i}\leftarrow i$ \textbf{for} $i=1\rightarrow\N$;
		\ELSE
			\STATE set $\M[\N]{n}\leftarrow0$;
		\ENDIF
		\RETURN $(\enoch{n+1}{i})_{i=1}^\N$, $\M[\N]{n}$
	\end{algorithmic}
\end{algorithm}

\section{THEORETICAL RESULTS}\label{sec:theory}
\subsection{Deterministic selection and backward-sampling schedules}\label{sec:detsched}
Our initial analysis of Algorithm~\ref{algo:new} will be conducted under the assumption that the selection and backward-sampling schedule is deterministic.
\begin{assumption}\label{assum:seqdet}
	For all $n\in\nset$, $\res[\N]{n}=\res{n}$ and $ \M[\N]{n}=\M{n} $, where the sequences $(\res{n})_{n \in \nset}$ and $(\M{n})_{n \in \nset}$ are deterministic and such that $\M{n}=0$ whenever $\res{n}=0$.
\end{assumption}
In this setting we establish two results: the almost-sure convergence (Theorem~\ref{thm:as}) of the estimator $\wgtsum{n}^{-1}\sum_{i=1}^{\N}\wgt{n}{i}\tstat[i]{n}$, where $(\tstat[i]{n}, \wgt{n}{i})_{i=1}^\N$ is produced by $n$ steps of Algorithm~\ref{algo:new}, as well as a central limit theorem (Theorem~\ref{thm:clt}), whose asymptotic variance is subject to further investigation regarding the stochastic stability of the algorithm. Proofs are found in the Appendix. For every $n\in\nset$ we define the weight function
\begin{equation}
	w_n\langle\res{n}\rangle:\set{X}_n\times\set{X}_{n+1}\ni (x,x')\mapsto \frac{\ld{n}(x,x')}{(\am{n}(x))^{\res{n}}\hd_n(x,x')}.
\end{equation}
In addition, we set
\begin{equation}
	w_{-1}:\set{X}_0\ni x\mapsto \frac{\Xinit(x)}{\nu(x)}.
\end{equation}
\begin{assumption}\label{assum:1}
	For all $(\rho_n)_{n \in \nset}$ the weight functions $(w_n\langle\res{n}\rangle)_{n\in \nset}$ and $w_{-1}$ are bounded. So are also the auxiliary weight functions $(\am{n})_{n\in \nset}$. % and measurable.
\end{assumption}
In the following we define, for every $n \in \nset$, $\set{H}_n$ as the set of additive functionals $\adds{n}$ in the form \eqref{eq:adds} with bounded terms. In addition, we let $\res{0:n-1} = (\res{0},\dots,\res{n-1})$ and $\M{0:n-1} = (\M{0},\dots,\M{n-1})$. 
\begin{theorem}[strong consistency]\label{thm:as}
	Let Assumptions \ref{assum:seqdet} and \ref{assum:1} hold. Then for every $n\in\nset$ and $\adds{n}\in\set{H}_n$,
	\begin{equation}
		\lim_{\N\rightarrow\infty}\sum_{i=1}^{\N}\frac{\wgt{n}{i}}{\wgtsum{n}}\tstat[i]{n}=\post{0:n}\adds{n},\quad\mbox{$\prob$-a.s.}
	\end{equation}
\end{theorem}

\begin{theorem}[asymptotic normality]\label{thm:clt}
	Let Assumptions \ref{assum:seqdet} and \ref{assum:1} hold. Then for every $n\in\nset$ there exists a positive functional $\sigma_n \langle\res{0:n-1},\M{0:n-1}\rangle$ on $\set{H}_n$ such that for every $\adds{n}\in\set{H}_n$, as $\N \rightarrow\infty$,
	\begin{equation}
		\sqrt{\N}\left(\sum_{i=1}^{\N}\frac{\wgt{n}{i}}{\wgtsum{n}}\tstat[i]{n}-\post{0:n}\adds{n}\right)\convd \sigma_n\langle\res{0:n-1},\M{0:n-1}\rangle (\adds{n}) Z,
	\end{equation}
	where $Z$ has standard Gaussian distribution.
\end{theorem}
The almost sure convergence established by Theorem~\ref{thm:as} is in fact a direct consequence of a stronger result in the form of a \emph{Hoeffding-type exponential concentration inequality} for finite sample sizes $\N$; see Section~\ref{sec:proof:3.1} for details. An explicit expression of the asymptotic variance $\sigma_n^2(\adds{n})\langle\res{0:n-1},\M{0:n-1}\rangle$ of Theorem~\ref{thm:clt} is provided in Section~\ref{sec:proof:3.2}. Next, we establish, again under Assumption~\ref{assum:seqdet}, the stochastic stability of Algorithm~\ref{algo:new} by bounding $(\sigma^2(h_n)\langle\res{0:n-1},\M{0:n-1}\rangle / n)_{n \in \nset}$ uniformly in $n$. Again, the proof is provided in the Appendix, Section~\ref{sec:proof:3.3}. The analysis proceeds in two steps, where we in the first step analyse the algorithm in the case of systematic selection at each time point, and then, in the second step, extend these results to non-systematic, but still deterministic, selection schedules using an auxiliary model extension. In the first step, our proofs build upon recent works on the PaRIS by \citet{olsson:westerborn:2017} and \citet{gloaguen:lecorff:olsson:2021}; however, the fact that the updating rule \eqref{eq:newupdate} combines forward as well as backward indices induces a complex dependence structure that makes the adaptation highly non-trivial.  

The following assumption is used to control the stochastic stability of the marginal particle approximations produced by the APF, by bounding uniformly the distance between any two consecutive resampling times.
\begin{assumption} \label{assum:dist}
	There exists $\maxd\in \nsetpos$ such that for all $n\in\nset$, $\min \{ k\in\nsetpos : \res{n+k}=1\} \le \maxd$, \ie, the distance between two resampling times is always less than or equal to $\maxd$.
\end{assumption}
For any bounded measurable function $h$, let $\supn{h}$ denote the supnorm of $h$. Our stability analysis will be carried through under the following---now classical---\emph{strong mixing assumption}, which typically require the state spaces to be compact sets \citep[see \eg][Section~4]{delmoral:2004}.  
\begin{assumption}\label{assum:compact}
	There exist constants $0<\hklow<\hkup<\infty$ such that for every $n\in \nset $ and $(x,x')\in \set{X}_n\times \set{X}_{n+1}$, $\hklow \le \ld{n}(x,x')\le \hkup$. Moreover, there exist positive constants $\mdup$ and $\amup$ such that for all $n \in \nset$ and $\rho \in \{0, 1\}$, $\supn{w_n\langle\rho \rangle} \le \mdup$ and $\supn{\am{n}} \le \amup$. In addition, $\supn{w_{-1}}\le \mdup$. 
\end{assumption}

For every $n\in\nsetpos$ and $j\in\nset$, we define $r_n \eqdef \sum_{m=0}^{n-1}\res{m}$, \ie, the number of selection operations before time $n$, and $n_j \eqdef \min\{n\in\nset : r_{n+1}=j+1\}$, the time of the $(j + 1)$th selection operation. 

\begin{theorem}\label{thm:bound}
	Let Assumptions~\ref{assum:seqdet}, \ref{assum:dist} and \ref{assum:compact} hold. Then there exist positive constants $C_1$ and $C_2$, both depending on $\hklow,\hkup,\mdup$ and $\maxd$, such that for all additive functionals $(\adds{n})_{n \in \nset}$ in the form \eqref{eq:adds} for which there exists $\hbd > 0$ such that for all $n \in \nset$, $\supn{\addf{n}}\le \hbd$ and $\supn{\adds{0}+\addf{0}}\le \hbd$, 
	\begin{multline} 
	\limsup_{n\rightarrow\infty}\frac{1}{n}\sigma_n^2\langle\res{0:n-1},\M{0:n-1}\rangle(\adds{n}) \\
	\le \maxd^2\hbd^2\amup\left( C_1+C_2\lim_{n\rightarrow\infty}\frac{1}{r_n}\sum_{m=0}^{r_n-1}\sum_{\ell = 0}^{m}\prod_{j=\ell}^{m}(1+\M{n_j})^{-1}\right)\label{eq:asvarbound}.
	\end{multline}
\end{theorem}
As discussed above, our aim is to establish the stability of Algorithm \ref{algo:new} by bounding the right-hand side of \eqref{eq:asvarbound} uniformly in $n$. However, such a bound is not possible for all sequences $ (\M{n_j})_{j \in \nset} $; indeed, in the case where $\M{n_j} = 0$ for all $j \in \nset$ it holds that
\begin{equation}
	\lim_{n\rightarrow\infty}\frac{1}{r_n}\sum_{m=0}^{r_n-1}\sum_{\ell = 0}^{m}\prod_{j=\ell}^{m}(1+\M{n_j})^{-1}=\lim_{n\rightarrow\infty}\frac{1}{r_n}\frac{r_n(r_n+1)}{2}=\infty,
\end{equation}
which is not surprising since Algorithm \ref{algo:new} coincides with the poor man's smoother (with adaptive resampling) when the backward simulation mechanism is de-activated. Still, as established by the following theorem, a regular backward sampling schedule is sufficient to obtain a linearly increasing asymptotic variance. We define $\dd_j \eqdef \min\{ k \in\nsetpos : \M{n_{k + j}} = 1 \}$, $j \in \nset\cup\{-1\}$, which corresponds to the distance, in terms of the number of selection operations, between any selection time $n_j$ and the first subsequent backward-sampling time. If these distances are uniformly bounded, then we may obtain the desired linear bound.
\begin{proposition}\label{prop:delta}
	Assume that there exists $\dd \in \mathbb{N}^\ast$ such that $\dd_j \leq \dd$ for all $j \ge -1$. Then
	\begin{equation}\label{eq:limdelta}
		\lim_{n\rightarrow\infty}\frac{1}{r_n}\sum_{m=0}^{r_n-1}\sum_{\ell = 0}^m \prod_{j=\ell}^{m}(1+\M{n_j})^{-1} \leq \frac{3\dd-1}{2}, 
	\end{equation}
	with equality if $\dd_j = \dd$ for all $j \ge -1$. 
\end{proposition}
The proof of Proposition~\ref{prop:delta} is given in the Appendix, Section~\ref{sec:propbound}. 

\subsection{Adaptive selection and backward-sampling schedules}
Next we will show that the central limit theorem in Theorem \ref{thm:clt} can be extended to the case where the selection schedule is random and adapted to the values of the ESS. In order to guarantee the stability of the algorithm, we will still assume that selection is performed at least every $\maxd\in\nsetpos$ steps; however, this assumption can be relaxed in practice. 

\begin{assumption}\label{assum:adaptiveRes}
For given $\alpha\in(0,1)$ and $\maxd\in \nsetpos$, let $ (\res[\N]{n})_{n\in\nset} $ be defined recursively as 
\begin{equation}\label{eq:rhoN}
	\res[\N]{0} \eqdef \1{\{\ess_0<\alpha\N\}}\text{ and }\res[\N]{n+1} \eqdef 1-\1{\{\ess_{n+1}\ge\alpha\N\}}\1{\{\maxd_n^\N+1<\maxd\}},\quad n\in\nset,
\end{equation}
with $(\maxd_{n}^\N)_{n\in\nset}$ being also recursively defined through
\begin{equation}
	\maxd_0^\N \eqdef 1-\res[\N]{0}\text{ and }\maxd_{n+1}^\N \eqdef (1-\res[\N]{n+1})(1+\maxd_n^\N),\quad n\in\nset.
\end{equation}
\end{assumption}
Note that $\maxd_n^\N$ counts the number of consecutive times, including $n$, for which resampling has not been performed. The following lemma is proven in Section~\ref{sec:proof:3.5} of the Appendix.
\begin{lemma}\label{lemma:ess}
	Let Assumption \ref{assum:adaptiveRes} hold. Then for all $n\in\nset$  there exists $\res[\alpha,\maxd]{n}\in\{0,1\}$ such that, as $\N\to\infty$,
	\begin{equation}
		\res[\N]{n}\convp \res[\alpha,\maxd]{n}.
	\end{equation}
\end{lemma}
\begin{assumption}\label{assum:adaptedBS}
	For every $n\in\nset$, $\M[\N]{n}$ is $\sigma(\res[\N]{0:n})$-measurable and such that $\M[\N]{n}=0$ whenever $\res[\N]{n}=0$. 
\end{assumption}
For instance, a simple rule covered by Assumption \ref{assum:adaptedBS} is to trigger backward sampling after a fixed, deterministic number of intermediate resampling operations.

\begin{lemma}\label{lemma:BSconv}
	Let Assumptions \ref{assum:adaptiveRes} and \ref{assum:adaptedBS} hold. Then for all $n\in\nset$  there exists $\M[{\alpha,\maxd}]{n}\in\{0,1\}$ such that, as $\N\rightarrow\infty$, 
	\begin{equation}
		\M{n}^\N\convp \M[\alpha,\maxd]{n}. 
	\end{equation}
\end{lemma}
\begin{proof}
	By Lemma \ref{lemma:ess}, $\res[\N]{0:n}\convp \res[\alpha,\maxd]{0:n}$ as $\N\rightarrow\infty$. Since $\M[\N]{n}$ is $\sigma(\res[\N]{0:n})$-measurable there exists a measurable function $f_n:\{0,1\}^{n+1}\mapsto\{0,1\}$ such that $\M[\N]{n}=f_n(\res[\N]{0:n})$. Now, let $\M[\alpha,\maxd]{n} \eqdef f_n(\res[\alpha,\maxd]{0:n})$. Thus, $\{\M[\N]{n}\ne\M[\alpha,\maxd]{n}\}=\{f_n(\res[\N]{0:n})\ne f_n(\res[\alpha,\maxd]{0:n})\}\subset \{\res[\N]{0:n}\ne \res[\alpha,\maxd]{0:n}\}$, implying that for every $\epsilon>0$, 
	\begin{align}
		\prob(|\M[\N]{n}- \M[\alpha,\maxd]{n}|\ge\epsilon)= \prob(\M[\N]{n}\ne\M[\alpha,\maxd]{n})\le \prob(\res[\N]{0:n}\ne \res[\alpha,\maxd]{0:n})\rightarrow 0,
	\end{align}
	as $\N\rightarrow\infty$. Hence $\M[\N]{n}\convp \M[\alpha,\maxd]{n}$.
\end{proof}

\begin{corollary}\label{cor:cltadapt}
	Let Assumptions \ref{assum:1}, \ref{assum:adaptiveRes} and \ref{assum:adaptedBS} hold and $(\wgt{n}{i},\tstat[i]{n})_{i=1}^\N$ be generated by $n$ iterations of Algorithm \ref{algo:new}. Then for every $n\in \nsetpos$ and $\adds{n}\in\set{H}_n$, as $\N \rightarrow\infty$,
	\begin{equation}
		\sqrt{\N}\left(\sum_{i=1}^{\N}\frac{\wgt{n}{i}}{\wgtsum{n}}\tstat[i]{n}-\post{0:n}\adds{n}\right)\convd \sigma_n\langle \res[\alpha,\maxd]{0:n-1},\M[\alpha,\maxd]{0:n-1}\rangle (\adds{n}) Z,
	\end{equation}
	where $Z$ has standard Gaussian distribution and $\sigma_n^2\langle \res[\alpha,\maxd]{0:n-1},\M[\alpha,\maxd]{0:n-1}\rangle(\adds{n})$ is the asymptotic variance of Theorem \ref{thm:clt} with selection schedule $\res[\alpha,\maxd]{0:n-1}$ and backward-sampling schedule $\M[\alpha,\maxd]{0:n-1}$ given by Lemmas \ref{lemma:ess} and \ref{lemma:BSconv}, respectively.
\end{corollary}
\begin{proof}
	Let $\set{S}_n$ be the set of sequences $(\res{0:n-1},\M{0:n-1})\in\{0,1\}^{2n}$ satisfying Assumption \ref{assum:dist} and being such that $\M{m}=0$ whenever $ \res{m}=0 $ for any $m\in\intvect{0}{n-1}$. For all $(\res{0:n-1}, \M{0:n-1}) \in \set{S}_n$, let $\adds{n}^\N \langle \res{0:n-1},\M{0:n-1}\rangle \eqdef \wgtsum{n}^{-1}\sum_{i=1}^{n}\wgt{n}{i}\tstat[i]{n}$ be independent estimators calculated on the basis of independent realizations $(\tstat[i]{n}, \wgt{n}{i})_{i=1}^\N$ of Algorithm~\ref{algo:new}, each realization governed by a distinct selection and backward-sampling schedule $(\res{0:n-1}, \M{0:n-1})$. Then for every $\N \in \nset^\ast$, by the law of total probability,
	\begin{multline}
		\sqrt{\N}\left(\adds{n}^\N \langle \res[\N]{0:n-1},\M[\N]{0:n-1}\rangle-\post{0:n}\adds{n}\right) \\ 
		\eqdist\sum_{(\res{0:n-1},\M{0:n-1})\in\set{S}_n}\sqrt{\N}\left(\adds{n}^\N \langle \res{0:n-1},\M{0:n-1}\rangle-\post{0:n}\adds{n}\right)\prod_{m=0}^{n-1}\1{\{\res[\N]{m}=\res{m}\}}\1{\{\M{m}^\N=\M{m}\}}.\label{eq:sumclt}
	\end{multline}
	Now, note that for all $m\in \intvect{0}{n-1}$ it holds, by Lemmas~\ref{lemma:ess} and \ref{lemma:BSconv},
	\begin{equation}
		\prod_{m=0}^{n-1}\1{\{\res[\N]{m}=\res{m}\}}\1{\{\M{m}^\N=\M{m}\}}\convp \1{\{(\res[\alpha,\maxd]{0:n-1},\M[\alpha,\maxd]{0:n-1})=(\res{0:n-1},\M{0:n-1})\}}.
	\end{equation}
	By Slutsky's lemma and Theorem \ref{thm:clt}, all terms of \eqref{eq:sumclt} tend to zero in probability except one which converges in distribution to $\sigma_n\langle \res[\alpha,\maxd]{0:n-1},\M[\alpha,\maxd]{0:n-1}\rangle (\adds{n}) Z$. This completes the proof. 
\end{proof}
Note that an immediate consequence of Corollary \ref{cor:cltadapt} is that $\wgtsum{n}^{-1}\sum_{i=1}^{n}\wgt{n}{i}\tstat[i]{n}\allowbreak\convp\post{0:n}\adds{n}$, as $\N\rightarrow\infty$. The stochastic stability of the adaptive algorithm depends on the asymptotic variance $\sigma_n^2\langle \res[\alpha,\maxd]{0:n-1},\M[\alpha,\maxd]{0:n-1}\rangle(\adds{n})$, more specifically on the limit sequence $(\M[\alpha,\maxd]{n})_{n\in\nset}$. Proposition~\ref{prop:delta} guarantees a linear growth of the variance with respect to $n$ for any adaptation schedule that allows the number of selection operations between each backward-sampling operation to be uniformly bounded; for such schedules, this property will be transferred to the limit schedule, providing an $\ordo(n)$ bound on the asymptotic variance. 

A thorough analysis of the setting where also the backward-sampling mechanism is activated adaptively using the technology described in Algorithm \ref{algo:backdraws} is beyond the scope of the present paper. Instead, we limit ourselves to justifying heuristically that triggering, as in Algorithm \ref{algo:backdraws}, backward sampling only when the proportion of distinct Enoch indices falls below a given threshold leads, in accordance with Proposition~\ref{prop:delta}, to regular distances on average between the times of backward sampling. Under strong mixing assumptions similar to Assumption~\ref{assum:compact}, \citet{koskela2020} derive an $\ordo(\N)$ bound on the expected time to the \emph{most recent common ancestor} (MRCA) in the case where multinomial resampling is executed systematically at every time step. Using the notation of the mentioned paper, let $\tau_N(T_{n'})$ be the number of SMC generations required to reach back to the MRCA for a subsample of $n'\le\N$ particles. Here $T_{n'}$ represents the continuous time required to the reach the MRCA for a partition of $n'$ elements in the \emph{Kingman's $n$-coalescent model}, while $\tau_N$ applies a rescaling providing the same corresponding coalescing time, in terms of generations, in the genealogy of an SMC particle cloud of $\N$ samples. Then in Corollary~2 in the same paper, it is shown that the expectation of $\tau_N(T_{n'})$ is $\ordo(\N)$, uniformly in time. In Kingman's $n$-coalescent model there is an initial partition of size $n'$ in which any two elements merge into one after an exponentially distributed random time with unit rate. Thus, the partition reduces to $n'-1$ elements after an exponentially distributed random time with rate $n'(n'-1)/2$; then to $n'-2$ elements with rate $(n'-1)(n'-2)/2$; and so on. We may hence write $T_{n'} = \sum_{k=2}^{n'}S_k$, where $S_2,\dots,S_n$ are independent and $S_k$ is exponentially distributed with rate $k(k-1)/2$. 
We are now interested in the number of SMC generations required to reach, starting with the full sample, \ie, $n'=\N$, the generation corresponding to the most recent time of only $\beta\N$ distinct ancestors (assuming $\beta\N$ integer for simplicity), instead of one as for the MRCA. Thus, we denote $T_\N^\beta \eqdef \sum_{k=\beta\N+1}^{\N} S_k$, and by a straightforward adaptation of the proof of Corollary~2 in \citet{koskela2020} we may establish that the expectation of $\tau_N(T_\N^\beta)$ is uniformly bounded in $\N$. Thus, assuming systematic resampling, \ie, $\res[\N]{n}=1$ for all $n\in\nset$, and using the adaptive criterion of Algorithm \ref{algo:backdraws}, this suggests that the distance between two subsequent backward sampling steps will be regular on the average and also independent of the sample size $\N$. In the general case where resampling is not applied systematically, we will count the distance between two backward sampling operations in terms of the number of intermediate resampling operations. As we will se in the next section, our simulations indicate that this number stays close to regular and constant on average with respect to the sample size $\N$.

\section{NUMERICAL RESULTS}\label{sec:simul}
We demonstrate numerically our algorithm on two different state-space models: a \emph{linear Gaussian HMM} and a \emph{stochastic volatility model} with correlated noise.

\subsection{Linear Gaussian HMM}\label{subsec:lgss}
We first consider a linear Gaussian HMM on $\rset$, described by the equations
\[
	\begin{split}
		X_{n+1} &= a X_n + \sigma_U U_{n+1}, \\
		Y_{n} &= b X_n + \sigma_V V_n,
	\end{split}
	\quad n \in \nset \eqsp,
\]
where $(U_n)_{n \in \nsetpos}$ and $(V_n)_{n \in \nset}$ are independent sequences of mutually independent standard normally distributed noise variables and $(a, b) \in \rset^2$ and $(\sigma_U, \sigma_V)\allowbreak \in \rsetpos$ are model parameters. If $|a| < 1$, the unobserved state process $(X_n)_{n \in \nset}$ has a stationary distribution given by the zero-mean Gaussian distribution with variance $\sigma^2_U/(1-a^2)$, according to which $X_0$ is initialized. In this section, focus is set on the problem of estimating the expectation of the state sum $\adds{n}(x_{0:n}) = \sum_{m=0}^n x_m$ under the joint-smoothing distribution $\post{0:n}$ on the basis of $n + 1 = 501$ observations generated by simulation of the model parameterized by $(a, b, \sigma_U, \sigma_V) = (0.7, 1, 0.2, 1)$. The main reason for considering a linear Gaussian HMM and this particular state functional is that these allow exact solutions to the additive smoothing problem to be calculated using \emph{disturbance smoothing} \citep[see, \eg,][Section~5.2]{cappe:moulines:ryden:2005}. Having access to the exact solution, we may study the convergence and accuracy of AdaSmooth and benchmark the same against existing algorithms. 
 
In order to investigate how the choices of $\alpha$ and $\beta$, governing the adaptation criteria for resampling through the ESS and the backward-sampling through Algorithm~\ref{algo:backdraws}, respectively, affect the performance of the algorithm, we run the algorithm for varying combinations of $(\alpha, \beta) \in [0,1]^2$ and $\N \in \{50,100,200,500\}$. For each combination, we replicated 100 independent estimates of $\post{0:n} \adds{n}$ for each of the AdaSmooth, PaRIS and forward-only FFBSm algorithms, all running with the same number $\N$ of particles. For simplicity, the underlying particles were mutated according to the dynamics of the state process and selected without any adjustment of the particle weights. 

In our comparison we first evaluate the \emph{efficiency} of each algorithm, which we define as the ratio of inverse sample variance to computational time, scaled further by $1 /  \sqrt{\N}$; \ie, \texttt{$\text{efficiency} = 1/ (\sqrt{\N} \times \text{sample variance} \times \text{computational}$ $ \text{time}$}). Figure~\ref{fig:lg:eff} displays efficiencies of AdaSmooth for different combinations of the algorithmic parameters $\alpha$  and $\beta$, and in these plots there is clearly a region in the parameter space, with $\alpha$ and $\beta$ being around 0.6 and 0.5, respectively, varying slightly with $\N$, for which the efficiency is maximal. For a comparison, Table~\ref{tab:lg:eff} shows the efficiencies also of the PaRIS and the forward-only FFBSm algorithms, which are outperformed by AdaSmooth by about one and two orders of magnitude, respectively.

\begin{figure}[htb]
	\centering
	\includegraphics[width =.95\textwidth]{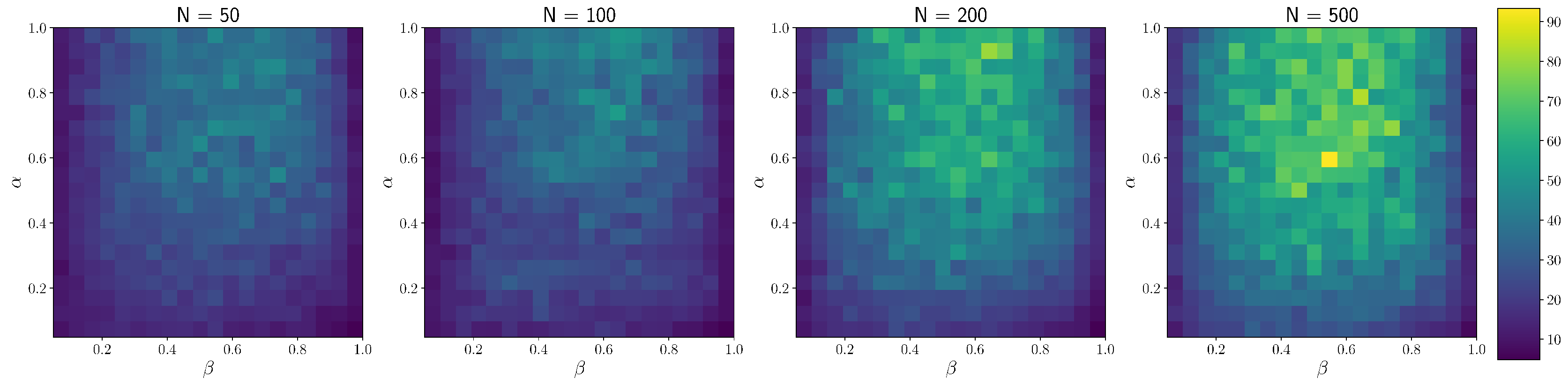}
	\caption{Efficiencies of AdaSmooth operating on the linear Gaussian HMM with different combinations of $(\alpha,\beta) \in [0,1]^2$ and $N = \{50, 100, 200, 500\}$. Each estimate is based on 100 replicates.}
	\label{fig:lg:eff}
\end{figure}

\begin{table}[htb]
	\centering
	\resizebox{8cm}{!}{\begin{tabular}{l|c|c|c|c}\toprule
		$\N $ & 50 & 100 & 200 & 500 \\ \midrule
		FFBSm & $0.86$ & $0.45$ & $0.33$ & $0.16$ \\
		PaRIS & $5.09$ & $5.62$ & $6.08$ & $5.86$ \\
		AdaSmooth $ (0.6, 0.5) $ & $31.17$ & $38.01$ & $58.32$ & $68.92$ \\\bottomrule
	\end{tabular}}
	\caption{Efficiencies of the forward-only FFBSm, the PaRIS and AdaSmooth parameterized by $(\alpha, \beta) = (0.6, 0.5)$, 
	operating on the linear Gaussian model in Section~\ref{subsec:lgss} with different sample sizes $\N$.}
	\label{tab:lg:eff}
\end{table}

Next, we illustrate that the output of AdaSmooth converges, as $\N$ increases and for any combination of $\alpha$ and $\beta$, to the exact solution provided by the disturbance smoother and compare the same to the outputs of the competitors. Figure~\ref{fig:boxplotLG} displays boxplots of independent estimates obtained with AdaSmooth for a selection of parameterizations as well as with the forward-only FFBSm, the PaRIS and the poor man's smoother for an observation record comprising $n + 1 = 1001$ observations. The figure also displays exact solutions provided by the disturbance smoother. Each box is based on 100 replicates and for each particle sample size $\N \in \{50, 500\}$, estimates of $\post{0:n} \adds{n}$ for $n \in \{100, 500, 1000\}$ are reported. The estimates are divided by $\sqrt{n}$ with the purpose of illustrating the different smoothers' stability properties as $n$ increases. In short: completely in line with the theoretical results obtained in Section~\ref{sec:theory}, the range of the boxes decrease with $\N$ and stay, with exception of the poor man's smoother, close to constant in $n$. Interestingly, the algorithm parameterized by $(\alpha, \beta)=(1, 0.1)$ (corresponding to systematic selection and infrequent backward sampling), exhibiting the largest variance among the AdaSmooth estimators, still does not not show the quadratic variance growth of the poor man's smoother. This is even clearer from Figure~\ref{fig:varianceLGSS}, displaying time-normalized variances, where all the algorithms except the poor man's smoother (whose variance growth is quadratic) present a linear increase of the variance, although with different rates. Finally, as clear from Figures~\ref{fig:boxplotLG} and \ref{fig:varianceLGSS}, the accuracy of AdaSmooth with $(\alpha, \beta)=(0.6, 0.5)$ is on par with that of the PaRIS and the forward-only FFBSm, despite the drastic improvement in terms of computational speed. 

\begin{figure}[htb]
	\centering
	\includegraphics[width=0.95\textwidth]{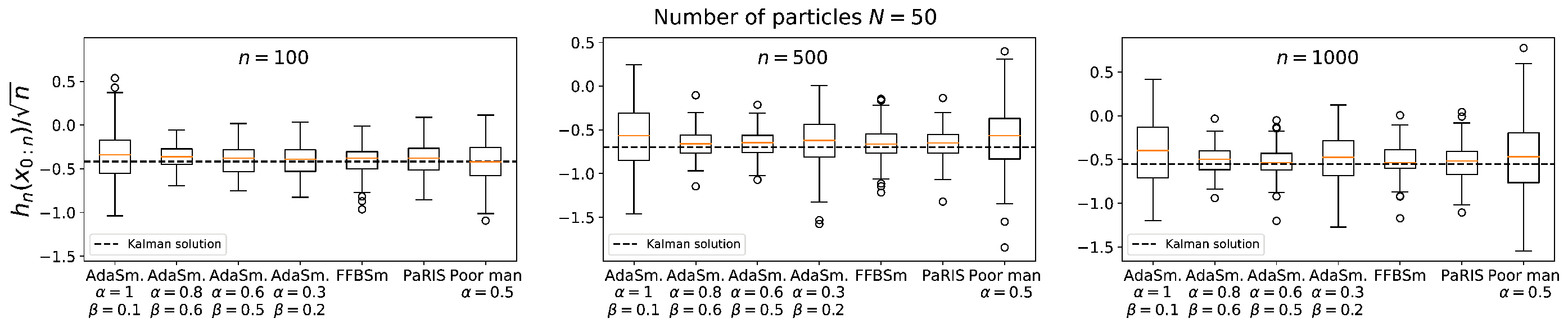}
	\includegraphics[width=0.95\textwidth]{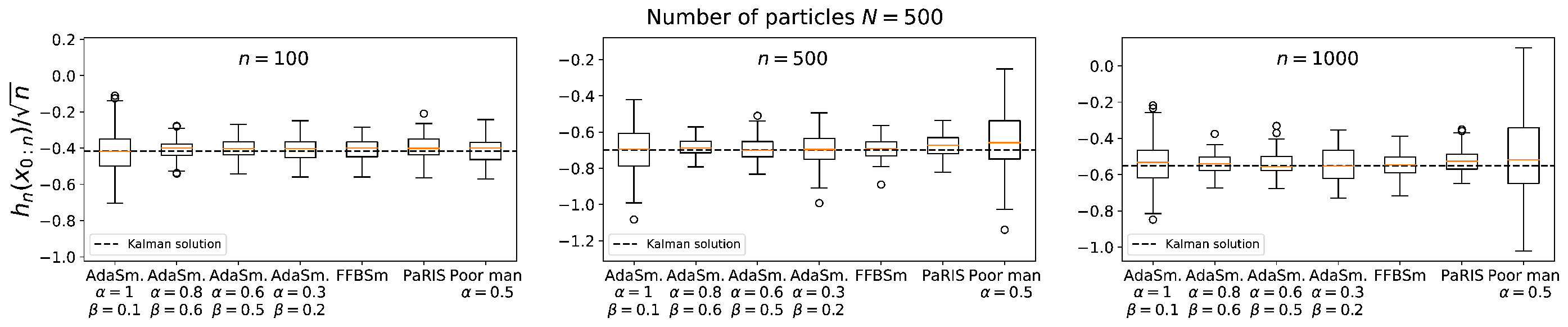}
	\caption{Boxplots of estimates (divided by $\sqrt{n}$) of smoothed expectations of $\adds{n}(x_{0:n}) = \sum_{m=0}^n x_m$, for $n \in \{100,500,1000\}$, in the linear Gaussian HMM in Section~\ref{subsec:lgss}. Each algorithm was rerun 100 times with $\N=50$ and $\N=500$. The black-dashed lines represent exact solutions provided by the disturbance smoother.  
	}
	\label{fig:boxplotLG}
\end{figure}

\begin{figure}[htb]
	\centering
	\includegraphics[width=0.85\textwidth]{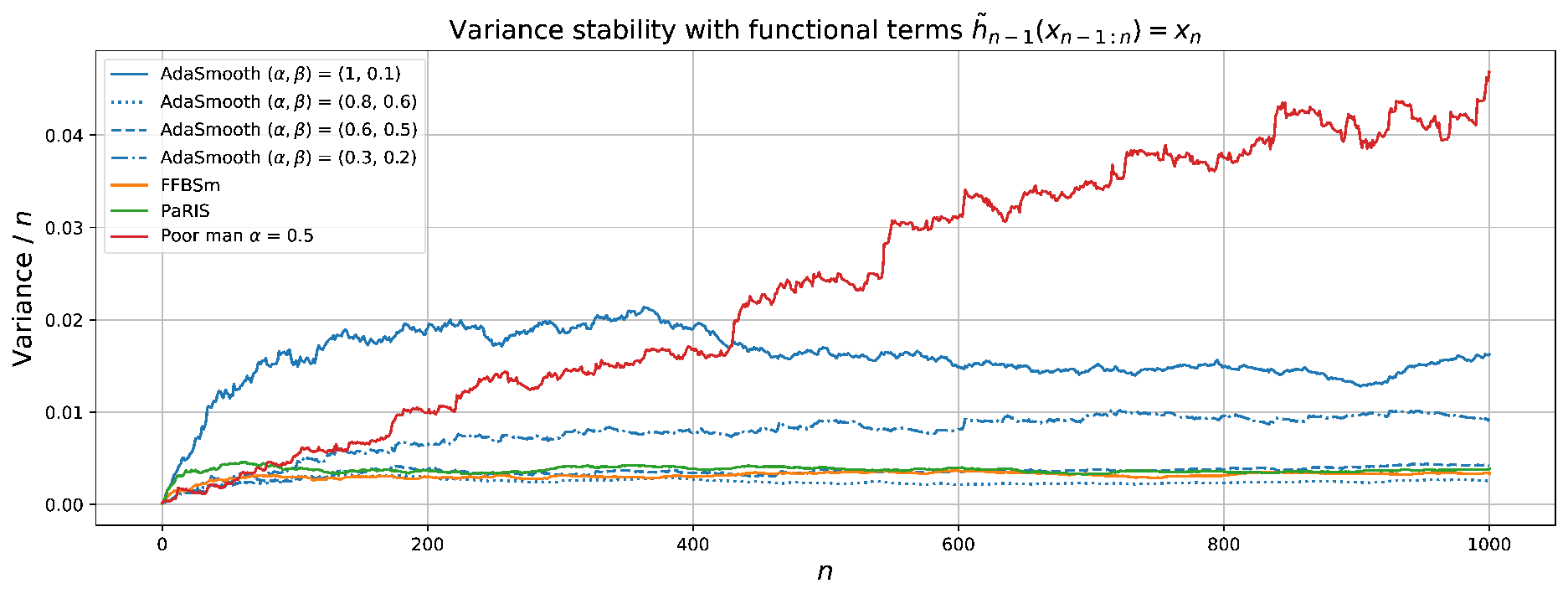}
	\caption{
	Time-normalized empirical variances, produced by AdaSmooth (with different parameterizations), the forward-only FFBSm, the PaRIS, and the poor man's smoother, of estimates of smoothed expectations of $\adds{n}(x_{0:n}) = \sum_{m=0}^n x_m$ for different $n$, in the linear Gaussian HMM in Section~\ref{subsec:lgss}. The empirical variances were obtained by rerunning each algorithm 100 times with $\N=500$ particles.}
	\label{fig:varianceLGSS}
\end{figure}

\subsection{Stochastic volatility model}\label{subsec:sv}
In order to investigate the performance of AdaSmooth in a nonlinear setting, we consider a modification of the stochastic volatility model proposed by \citet{hull:white:1987}. The observed \emph{stock returns} $(Y_n)_{n \in \nset}$ and the unobserved \emph{log-volatility} $(X_n)_{n \in \nset}$ are modeled as $\rset$-valued processes evolving recursively according to 
\[
\begin{split}
	X_{n+1}&=a X_n +\sigma U_{n+1}, \\
	Y_n&=b \exp(X_n/2) V_n, 
\end{split}
\quad n \in \nset,
\]
where $a \in \rset$, $b > 0$ and $\sigma > 0$ are model parameters. Here $V_0$ has standard Gaussian distribution, independent of $X_0$, while $(U_n, V_n)_{n \in\nsetpos}$ is a sequence of independent bivariate Gaussian random variables, with standard marginals and correlation $\rho \in (-1,1)$. All parameters of the model are assumed to be known, with $a = 0.975$, $b = 0.641$, $\sigma = 0.165$, values which appear frequently in the literature, and $ \rho = -0.1$; here the negative correlation reflects the fact that stock returns tend to be lower than average and oftentimes negative in high-risk environments (high volatility). With this parameterization, the log-volatility has a stationary distribution given by the zero-mean Gaussian distribution with variance $\sigma^2/(1-a^2)$, according to which $X_0$ is initialized. It immediately follows that, for all $n\in\nsetpos$, we may write $V_n=\rho U_n+\sqrt{1-\rho^2}W_n$, where $W_n$ has standard Gaussian distribution and is independent of $U_n$. Thus, for $n\in\nset$ the observation process can be alternatively expressed as $Y_{n+1}=b \exp(X_{n+1}/2)\{\rho (X_{n+1}-a X_n)/\sigma+\sqrt{1-\rho^2}W_{n+1}\}$, corresponding to a transition density $\g{}((x_n, x_{n+1}), y_{n+1})$. Note that model $(X_n, Y_n)_{n \in \nset}$ is not an HMM, since the correlation of the noise variables induces a conditional correlation between $Y_{n + 1}$ and $X_n$ given $X_{n + 1}$. Still, this does not cause any problem for us, since the general setting of Section~\ref{sec:background} does not presuppose the densities $(\ld{n})_{n \in \nset}$ to satisfy a Feynman--Kac-type decomposition \eqref{eq:feynkac}, and we may simply set $\ld{n}(x_n,x_{n+1})=q(x_n,x_{n+1})\g{}((x_n,x_{n+1}),y_{n+1})$, where $q$ is the transition density of the log-volatility and $y_{n + 1}$ given data at time $n + 1$. (Using the modified model $(\bar{X}_n, Y_n)_{n \in \nset}$, with $\bar{X}_n \eqdef X_{n - 1:n}$ being compound states, which is indeed an HMM, would not be an option, since the fact that the transition kernel of $(\bar{X}_n)_{n \in \nset}$ involves a Dirac mass implies that this HMM is not fully dominated.) 

For this model we consider online additive smoothing for three different additive functionals with terms given by $\addf{n}^{(1)}(x_n, x_{n+1})=x_{n+1}$, $\addf{n}^{(2)}(x_n, x_{n+1})=x_{n+1}^2$ and $\addf{n}^{(3)}(x_n, x_{n+1})=x_n x_{n+1}$, and compare the performance of AdaSmooth, for different parameterizations $(\alpha, \beta)$, to the poor man's smoother with adaptive selection, the PaRIS, and the forward-only FFBSm. Each of these algorithms was rerun 100 times for $ n=1000$ time steps. As in the previous example, the underlying particles were mutated using $q$ and selected without any adjustment of the particle weights. This choice is obviously sub-optimal, since evolving the particles ``blindly'', without taking information concerning subsequent observations into account, may cause faster weight degeneration. Hence, with more sophisticated adaptive proposals and adjustment functions we would expect even better results than those we are about to report, since a slower weight degeneration requires selection and backward-sampling to be applied less frequently. Like in the previous example, AdaSmooth outperforms by far its competitors. In Figure~\ref{fig:boxplotSV} we observe that with suitable choices of $\alpha$ and $\beta$, AdaSmooth is not only significantly faster---by one to two orders of magnitude---than the PaRIS and the forward-only FFBSm, it also exhibits lower variance. In fact, its computational complexity is of the same order as that of the poor man's smoother, whose stochastic instability is evident from the plots. Figure~\ref{fig:varianceSV}, which displays time-normalized empirical variances over time, confirms perfectly well our theoretical results in that AdaSmooth exhibits a linear increase of variance with $n$ for any parameterization. The fact that AdaSmooth provides the lowest variance in some cases is due to the adaptation of the selection schedule. In this example, the choice $(\alpha, \beta) = (0.8, 0.6)$ leads to a doubled computational complexity compared to $(0.6, 0.5)$, however without increasing notably the accuracy. Like in the previous example, we observed that an optimal tradeoff between variance and computational effort was obtained by setting $\alpha$ and $\beta$ to values around 0.5, with $\alpha \geq \beta$. 
 In all these simulations, any backward-sampling operation in AdaSmooth and the PaRIS was performed using the rejection-sampling technique described in Section~\ref{sec:algo}; on the other hand, when the backward probabilities were instead computed explicitly, the computational time of the PaRIS became similar to that of the FFBSm, while AdaSmooth slowed down by a factor 10--20.

\begin{figure}
	\centering
	\includegraphics[width=0.9\textwidth]{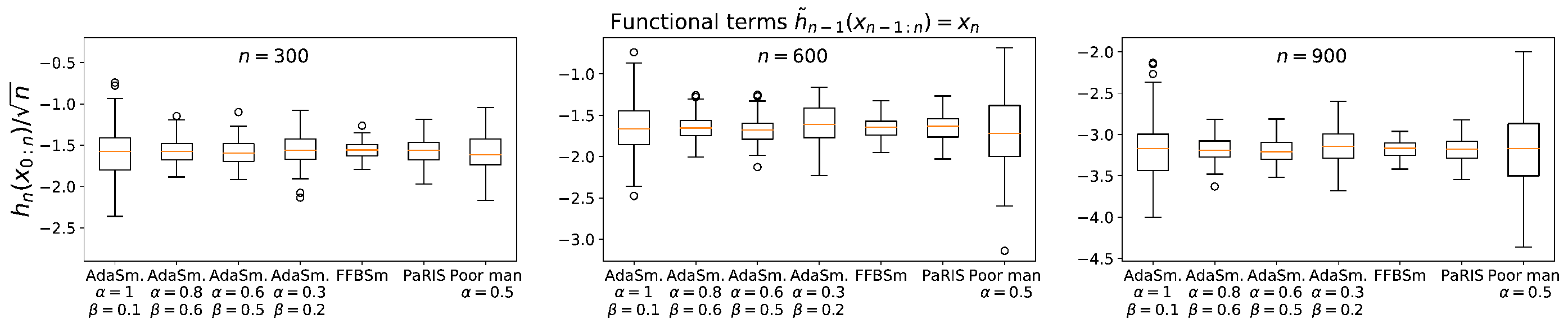}
	\includegraphics[width=0.9\textwidth]{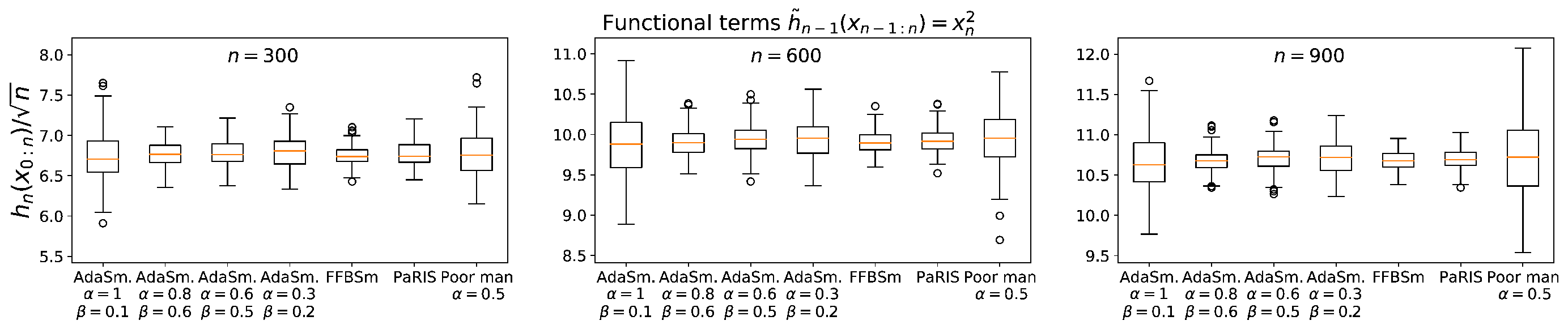}
	\includegraphics[width=0.9\textwidth]{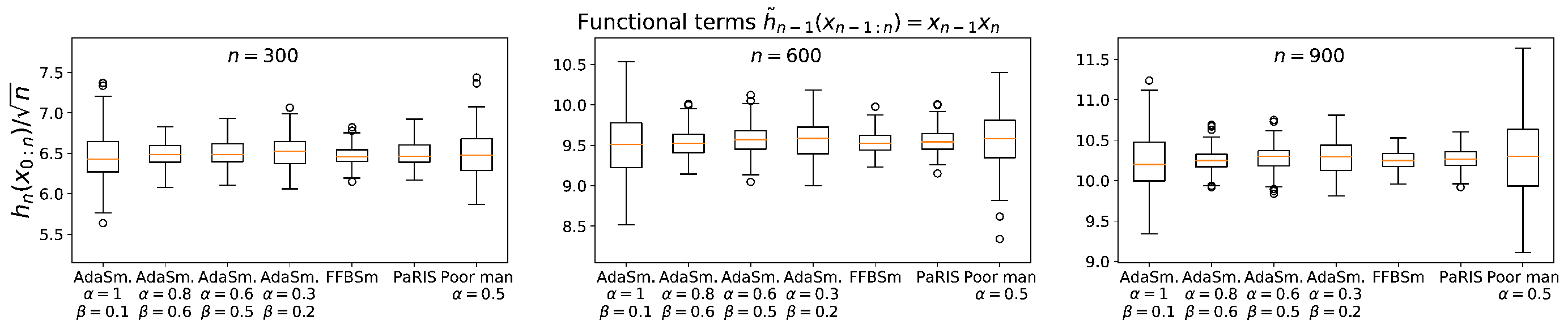}
		\caption{Boxplots of estimates (divided by $\sqrt{n}$) of smoothed expectations of three distinct smoothed functionals for $n \in \{300, 600,900\}$ in the stochastic volatility model in Section~\ref{subsec:sv}. For each algorithm, 100 estimates were produced using $\N=1000$ particles. Depending on the chosen parameters $\alpha$ and $\beta$, AdaSmooth was 8--45 times faster than the PaRIS, 80--400 times faster than the forward-only FFBSm and 2--10 times slower than the poor man's smoother. 
	%The same box plots of Figure \ref{fig:boxplotSV}, but for $\N=1000$. We observe that the ranges of the values are smaller than before, indicating the convergence of all algorithms as $\N$ increases. Depending on $\alpha$ and $\beta$, AdaSmooth was between 8 and 45 times faster than PaRIS, between 80 and 400 times faster than forward-only FFBSm and between 2 and 10 times slower than poor man's smoother.
	}
	\label{fig:boxplotSV}
\end{figure}

\begin{figure}[htb]
	\centering
	\includegraphics[width=0.85\textwidth]{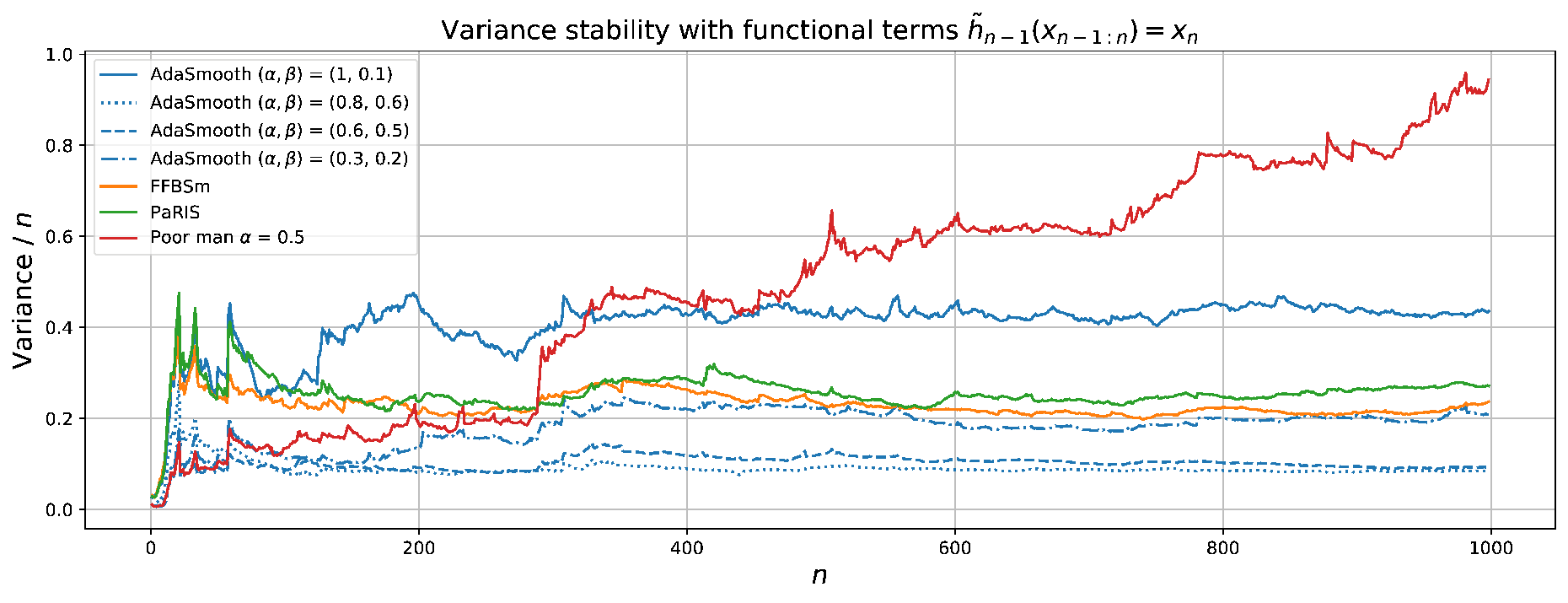}
	\caption{
	Time-normalized empirical variances, produced by AdaSmooth (with different parameterizations), the forward-only FFBSm, the PaRIS, and the poor man's smoother, of estimates of smoothed expectations of $\adds{n}(x_{0:n}) = \sum_{m=0}^n x_m$ for different $n$, in the stochastic volatility model in Section~\ref{subsec:sv}. The empirical variances were obtained by rerunning each algorithm 100 times with $\N=250$ particles.
	}
	\label{fig:varianceSV}
\end{figure}

Finally, Table~\ref{tab:samplingTimesSV} reports the average time duration between adaptive resampling operations as well as the number of selection operations on average between subsequent backward-sampling operations for different parameterizations. We observe that these averages stay basically constant when $\N$ is varied, suggesting that the operations are triggered regularly for given $\alpha$ and $\beta$. This supports our heuristic arguments outlined in Section~\ref{sec:theory}. We have also observed that the parameter $\beta$ may become useless if greater than $\alpha$, especially for $\alpha$ being greater than about 0.5, since selection is likely to automatically trigger backward sampling in that case.

\begin{table}[htb]
	\centering
	\resizebox{12cm}{!}{\begin{tabular}{c|cc|cc|cc|cc|cc|cc|cc|cc|cc|cc} %\begin{tabular}{c|rr|rr|rr|rr|rr|rr|rr|rr|rr|rr}
			\toprule
		$\N$%\diagbox{$(\alpha,\beta)$}{$N$}
	 &       \multicolumn{2}{c|}{50}     &       \multicolumn{2}{c|}{100}    &       \multicolumn{2}{c|}{250}    &       \multicolumn{2}{c|}{500}    &       \multicolumn{2}{c|}{1,000}   &       \multicolumn{2}{c|}{2,000}   &       \multicolumn{2}{c|}{5,000}   &       \multicolumn{2}{c|}{10,000}  &       \multicolumn{2}{c|}{50,000}  &       \multicolumn{2}{c}{100,000} \\ \midrule$\alpha=1.0$, $ \beta=0.1$   &  1.0& 14.2 &  1.0& 14.5 &  1.0& 14.3 &  1.0& 14.3 &  1.0& 14.3 &  1.0& 14.3 &  1.0& 14.3 &  1.0& 14.3 &  1.0& 14.3 &  1.0& 14.3 \\$\alpha=0.8$, $ \beta=0.6$ &   4.8& 1.2 &   4.6& 1.2 &   4.6& 1.1 &   4.6& 1.1 &   4.6& 1.0 &   4.7& 1.0 &   4.7& 1.0 &   4.6& 1.0 &   4.7& 1.0 &   4.6& 1.0 \\$\alpha=0.6$, $ \beta=0.5$ &   8.9& 1.5 &   8.9& 1.5 &   8.7& 1.6 &   8.6& 1.7 &   8.6& 1.7 &   8.5& 1.7 &   8.6& 1.7 &   8.6& 1.7 &   8.6& 1.7 &   8.6& 1.7 \\$\alpha=0.3$, $ \beta=0.2$ &  20.0& 2.6 &  19.2& 2.6 &  18.8& 2.7 &  18.4& 2.8 &  18.4& 2.7 &  18.2& 2.8 &  18.3& 2.8 &  18.1& 2.8 &  18.1& 2.8 &  18.1& 2.8 \\$\alpha=0.5$, $ \beta=0.6$ &  11.7& 1.0 &  11.5& 1.0 &  11.2& 1.0 &  11.2& 1.0 &  11.1& 1.0 &  11.2& 1.0 &  11.2& 1.0 &  11.2& 1.0 &  11.2& 1.0 &  11.2& 1.0 \\\bottomrule
	\end{tabular}}
	\caption{
	Average time duration between adaptive resampling operations (left columns) and the number of selection operations on average between subsequent backward-sampling operations (right columns) for different parameterizations of AdaSmooth in the stochastic volatility model in Section~\ref{subsec:sv}. For each parameterization and particle sample size, the values are based on a single run of the algorithm across $n=10,\!000$ time steps.
	}
	\label{tab:samplingTimesSV}
\end{table}

\section{CONCLUSIONS}\label{sec:discuss}
The presented algorithm, AdaSmooth, aims to combine the best of standard adaptive sequential importance sampling with resampling---which is computationally fast but numerically unstable---and the best of the PaRIS---whose long-term numerical stability is obtained via computationally costly backward sampling. As only limited code extensions of the standard particle filter is needed, AdaSmooth is very easily implemented or at least not significantly more complicated than the PaRIS in this respect. Still, the estimator is \emph{function-specific} in that the implementation depends on the additive functional under consideration. 

Even though the recursive updating step \eqref{eq:newupdate} of AdaSmooth, combining the forward and backward indices produced by the selection and backward-sampling operations, respectively, gives the estimator a very complex intrinsic dependence structure, we have been able to, by adapting existing theoretical analyses of the PaRIS, furnish the proposed algorithm with solid convergence and stability results, at least as long as the backward sampling schedule is adapted to the forward sampling schedule. As indicated by our numerical examples, AdaSmooth provides a tremendous improvement, of about one and two orders of magnitude, in terms of accuracy and computational efficiency compared to the PaRIS and the forward-only FFBSm algorithms, respectively. The improvement depends on the algorithmic parameters $\alpha$ and $\beta$, and in the models we tested it was observed that having both values around 0.5 with $\alpha$ larger than $\beta$ provides the best results. Even if we do not exclude that other combinations could work better on other examples, we dare to elevate this to a general rule of thumb when it comes to selecting these parameters. 

The theoretical analysis of the backward-sampling schedule $(\M[\N]{n})_{n \in \nset}$ proposed in Algorithm~\ref{algo:backdraws} remains an open---and possibly very complex---problem that we leave as future research. Other possible directions of future research are the extension of AdaSmooth beyond additive functionals and the further improvement of the algorithm via adaptation of the proposal kernels and adjustment multipliers of the APF. 

\begin{acks}[Acknowledgments]
	The authors would like to thank Dr. Sumeetpal Singh for valuable discussions.
\end{acks}

\begin{funding}
	J. Olsson and A. Mastrototaro are supported by the Swedish Research Council, Grant 2018-05230. 
\end{funding}

%\bibliographystyle{imsart-nameyear}
%\bibliography{motherofallbibs}

%%%%%%%%%%%%%%%%%%%%%%%%%%%%%%%%%%%%%%%%%%%%%%
%% Single Appendix:                         %%
%%%%%%%%%%%%%%%%%%%%%%%%%%%%%%%%%%%%%%%%%%%%%%
\begin{appendix}
	\newaliascnt{assumption2}{theorem}
	\newtheorem{assumption2}[assumption2]{Assumption}
	\aliascntresetthe{assumption2}	
	\renewcommand\thesection{\Alph{section}}
	\numberwithin{equation}{section}
	\numberwithin{algorithm}{section}
	\numberwithin{assumption}{section}
	
	\begin{center}
		\textbf{\\APPENDIX}
	\end{center}
	\section{Introduction to the Appendix}
\label{sec:introsupp}

In this Appendix we provide the proofs of the theoretical results in Section~\ref{sec:theory}; more specifically, we will present the proofs of Theorem~\ref{thm:as} (strong consistency of AdaSmooth in the case of a deterministic selection and backward-sampling schedule), Theorem~\ref{thm:clt} (asymptotic normality in the deterministic case), Theorem~\ref{thm:bound} ($\ordo(n)$ bound on the asymptotic variance). To this comes proofs of Proposition~\ref{prop:delta} and Lemma~\ref{lemma:ess}. The proofs of the first three theorems are, to some extent, adaptations of the corresponding proofs for the PaRIS presented by \citet{olsson:westerborn:2017} and \citet{gloaguen:lecorff:olsson:2021}. Still, the AdaSmooth updating rule \eqref{eq:newupdate}, which includes the forward indices of the APF as well as backward indices generated by backward sampling, induces a complex dependence structure that makes, as we will see, these adaptations highly non-trivial. In order to establish the mentioned results, we will proceed in two steps: first, we will, in Section~\ref{sec:convstab}, analyze an AdaSmooth algorithm with systematic resampling and an arbitrary, time-varying number of backward samples (Algorithm~\ref{algo:adasmsupp}), and prove the results in that case; second, we will, in Section~\ref{sec:detsel}, extend the results obtained in the systematic case to general deterministic schedules satisfying Assumption~\ref{assum:adaptiveRes} using an auxiliary path-based extension of the model in Section~\ref{sec:background}. 

The Appendix is organized as follows. The rest of this section recapitulates our model, presents some kernel notation needed in the proofs and introduces a modification of AdaSmooth with systematic resampling and arbitrary backward sample sizes (instead of either zero or one such draw, as in Algorithm~\ref{algo:new}). In Section~\ref{sec:convstab} we establish a Hoeffding-type exponential concentration inequality and the asymptotic normality of Algorithm~\ref{algo:adasmsupp}, together with a bound on the limit superior of the time-normalized asymptotic variance, the latter establishing the long-term numerical stability of the algorithm. In Section~\ref{sec:detsel}, the results of Section~\ref{sec:convstab} are extended to general deterministic resampling schedules via the auxiliary model extension mentioned above. Sections~\ref{sec:adaptres} and \ref{sec:propbound} provide the proofs of Lemma~\ref{lemma:ess} and Proposition~\ref{prop:delta}, respectively, the latter ensuring stability in the case of a regular backward-sampling schedule. 

\subsection{Notation and model description}

\subsubsection{Some kernel notation}

The coming developments require an expansion of the notation apparatus used in Section~\ref{sec:intro} \citealp[similar notation was used in][]{olsson:westerborn:2017,gloaguen:lecorff:olsson:2021}. For any measurable space $(\set{E}, \alg{E})$ we let $\meas{\alg{E}}$, $\probmeas{\alg{E}} \subset \meas{\alg{E}}$ and $\bmf{\alg{E}}$ be the sets of $\sigma$-finite measures on $\alg{E}$, probability measures on $\alg{E}$ and bounded $\alg{E}/\borel{\rset}$-measurable functions, respectively. For any $\mu \in \meas{\alg{E}}$ and $h \in \bmf{\alg{E}}$ we denote by $\mu h \eqdef \int h(x) \, \mu(dx)$ the Lebesgue integral of $h$ under $\mu$.  

The following kernel notation will be used over and over again. Let $(\set{E}_1, \alg{E}_1)$ and $(\set{E}_2, \alg{E}_2)$ be general measurable spaces; then a possibly unnormalized transition kernel $\kernel{Q}$ on $\set{E}_1 \times \alg{E}_2$ induces three operations: one on $\bmf{\alg{E}_1 \tensprod \alg{E}_2}$ and two on $\meas{\alg{E}_1}$. More specifically, for any $h \in \bmf{\alg{E}_1 \tensprod \alg{E}_2}$ and $\mu \in \meas{\alg{E}_1}$ we define the measurable function 
$$
	\kernel{Q} h: \set{E}_1 \ni x \mapsto \int h(x,y) \, \kernel{Q}(x, dy)
$$
as well as the measures 
\[
	\begin{split}
		\mu \kernel{Q} : \alg{E}_2 \ni A &\mapsto \int \kernel{Q}(x, A) \, \mu(dx), \\
		\mu \tensprod \kernel{Q} : \alg{E}_1 \tensprod \alg{E}_2 \ni A &\mapsto \iint_A  \kernel{Q}(x, dy) \, \mu(dx). 
	\end{split}
\]
Now, let $(\set{E}_3, \alg{E}_3)$ be another measurable space and $\kernel{P}$ a possibly unnormalized transition kernel on $\set{E}_2 \times \alg{E}_3$; then we define two kind of products of $\kernel{Q}$ and $\kernel{P}$, namely the \emph{product kernel} 
$$
	\kernel{Q} \kernel{P}: \set{E}_1 \times \alg{E}_3 \ni (x, A) \mapsto \int \kernel{Q}(x, dy) \, \kernel{P}(y, A)
$$ 
and the \emph{tensor-product kernel} 
$$ 
	\kernel{Q} \tensprod \kernel{P} : \set{E} _1\times (\alg{E}_2 \tensprod \alg{E}_3) \ni (x, A) \mapsto \iint_A \kernel{P}(x, dy) \, \kernel{Q}(y, dz).
$$ 
We will sometimes define kernels (measures) by specifying their products with (expectations of) bounded measurable functions. Finally, from time to time we will write $\kernel{Q}^2 h \eqdef (\kernel{Q} h)^2$, $\kernel{Q} h^2 \eqdef \kernel{Q}(h^2)$, $ \mu^2 h \eqdef (\mu h)^2$ and $\mu h^2 \eqdef \mu(h^2)$. 

\subsubsection{The path-space model in Section~1.1 reconsidered}
\label{sec:model}

Recall the general path-space model given in Section~1.1, comprising sequences $(\set{X}_n, \alg{X}_n)_{n \in \nset}$, $(\ld{n})_{n \in \nset}$ and $(\mu_n)_{n \in \nset}$ of measurable spaces, possible unnormalized transition densities and reference measures, respectively, as well as a possibly unnormalized density function $\chi$ on $\set{X}_0$. Here, for each $n \in \nset$, $\ld{n}$ and $\mu_n$ are defined on $\set{X}_n \times \set{X}_{n + 1}$ and $\alg{X}_n$, respectively. Using these quantities, we now introduce the unnormalized transition kernels 
\begin{equation}
	\lk{n}: \set{X}_n \times \alg{X}_{n+1} \ni (x_n, A) \mapsto \int_A \ld{n}(x_n, x_{n + 1}) \, \mu_{n+1}(dx_{n + 1}), \quad n \in \nset, 
\end{equation}
with, by convention, $\lk{n} \lk{m} = \text{id}$ if $n > m$. In addition, we will abuse notations and let $\chi$ also denote the distribution $\alg{X}_0 \ni A \mapsto \int_A \chi(x) \, \mu_0(dx)$. These notations allow us to express the path-space distributions defined in \eqref{eq:smooth} in the compact form 
\begin{equation} \label{eq:joint:smoothing:distribution}
	\post{0:n} = \frac{\Xinit \varotimes \lk{0} \varotimes \cdots \varotimes \lk{n-1}}{\Xinit \lk{0}\cdots\lk{n-1}\1{\set{X}_n}}, \quad n \in \nset, 
\end{equation}
and the corresponding marginals as
\begin{equation} \label{eq:filter:distribution}
	\post{n} = \frac{\Xinit  \lk{0} \cdots \lk{n-1}}{\Xinit \lk{0} \cdots \lk{n-1} \1{\set{X}_n}}, \quad n \in \nset. 
\end{equation}
The following \emph{backward kernels} will play a key role in the forthcoming developments. For every $n \in \nset$, define
\begin{equation}
	\bkm{n} : \set{X}_{n + 1} \times \alg{X}_n \ni (x_{n + 1}, A) \mapsto \frac{\int_A \ld{n}(x_n, x_{n+1}) \, \post{n}(dx_n)}{\int \ld{n}(x_n' ,x_{n+1}) \, \post{n}(dx_n')}. 
\end{equation}
\citet[Lemma 2.2]{gloaguen:lecorff:olsson:2021} show that $\bkm{n}$ is a \emph{reverse kernel} with respect to $\post{n}$ and $\lk{n}$ in the sense that 
\begin{equation} \label{eq:reversed:kernel}
	\post{n} \varotimes \lk{n} = (\post{n} \lk{n}) \varotimes \bkm{n}. 
\end{equation}
On the basis of the backward kernels we define, for every $n \in \nset$, 
\begin{equation}
	\tstat{n} \eqdef 
	\begin{cases}
		\bkm{n - 1} \varotimes \cdots \varotimes \bkm{0} &\text{for } n \in \nsetpos, \\ 
		\text{id} &\text{for } n = 0,
	\end{cases}
\end{equation}
and, again by \citet[Lemma~2.2]{gloaguen:lecorff:olsson:2021}, it holds that $\post{0:n} = \post{n} \varotimes \tstat{n}$. For additive functionals $(h_n)_{n \in \nset}$ in the form \eqref{eq:adds}, the functions $(\tstat{n} h_n)_{n \in \nset}$ satisfy the forward recursion 
\begin{equation} \label{eq:forward:recursion}
	\tstat{n+1} \adds{n+1}=\bkm{n} \varotimes \tstat{n}(\adds{n} + \addf{n}) = \bkm{n}(\tstat{n}\adds{n}+\addf{n}), \quad n \in \nset. 
\end{equation}
Finally, for $n \in \nset$, $m \in \intvect{0}{n}$, $x_m \in\set{X}_m$ and $h \in \bmf{\alg{X}_0 \varotimes \cdots \varotimes \alg{X}_n}$, we define the \emph{retro-prospective kernels}
\begin{equation} \label{eq:retro:prospective}
	\begin{split}
		\BF{m}{n}h(x_m) &\eqdef \iint h(x_{0:n}) \, \tstat{m}(x_m ,dx_{0:m-1}) \, \lk{m}\cdots\lk{n-1}(x_m,dx_{m+1:n}), \\
		\BFcent{m}{n}h(x_m) &\eqdef \BF{m}{n}(h - \post{0:n}h)(x_m).
	\end{split}
\end{equation}

\section{Theoretical analysis of AdaSmooth in the case of systematic resampling}
\label{sec:convstab}

\subsection{AdaSmooth with systematic resampling}

As explained above we will first analyse a version of Algorithm~\ref{algo:new} with systematic resampling, \ie, with $\res[\N]{n} = \res{n} =1$ almost surely for all $n$. Moreover, instead of being restricted to being binary-valued, the sequence $(\varepsilon_n)_{n \in \nset}$ may now take on any nonnegative integer values. This allows us to incorporate multiple backward draws (and not only a single such draw) into the AdaSmooth updates, as described in Algorithm~\ref{algo:adasmsupp} below. Even if using more than one backward draw does not, as we will see, improve further on the stability of the algorithm, it decreases somewhat the variance of the estimator; we hence present this extension here for completeness. The APF routine is given in Algorithm~\ref{algo:sisr} in the main paper, which is parameterized by sequences $(\am{n})_{n \in \nset}$ and $(\hd_n)_{n \in \nset}$ of adjustment-weight functions and proposal transition densities, respectively. For each $n \in \nset$, we let 
$$
\hk_n : \set{X}_n \times \alg{X}_{n + 1}  \ni (x_n, A) \mapsto \int_A \hd_n(x_n, x_{n + 1}) \, \mu_{n+1}(dx_{n + 1}) 
$$
the Markov transition kernel induced by $\hd_n$. 

\begin{algorithm}[H]
\caption{AdaSmooth with systematic resampling and multiple backward draws.}
	\begin{algorithmic}[1]\label{algo:adasmsupp}
		\REQUIRE $(\epart{n}{i}, \tstat[i]{n}, \wgt{n}{i})_{i=1}^\N$, $\M{n}$.
		\STATE run $(\epart{n+1}{i}, \I{n+1}{i}, \wgt{n+1}{i})_{i=1}^\N \leftarrow \apf((\epart{n}{i},\wgt{n}{i})_{i=1}^\N)$;
		\FOR{$i=1\rightarrow\N$}
		\IF{$\M{n}\ne0$}
		\FOR{$j=1\rightarrow\M{n}$}
		\STATE draw $\bi{n+1}{i}{j}\sim \catdist((\backprob{n}{\N}(i,\ell))_{\ell=1}^\N)$;\label{line:backsample}
		\ENDFOR
		\ENDIF
		\STATE set $\tstat[i]{n+1}\leftarrow(1+\M{n})^{-1}\bigg(\tstat[\I{n+1}{i}]{n}+\addf{n}(\epart{n}{\I{n+1}{i}},\epart{n+1}{i}) + \sum_{j=1}^{\M{n}} \big( \tstat[\bi{n+1}{i}{j}]{n} + \addf{n}(\epart{n}{\bi{n+1}{i}{j}},\epart{n+1}{i}) \big) \bigg)$;\label{line:updatesupp}
		\ENDFOR
		\RETURN $(\epart{n+1}{i}, \tstat[i]{n+1}, \wgt{n+1}{i})_{i=1}^\N$.
	\end{algorithmic}
\end{algorithm}

On Line~\ref{line:backsample}, each backward index $\bi{n+1}{i}{j}$ is drawn from the particle-induced backward probabilities 
$$
	\backprob{n}{\N}(i, \ell) = \frac{\wgt{n}{\ell} \ld{n}(\epart{n}{\ell},\epart{n+1}{i})}{\sum_{\ell'=1}^{\N} \wgt{n}{\ell'} \ld{n}(\epart{n}{\ell'}, \epart{n+1}{i})}, \quad (i, \ell) \in \intvect{1}{\N}^2, 
$$
defined in Section~\ref{sec:prevwork}. As in there, the initial particles $(\epart{0}{i})_{i=1}^{\N}$ are drawn from $\nu^{\tensprod  \N}$, where $\nu$ is a probability measure which is supposed to dominate $\Xinit$ and whose density function we denote by the same symbol, $\nu$, and assigned the weights $\wgt{0}{i} \eqdef \Xinit(\epart{0}{i})/\nu(\epart{0}{i})$. In addition, we set $\tstat[i]{0} \eqdef \adds{0}(\epart{0}{i})$ for all $i \in \intvect{1}{\N}$. In the updating rule on Line~\ref{line:updatesupp} (and everywhere else in the paper), the convention $\sum_{j=1}^0 = 0$ is used. 

\subsection{Exponential concentration of Algorithm~\ref{algo:adasmsupp}}

In the following, let $(\M{n})_{n\in\nset}$ be a given sequence of nonnegative integers. Recall the $\sigma$-fields 
\begin{align}
	\partfiltbar{n} = 
	\begin{cases}
		\sigma((\epart{0}{i})_{i=1}^{\N}) & \text{for } n = 0, \\
		\sigma((\epart{0}{i})_{i=1}^{\N}, (\epart{m}{i}, I_{m}^{i}, \tstat[i]{m})_{i = 1}^\N : m \in \intvect{1}{n}) & \text{for } n \in \nsetpos,
	\end{cases}
\end{align}
defined in Section~1.2. In addition, we define  
\begin{align}
	\partfilt{n} \eqdef \begin{cases}
		\partfiltbar{0} & \text{for } n = 0, \\
		\partfiltbar{n-1} \vee \sigma(\{ \epart{n}{i}, I_{n}^i \}_{i=1}^{\N}) & \text{for } n \in \nsetpos.
	\end{cases}
\end{align}
Here $\partfiltbar{n}$ is generated by the output of the first $n$ iterations of Algorithm~\ref{algo:adasmsupp}, while $\partfilt{n}$ is generated by the first $n - 1$ iterations and one additional update of the APF.

The following lemmas will be used repeatedly in the following developments, where the first is imported from \citet[Lemma~C.2]{gloaguen:lecorff:olsson:2021} and restated here for completeness. The second lemma extends a similar result obtained by \citet[Lemma~12]{olsson:westerborn:2017} \citealp[see also][Lemma~B.2]{gloaguen:lecorff:olsson:2021} for the PaRIS to the more complex AdaSmooth updating rule. 

\begin{lemma}[\citet{gloaguen:lecorff:olsson:2021}] \label{lemma:1}
	For all $n\in \nset$ and $(\testf[n+1], \testfp[n+1]) \in \bmf{\alg{X}_{n+1}}^2$ it holds that
	\begin{align*}
		\post{n+1}(\tstat{n+1}\adds{n+1}\testf[n+1]+\testfp[n+1])=\frac{\post{n}\{\tstat{n}\adds{n}\lk{n}\testf[n+1]+\lk{n}(\addf{n}\testf[n+1]+\testfp[n+1])\}}{\post{n}\lk{n}\1{\set{X}_{n+1}}}.
	\end{align*}
\end{lemma}

\begin{lemma}\label{lemma:2}
	For all $n\in \nset$, $(\testf[n+1], \testfp[n+1]) \in \bmf{\alg{X}_{n+1}}^2$, $\N \in \nsetpos$ and $(\M{n})_{n\in\nset}$, the random variables $(\wgt{n+1}{i}\{\tstat[i]{n+1}\testf[n+1](\epart{n+1}{i}) + \testfp[n+1](\epart{n+1}{i})\})_{i=1}^{\N}$ are conditionally independent and identically distributed given $\partfiltbar{n}$ with common expectation 
	\begin{multline} \label{eq:cond:exp}
		\E\left[\wgt{n+1}{1}\{\tstat[1]{n+1}\testf[n+1](\epart{n+1}{1}) + \testfp[n+1](\epart{n+1}{1})\} \mid \partfiltbar{n}\right] \\
		=\left(\post[N]{n}\am{n}\right)^{-1}\sum_{i=1}^{\N} \frac{\wgt{n}{i}}{\wgtsum{n}}\{ \tstat[i]{n}\lk{n}\testf[n+1](\epart{n}{i})+\lk{n}(\addf{n}\testf[n+1] + \testfp[n+1])(\epart{n}{i}) \}.
	\end{multline}
\end{lemma}

\begin{proof}
At time $n$ the particles are resampled independently in proportion to their weights, yielding $(\epart{n}{\I{n+1}{i}})_{i=1}^\N$, where $\I{n+1}{i}$ is distributed according to $\catdist((\wgt{n}{\ell}\allowbreak \am{n}(\epart{n}{\ell}))_{\ell=1}^{\N})$. After this, each selected particle $\epart{n}{\I{n+1}{i}}$ is propagated according to $\hk_{n}(\epart{n}{\I{n+1}{i}},\cdot)$ to obtain $\epart{n+1}{i}$. The backward indices $(\bi{n+1}{i}{j})_{j=1}^\N$ are conditionally independent and identically distributed given the particle $\epart{n+1}{i}$ and the $\sigma$-field $\partfiltbar{n}$, and hence the statistics $(\tstat[i]{n+1})_{i=1}^\N$, obtained through Line~\ref{line:updatesupp}, are conditionally independent and identically distributed as well. It follows that also $(\wgt{n+1}{1}\{\tstat[1]{n+1}\testf[n+1](\epart{n+1}{1}) + \testfp[n+1](\epart{n+1}{1})\})_{i=1}^\N$ are conditionally independent and identically distributed given $\partfiltbar{n}$. 
	
In order to establish \eqref{eq:cond:exp}, we consider the two terms separately, the first one being
\begin{align} 
	\lefteqn{\E \left[ \wgt{n+1}{1} \tstat[1]{n+1}\testf[n+1](\epart{n+1}{1}) \mid \partfiltbar{n} \right]} \nonumber \\
	&= \E \left[ \wgt{n+1}{1}\testf[n+1](\epart{n+1}{1}) \E \left[ \tstat[1]{n+1} \mid \partfilt{n+1} \right] \mid \partfiltbar{n} \right] \nonumber \\
	&= \frac{1}{1+\M{n}} \E\left[\wgt{n+1}{1}\testf[n+1](\epart{n+1}{1})(\tstat[I_{n+1}^1]{n}+\addf{n}(\epart{n}{I_{n+1}^1},\epart{n+1}{1})) \mid \partfiltbar{n} \right] \nonumber \\ 
	& \quad + \frac{\M{n}}{1+\M{n}} \E\left[\wgt{n+1}{1}\testf[n+1](\epart{n+1}{1}) \sum_{j=1}^{\N} \backprob{n}{\N}(1, j) 
	(\tstat[j]{n}+\addf{n}(\epart{n}{j},\epart{n+1}{1})) \mid \partfiltbar{n} \right]. \nonumber \label{eq:condexp}
\end{align}
Now, recall from Algorithm~\ref{algo:sisr} that conditionally to $\partfiltbar{n}$, each particle $\epart{n + 1}{i}$ and its associated forward index $I_{n+1}^i$ at time $n+1$ are sampled from the mixture on $\intvect{1}{\N} \times \alg{X}_{n + 1}$ proportional to $\wgt{n}{i} \am{n}(\epart{n}{i}) \hk_n(\epart{n}{i}, \cdot)$; thus, for the first term it holds that  
\begin{align}
	\lefteqn{\E \left[ \wgt{n+1}{1} \tstat[1]{n+1}\testf[n+1](\epart{n+1}{1}) \mid \partfiltbar{n} \right]} \\
	&
	\!\begin{multlined}
	= \frac{1}{1+\M{n}}\\\times \sum_{i=1}^{\N} \frac{\wgt{n}{i}\am{n}(\epart{n}{i})}{\sum_{i'=1}^{\N} \wgt{n}{i'}\am{n}(\epart{n}{i'})} \int 	\frac{\ld{n}(\epart{n}{i},x)}{\am{n}(\epart{n}{i})\hd_{n}(\epart{n}{i},x)}\testf[n+1](x)(\tstat[i]{n}+\addf{n}(\epart{n}{i},x)) \, \hk_n(\epart{n}{i}, dx) \\ 
	+ \frac{\M{n}}{1+\M{n}} \sum_{i=1}^{\N}\frac{\wgt{n}{i}\am{n}(\epart{n}{i})}{\sum_{i'=1}^{\N} \wgt{n}{i'} \am{n}(\epart{n}{i'})}\int \frac{\ld{n}(\epart{n}{i},x)}{\am{n}(\epart{n}{i})\hd_{n}(\epart{n}{i},x)}\testf[n+1](x) \\
	\times\sum_{j=1}^{\N}\frac{\wgt{n}{j}\ld{n}(\epart{n}{j},x)}{\sum_{k=1}^{\N} \wgt{n}{j'}\ld{n}(\epart{n}{j'}, x)}(\tstat[j]{n}+\addf{n}(\epart{n}{j},x)) \, \hk_n(\epart{n}{i}, dx)
	\end{multlined} \\
	&= (\post[N]{n}\am{n})^{-1}\sum_{i=1}^{\N}\frac{\wgt{n}{i}}{\wgtsum{n}}\{\tstat[i]{n}\lk{n}\testf[n+1](\epart{n}{i})+\lk{n}(\addf{n}\testf[n+1])(\epart{n}{i})\}.
\end{align}
Similarly, for the second term, 
\begin{align*}
	\E \left[ \wgt{n+1}{1} \testfp[n+1](\epart{n+1}{1}) \mid \partfiltbar{n} \right] = (\post[N]{n}\am{n})^{-1} \sum_{i=1}^{\N}\frac{\wgt{n}{i}}{\wgtsum{n}}\lk{n}\testfp[n+1](\epart{n}{i}),
\end{align*}
and the proof is completed by summing up these two identities.  
\end{proof}

The following assumption imposes the particle weights to the bounded, which is standard in importance sampling. 

\begin{assumption2}\label{assum:supp1}
	For every $n \in \nset$, the weight function
	\begin{equation}
		w_n:\set{X}_n\times\set{X}_{n+1}\ni(x,x')\mapsto\frac{\ld{n}(x,x')}{\am{n}(x)\hd_n(x,x')}\quad \text{and}\quad  w_{-1}:\set{X}_0\ni x\mapsto\frac{\Xinit(x)}{\nu(x)}
	\end{equation}
	as well as the adjustment-weight function $\am{n}$ are bounded. 
\end{assumption2}

Under Assumption~\ref{assum:supp1}, the following exponential concentration inequalities can, using Lemma~\ref{lemma:2}, be established along the very same lines as the proof of Proposition~B.1 in \citet{gloaguen:lecorff:olsson:2021}, and the proof is hence omitted. Recall that we in Section~\ref{sec:theory} defined $\mathsf{H}_n$ as the set of additive functionals $h_n$ in the form \eqref{eq:adds} with bounded terms.  

\begin{theorem}[Hoeffding-type inequalities] \label{thm:hoef}
	Let Assumption~\ref{assum:supp1} hold. Then for every $n\in\nset$, $\adds{n} \in \mathsf{H}_n$, $(\testf[n], \testfp[n]) \in \bmf{\alg{X}_{n}}^2$ and $(\M{m})_{m=0}^{n-1}$ there exist $(c_n, c'_n) \in (\rsetpos)^2$ (depending on $\adds{n}$, $(\M{m})_{m=0}^{n-1}$, $\testf[n]$ and $\testfp[n]$) such that for all $\N \in \nsetpos$ and all $\epsilon > 0$,
	\begin{itemize}
		\item[(i)] 
		{\footnotesize$
			\displaystyle \prob \left(\left \lvert \frac{1}{N} \sum_{i=1}^{\N} \wgt{n}{i} \{\tstat[i]{n}\testf[n](\epart{n}{i}) + \testfp[n](\epart{n}{i})\} - \dfrac{\post{n-1} \lk{n-1}(\tstat{n}\adds{n}\testf[n] + \testfp[n])}{\post{n-1} \am{n-1}} \right\rvert \ge \epsilon \right) \le c_n e^{- c_n' \N \epsilon^2}, 
		$}
		\item[(ii)] 
		{\footnotesize$
			\displaystyle \prob\left(\left\lvert \sum_{i=1}^{\N}\frac{\wgt{n}{i}}{\wgtsum{n}}\{\tstat[i]{n}\testf[n](\epart{n}{i})+ \testfp[n](\epart{n}{i})\}-\post{n}(\tstat{n}\adds{n}\testf[n]+ \testfp[n])  \right \rvert \ge \epsilon\right) \le c_n e^{- c'_n\N \epsilon^2}.
		$}
	\end{itemize}
\end{theorem}

The following corollary follows immediately by letting $\testf[n] \equiv \1{\set{X}_n}$ and $\testfp[n] \equiv 0$ in Theorem~\ref{thm:cltsupp}.

\begin{corollary} \label{cor:hoeffding}
Let Assumption~\ref{assum:supp1} hold. Then for every $n\in\nset$, $\adds{n} \in \mathsf{H}_n$ and $(\M{m})_{m=0}^{n-1}$ there exist $(c_n, c'_n) \in (\rsetpos)^2$ (depending on $\adds{n}$, $(\M{m})_{m=0}^{n-1}$, $\testf[n]$ and $\testfp[n]$) such that for all $\N \in \nsetpos$ and all $\epsilon > 0$,
$$
	\displaystyle \prob\left(\left\lvert \sum_{i=1}^{\N}\frac{\wgt{n}{i}}{\wgtsum{n}} \tstat[i]{n} - \post{0:n} \adds{n} \right \rvert \ge \epsilon\right) \le c_n e^{- c'_n\N \epsilon^2}.
$$
\end{corollary}

\subsection{Asymptotic normality}

Next, we aim to establish the following central limit theorem for estimates produced by Algorithm~\ref{algo:adasmsupp}. 

\begin{theorem}[asymptotic normality]\label{thm:cltsupp}
	Let Assumption~\ref{assum:supp1} hold. Then for every $n \in \nset$, $(\M{m})_{m=0}^{n-1}$, $(\testf[n],\testfp[n]) \in \bmf{\alg{X}_n}^2$ and $\adds{n} \in \mathsf{H}_n$, as $\N \rightarrow\infty$,
	\begin{align*}
		\sqrt{\N}\left(\sum_{i=1}^{\N}\frac{\wgt{n}{i}}{\wgtsum{n}}\{\tstat[i]{n}\testf[n](\epart{n}{i})+\testfp[n](\epart{n}{i})\}-\post{n}(\tstat{n}\adds{n}\testf[n]+\testfp[n])\right)\convd \asvar[]{n}{\testf[n]}{\testfp[n]} Z,
	\end{align*}
	where $Z$ has standard Gaussian distribution and 
	{\footnotesize
	\begin{align}\label{asvar}
		\lefteqn{\asvar[2]{n}{\testf[n]}{\testfp[n]}} \\
		&\eqdef \frac{\Xinit\{w_{-1}\BFcent[2]{0}{n}(\adds{n}\testf[n] + \testfp[n])\}}{(\Xinit \lk{0} \cdots \lk{n-1}\1{\set{X}_n})^2} + \sum_{m=0}^{n-1} \post{m}\am{m} \frac{\post{m}\lk{m}\{w_m \BFcent[2]{m+1}{n}(\adds{n}\testf[n] + \testfp[n])  \}}{(\post{m} \lk{m} \cdots \lk{n-1} \1{\set{X}_n})^2} \\
		&
		\!\begin{multlined}
		+ \sum_{m=0}^{n-1} \frac{\M{m} \post{m} \am{m}}{1+\M{m}}\\
		\times\sum_{\ell = 0}^{m} \frac{\post{\ell}\lk{\ell} (\bkm{\ell}(\tstat{\ell}\af{\ell} + \addf{\ell} - \tstat{\ell+1}\af{\ell+1})^2 \lk{\ell+1} \cdots \lk{m} \{ \bkm{m}w_m (\lk{m+1} \cdots \lk{n-1} \testf[n] )^2 \} )}{(\post{\ell}\lk{\ell} \cdots \lk{m-1}\1{\set{X}_m})(\post{m}\lk{m}\cdots\lk{n-1}\1{\set{X}_n})^2 \prod_{k=\ell}^{m}(1+\M{k})} 
		\end{multlined}\\
		&+ \sum_{m=0}^{n-1} \frac{\post{m} \am{m}}{1+\M{m}}\sum_{\ell = 0}^{m} \frac{\post{\ell}\lk{\ell} (\bkm{\ell}(\tstat{\ell}\af{\ell} + \addf{\ell} - \tstat{\ell+1}\af{\ell+1})^2 \lk{\ell+1} \cdots \lk{m}\{w_m (\lk{m+1} \cdots \lk{n-1} \testf[n] )^2\})}{(\post{\ell} \lk{\ell} \cdots \lk{m-1}\1{\set{X}_m})(\post{m}\lk{m}\cdots\lk{n-1}\1{\set{X}_n})^2 \prod_{k=\ell}^m (1+\M{k}) } \\
		&
		\!\begin{multlined}
		+ 2 \sum_{m=0}^{n-1} \M{m} \post{m} \am{m} \\
		\times\frac{\post{m}\lk{m}\{w_m \BFcent{m+1}{n}(\adds{n}\testf[n] + \testfp[n])(\tstat{m}\adds{m}+\addf{m}-\tstat{m+1}\adds{m+1}) \lk{m+1}\cdots\lk{n-1}\testf[n] \}}{(\post{m}\lk{m}\cdots\lk{n-1}\1{\set{X}_n})^2 (1+\M{m})^2 } 
		\end{multlined}\\
		&
		\!\begin{multlined}
		+ \sum_{m=0}^{n-1} \frac{\post{m}\am{m}}{(\post{m}\lk{m}\cdots\lk{n-1}\1{\set{X}_n})^2 (1+\M{m})^2}\bigg(\post{m}\lk{m} ((w_m-\bkm{m}w_m) \{\BFcent{m+1}{n}(\adds{n}\testf[n] + \testfp[n])\\
		+(\tstat{m}\adds{m}+\addf{m}-\tstat{m+1}\adds{m+1})\lk{m+1}\cdots\lk{n-1}\testf[n]\}^2)\bigg).
		\end{multlined}
	\end{align}
	}
\end{theorem}

Again, the following corollary follows immediately.  

\begin{corollary} \label{corollary:1}
	Let Assumption~\ref{assum:supp1} hold. Then for all $n\in \nset$, $(\M{m})_{m=0}^{n-1}$ and $\adds{n} \in\set{H}_n$, as $\N \rightarrow\infty$,
	\begin{equation}
		\sqrt{\N}\left(\sum_{i=1}^{\N}\frac{\wgt{n}{i}}{\wgtsum{n}}\tstat[i]{n} - \post{0:n}\adds{n}\right)\convd \sigma_n(\adds{n}) Z,
	\end{equation}
	where $Z$ has standard Gaussian distribution and
	{\footnotesize
	\begin{align}
		&\sigma_n^2(\adds{n}) = \frac{\Xinit(w_{-1}\BFcent[2]{0}{n}\adds{n})}{(\Xinit\lk{0}\cdots\lk{n-1}\1{\set{X}_n})^2}+\sum_{m=0}^{n-1} \post{m} \am{m} \frac{\post{m}\lk{m}(w_m \BFcent[2]{m+1}{n}\adds{n})}{(\post{m} \lk{m} \cdots\lk{n-1}\1{\set{X}_n})^2}\label{var_term:1} \\
		& 
		\!\begin{multlined}
		+ \sum_{m=0}^{n-1} \frac{\M{m} \post{m} \am{m}}{1+\M{m}} \\ \times\sum_{\ell = 0}^m\frac{\post{\ell} \lk{\ell} (\bkm{\ell}(\tstat{\ell} \af{\ell} + \addf{\ell} - \tstat{\ell+1} \af{\ell+1})^2 \lk{\ell+1} \cdots \lk{m}\{\bkm{m}w_m (\lk{m+1} \cdots \lk{n-1} \1{\set{X}_n})^2 \} )}{(\post{\ell}\lk{\ell} \cdots \lk{m-1}\1{\set{X}_m})(\post{m}\lk{m}\cdots\lk{n-1}\1{\set{X}_n}  )^2 \prod_{k=\ell}^m(1+\M{k})}\label{var_term:2}
		\end{multlined}\\ 
		& + \sum_{m=0}^{n-1} \frac{\post{m} \am{m}}{1+\M{m}}\sum_{\ell = 0}^{m} \frac{\post{\ell} \lk{\ell} (\bkm{\ell}(\tstat{\ell} \af{\ell} + \addf{\ell} - \tstat{\ell+1}\af{\ell+1})^2 \lk{\ell+1} \cdots \lk{m}\{w_m (\lk{m+1} \cdots \lk{n-1} \1{\set{X}_n})^2\} )}{(\post{\ell} \lk{\ell} \cdots \lk{m-1}\1{\set{X}_m})(\post{m} \lk{m}\cdots\lk{n-1}\1{\set{X}_n}  )^2 \prod_{k=\ell}^m (1+\M{k})}\label{var_term:3}
		\\ &+ 2 \sum_{m=0}^{n-1} \M{m} \post{m} \am{m} \frac{\post{m}\lk{m}\{w_m \BFcent{m+1}{n}\adds{n}(\tstat{m}\adds{m}+\addf{m}-\tstat{m+1}\adds{m+1}) \lk{m+1}\cdots\lk{n-1}\1{\set{X}_n} \}}{(1+\M{m})^2 (\post{m}\lk{m}\cdots\lk{n-1}\1{\set{X}_n})^2}\label{var_term:4}
		\\ &
		\!\begin{multlined}
		+\sum_{m=0}^{n-1}  \frac{\post{m} \am{m}}{(1+\M{m})^2(\post{m}\lk{m}\cdots\lk{n-1}\1{\set{X}_n})^2}\bigg(\post{m}\lk{m}((w_m - \bkm{m} w_m) \{ \BFcent{m+1}{n}\adds{n}\\
		+(\tstat{m}\adds{m}+\addf{m}-\tstat{m+1}\adds{m+1})\lk{m+1}\cdots\lk{n-1}\1{\set{X}_n} \}^2)\bigg).\label{var_term:5}
		\end{multlined}
	\end{align}
	\par}
\end{corollary}

%%%% PROOF OF CLT 

%%%%% LEMMA SQUARED STATISTIC  

The following lemma will be instrumental in the proof of Theorem~\ref{thm:cltsupp}.  

\begin{lemma}\label{lemma:3}
	Let Assumption \ref{assum:supp1} hold. Then for all $n\in\nset$, $(\M{m})_{m = 0}^{n-1}$ and $\testf[n] \in \bmf{\alg{X}_{n}}$, as $\N \to \infty$,  
	\begin{equation}\label{square_conv}
		\sum_{i=1}^{\N}\frac{\wgt{n}{i}}{\wgtsum{n}}(\tstat[i]{n})^2 \testf[n](\epart{n}{i}) \convp \post{n}(\tstatletter_n^2\adds{n}\testf[n]) + \sqc{n}(\testf[n]),
	\end{equation}
	where
	\begin{equation}\label{eq:eta}
		\sqc{n}(\testf[n]) \eqdef \sum_{\ell=0}^{n-1}\frac{\post{\ell} \lk{\ell}\{\bkm{\ell}(\tstatletter_\ell \adds{\ell} + \addf{\ell} - \tstatletter_{\ell+1} \adds{\ell+1})^2 \lk{\ell+1}\cdots\lk{n-1}\testf[n]\}}{(\post{\ell} \lk{\ell} \cdots \lk{n-1}\1{\set{X}_n}) \prod_{k=\ell}^{n-1}(1+\M{k})}.
	\end{equation}
\end{lemma}

\begin{proof}[Proof of Lemma~\ref{lemma:3}]
	We proceed by induction over $n$. First, the claim is straightforwardly true for $n=0$; indeed, since the algorithm is initialized using standard importance sampling,
	\begin{equation}
		\sum_{i=1}^{\N}\frac{\wgt{0}{i}}{\wgtsum{0}}(\tstat[i]{0})^2\testf[0](\epart{0}{i})=\sum_{i=1}^{\N}\frac{\wgt{0}{i}}{\wgtsum{0}}\adds{0}^2(\epart{0}{i}) \testf[0](\epart{0}{i})\convp\post{0}(\adds{0}^2\testf[0])=\post{0}(\tstat{0}^2\adds{0}\testf[0]).
	\end{equation}
	Thus, we assume that \eqref{square_conv} holds true for some arbitrary $n\in \nset$ and show that it holds true also for $ n + 1 $. Note that by Theorem \ref{thm:hoef}(i) it holds that $\N^{-1} \wgtsum{n+1}\convp \post{n}\lk{n}\1{\set{X}_{n+1}}/\post{n}\am{n}$ as $\N$ tends to infinity. Moreover, let
	\begin{equation}
		\arr{i} \eqdef \N^{-1}\wgt{n+1}{i}(\tstat[i]{n+1})^2\testf[n+1](\epart{n+1}{i})\quad \N \in \nsetpos, \ i \in \intvect{1}{\N}.
	\end{equation}
	Using conditional independence of the multinomial resampling mechanism and that the backward draws are conditionally independent and identically distribution given the new particle and $\partfiltbar{n}$,
	$$
		\sum_{i=1}^{\N}\E \left[\arr{i}\mid\partfiltbar{n}\right] = \E\left[\wgt{n+1}{1}(\tstat[1]{n+1})^2\testf[n+1](\epart{n+1}{1})\mid\partfiltbar{n}\right]=a_\N^1+a_\N^2+a_\N^3+a_\N^4,
	$$
	where
	\begin{align}
		&a_\N^1 \eqdef \frac{1}{(1+\M{n})^{2}}\E\left[\wgt{n+1}{1}\testf[n+1](\epart{n+1}{1})(\tstat[I_{n+1}^1]{n}+\addf{n}(\epart{n}{I_{n+1}^1},\epart{n+1}{1}))^2\mid\partfiltbar{n}\right], \\
		&
		\!\begin{multlined}
		a_\N^2 \eqdef \frac{2\M{n}}{(1+\M{n})^2 }\E \left[\wgt{n+1}{1}\testf[n+1](\epart{n+1}{1}) (\tstat[I_{n+1}^1]{n}+\addf{n}(\epart{n}{I_{n+1}^1},\epart{n+1}{1})) \vphantom{\E \left[ \tstat[\bi{n+1}{1}{1}]{n} + \addf{n}(\epart{n}{\bi{n+1}{1}{1}},\epart{n+1}{1}) \mid \partfilt{n+1} \right]} \right.  \\ 
		\times \left. \E \left[ \tstat[\bi{n+1}{1}{1}]{n} + \addf{n}(\epart{n}{\bi{n+1}{1}{1}},\epart{n+1}{1}) \mid \partfilt{n+1} \right] \mid \partfiltbar{n} \right],
		\end{multlined} \\
		&a_\N^3 \eqdef \frac{\M{n}}{(1+\M{n})^2}\E\left[\wgt{n+1}{1}\testf[n+1](\epart{n+1}{1})\E\left[ (\tstat[\bi{n+1}{1}{1}]{n}+\addf{n}(\epart{n}{\bi{n+1}{1}{1}},\epart{n+1}{1}))^2 \mid \partfilt{n+1}\right] \mid \partfiltbar{n}\right], \\
		&a_\N^4 \eqdef \frac{\M{n}(\M{n}-1)}{(1+\M{n})^2}\E\left[\wgt{n+1}{1}\testf[n+1](\epart{n+1}{1})\E^2\left[\tstat[\bi{n+1}{1}{1}]{n}+\addf{n}(\epart{n}{\bi{n+1}{1}{1}},\epart{n+1}{1})\mid\partfilt{n+1}\right]\mid\partfiltbar{n}\right]. 
	\end{align}
	We treat separately the four terms, starting with $a_\N^1$. Write
	\begin{align}
		&
		\!\begin{multlined}
		a_\N^1=\frac{1}{(1+\M{n})^{2}}\\\times\sum_{i=1}^{\N}\frac{\wgt{n}{i}\am{n}(\epart{n}{i})}{\sum_{\ell=1}^{\N}\wgt{n}{\ell}\am{n}(\epart{n}{\ell})}\int \frac{\ld{n}(\epart{n}{i},x)}{\am{n}(\epart{n}{i})\hd_{n}(\epart{n}{i},x)}\testf[n+1](x)(\tstat[i]{n}+\addf{n}(\epart{n}{i},x))^2 \, \hk_n(\epart{n}{i}, dx)
		\end{multlined} \\
		&=\frac{(\post[\N]{n}\am{n})^{-1}}{(1+\M{n})^{2}} \sum_{i=1}^{\N} \frac{\wgt{n}{i}}{\wgtsum{n}}\{(\tstat[i]{n})^2\lk{n}\testf[n+1](\epart{n}{i})+2\tstat[i]{n}\lk{n}(\addf{n}\testf[n+1])(\epart{n}{i})+\lk{n}(\addf{n}^2\testf[n+1])(\epart{n}{i})\}; 
	\end{align}
	applying the induction hypothesis and Theorem~\ref{thm:hoef} yields 
	\begin{multline}
		a_\N^1\convp \frac{1}{(\post{n}\am{n})(1+\M{n})^2} \bigg(\post{n}(\tstatletter_n^2\adds{n}\lk{n}\testf[n+1])+\sqc{n}(\lk{n}\testf[n+1])\\+2\post{n}\{\tstat{n}\adds{n}\lk{n}(\addf{n}\testf[n+1])\} + \post{n}\lk{n}(\addf{n}^2\testf[n+1])\bigg)\\
		=\frac{(\post{n}\am{n})^{-1}}{(1+\M{n})^{2}}\left(\post{n}\lk{n}\{(\tstatletter_n\adds{n}+\addf{n})^2\testf[n+1]\}+\sqc{n}(\lk{n}\testf[n+1])\right).
	\end{multline}
	
	We turn to $a_\N^2$, which can be expressed as 
	\begin{align}
		&
		\!\begin{multlined}
			a_\N^2=\frac{2\M{n}(\post{n}\am{n})^{-1}}{(1+\M{n})^2} \sum_{i=1}^{\N} \frac{\wgt{n}{i}\am{n}(\epart{n}{i})}{\wgtsum{n}} \int \frac{\ld{n}(\epart{n}{i},x)}{\am{n}(\epart{n}{i})\hd_{n}(\epart{n}{i},x)} \testf[n+1](x)(\tstat[i]{n}+\addf{n}(\epart{n}{i},x)) \\
			\times \sum_{j=1}^{\N} \frac{\wgt{n}{j}\ld{n}(\epart{n}{j},x)}{\sum_{j = 1}^{\N} \wgt{n}{j'} \ld{n}(\epart{n}{j'},x)}(\tstat[j]{n}+\addf{n}(\epart{n}{j},x)) \, \hk_n(\epart{n}{i}, dx)
		\end{multlined} \\ 
		&
		\!\begin{multlined}
			=\frac{2\M{n}(\post[\N]{n}\am{n})^{-1}}{(1+\M{n})^{2}}\\ \times \sum_{i=1}^{\N}\frac{\wgt{n}{i}}{\wgtsum{n}} \int \testf[n+1](x)\left(\sum_{j=1}^{\N}\frac{\wgt{n}{j} \ld{n}(\epart{n}{j},x)}{\sum_{j' = 1}^{\N} \wgt{n}{j'} \ld{n}(\epart{n}{j'},x)}(\tstat[j]{n}+\addf{n}(\epart{n}{j},x)) \right)^2 
		\, \lk{n}(\epart{n}{i}, dx).
		\end{multlined} 
	\end{align}
	In order to find the limit of this quantity, we define the function 
	\begin{equation}
		\varphi_{\N}(x) \eqdef \testf[n+1](x) \left(\sum_{j=1}^{\N}\frac{\wgt{n}{j} \ld{n}(\epart{n}{j},x)}{\sum_{j' = 1}^{\N} \wgt{n}{j'} \ld{n}(\epart{n}{j'}, x)}(\tstat[j]{n}+\addf{n}(\epart{n}{j},x))\right)^2, \quad x \in \set{X}_{n + 1}, 
	\end{equation}
	which can be uniformly bounded according to $\supn{\varphi_{\N}}\le \supn{\testf[n+1]}\supn{\adds{n+1}}^2$. By Theorem~\ref{thm:hoef}, we conclude that for every $x$, $\prob$-a.s.,
	$$
		\lim_{\N\to\infty} \varphi_{\N}(x) = \testf[n+1](x) \bkm{n}^2(\tstat{n}\adds{n}+\addf{n})(x) = \testf[n+1](x)\tstat{n+1}^2\adds{n+1}(x), 
	$$
	and using Lemma~14 in \citet{olsson:westerborn:2017} yields, as $\N$ tends to infinity,  
	$$
		a_\N^2 = \frac{2\M{n}(\post[\N]{n}\am{n})^{-1}}{(1+\M{n})^{2}} \post[\N]{n} \lk{n} \varphi_{\N} \convp \frac{2\M{n}(\post{n}\am{n})^{-1}}{(1+\M{n})^2} \post{n} \lk{n}(\tstat{n+1}^2 \adds{n+1} \testf[n+1]).
	$$	
	
	The $a_\N^3$ is treated along the same lines as $a_\N^1$; indeed, write 
	\begin{align}
		&\!\begin{multlined}
			a_\N^3=\frac{\M{n}(\post[\N]{n}\am{n})^{-1}}{(1+\M{n})^{2}}\sum_{i=1}^{\N} \frac{\wgt{n}{i} \am{n}(\epart{n}{i})}{\wgtsum{n}} \int \frac{\ld{n}(\epart{n}{i},x)}{\am{n}(\epart{n}{i}) \hd_n(\epart{n}{i},x)} \testf[n+1](x) \\
			\times \sum_{j=1}^{\N} \frac{\wgt{n}{j}\ld{n}(\epart{n}{j},x)}{\sum_{j'=1}^{\N} \wgt{n}{j'} \ld{n}(\epart{n}{j'},x)}(\tstat[j]{n}+\addf{n}(\epart{n}{j}, x))^2 \, \hk_n(\epart{n}{i}, dx) \end{multlined}\\ 
		&
		\!\begin{multlined}
		=\frac{\M{n}(\post[\N]{n}\am{n})^{-1}}{(1+\M{n})^{2}}\\\times\sum_{j=1}^{\N}\frac{\wgt{n}{j}}{\wgtsum{n}} \{(\tstat[j]{n})^2\lk{n}\testf[n+1](\epart{n}{j})+2\tstat[i]{n}\lk{n}(\addf{n}\testf[n+1])(\epart{n}{j})+\lk{n}(\addf{n}^2\testf[n+1])(\epart{n}{j})\} \\
		\convp \frac{\M{n}(\post{n}\am{n})^{-1}}{(1+\M{n})^{2}}\left(\post{n}\lk{n}\{(\tstatletter_n\adds{n}+\addf{n})^2\testf[n+1]\} + \sqc{n}(\lk{n}\testf[n+1])\right), 
		\end{multlined}
	\end{align}
	where the limit follows from the induction hypothesis and Theorem~\ref{thm:hoef}. 
	
	Finally, the term $a_\N^4$ is handled in a similar way as $a_\N^2$, \ie, by applying Theorem~\ref{thm:hoef} and Lemma~14 in \citet{olsson:westerborn:2017} according to
	\begin{multline}
		a_\N^4 = \frac{\M{n}(\M{n}-1)(\post[\N]{n}\am{n})^{-1}}{(1+\M{n})^{2}}\sum_{i=1}^{\N} \frac{\wgt{n}{i}\am{n}(\epart{n}{i})}{\wgtsum{n}} \int \frac{\ld{n}(\epart{n}{i},x)}{\am{n}(\epart{n}{i})\hd_{n}(\epart{n}{i},x)} \testf[n+1](x) \\
		\times \left(\sum_{j=1}^{\N}\frac{\wgt{n}{j}\ld{n}(\epart{n}{j},x)}{\sum_{j' = 1}^{\N}\wgt{n}{j'} \ld{n}(\epart{n}{j'},x)}(\tstat[j]{n}+\addf{n}(\epart{n}{j},x))\right)^2 \, \hk_n(\epart{n}{i}, dx) \\ 
		\convp\frac{\M{n}(\M{n}-1)(\post{n}\am{n})^{-1}}{(1+\M{n})^{2}}\post{n}\lk{n}(\tstat{n+1}^2\adds{n+1}\testf[n+1]).
	\end{multline}
	Finally, combining the previous four limits,
	\begin{multline}\label{limexpsquare}
		\sum_{i=1}^{\N}\E[\arr{i}\mid\partfiltbar{n}]\convp\frac{(\post{n}\am{n})^{-1}}{(1+\M{n})^{2}}\left(\post{n}\lk{n}\{(\tstatletter_n\adds{n}+\addf{n})^2\testf[n+1]\}+\sqc{n}(\lk{n}\testf[n+1])\right) \\
		+ \frac{2\M{n}(\post{n} \am{n})^{-1}}{(1+\M{n})^2} \post{n}\lk{n}(\tstat{n+1}^2\adds{n+1}\testf[n+1]) \\
		+ \frac{\M{n}(\post{n} \am{n})^{-1}}{(1+\M{n})^2} \left(\post{n} \lk{n}\{(\tstatletter_n\adds{n}+\addf{n})^2\testf[n+1]\}+\sqc{n}(\lk{n}\testf[n+1])\right) \\
		+ \frac{\M{n}(\M{n}-1)(\post{n}\am{n})^{-1}}{(1+\M{n})^{2}} \post{n} \lk{n}(\tstat{n+1}^2\adds{n+1}\testf[n+1]) \\
		\hspace{-7cm}=\frac{\post{n}\lk{n}(\tstat{n+1}^2\adds{n+1}\testf[n+1])}{\post{n}\am{n}}\\+\frac{\post{n}\lk{n}\{(\tstatletter_n\adds{n}+\addf{n})^2\testf[n+1]\}-\post{n}\lk{n}(\tstat{n+1}^2\adds{n+1}\testf[n+1])+\sqc{n}(\lk{n}\testf[n+1])}{(1+\M{n})\post{n}\am{n}}. 
	\end{multline}
	Now, note that $\wgt{n+1}{i}(\tstat[i]{n+1})^2\testf[n+1](\epart{n+1}{i})$, $i \in \intvect{1}{\N}$, are all bounded by $ \supn{w_n}\allowbreak \supn{\adds{n+1}}^2 \supn{\testf[n+1]}$ and conditionally independent and identically distributed given $\partfiltbar{n}$; thus, using Hoeffding's inequality for conditional expectations we obtain, for all $\epsilon > 0$,
	\begin{align}
		\prob \left( \left\lvert\sum_{i=1}^{\N}\arr{i}-\sum_{i=1}^{\N}\E[\arr{i}\mid\partfiltbar{n}]\right\rvert\ge \epsilon\bigg)\le 2\exp\bigg(-\frac{2\N\epsilon^2}{\supn{w_n}\supn{\adds{n+1}}^2\supn{\testf[n+1]}} \right).
	\end{align}
	The limit of $\sum_{i=1}^{\N}\arr{i}$ in probability is hence equal to \eqref{limexpsquare}. Moreover, using Lemma~\ref{lemma:1},  
	\begin{align}
		\lefteqn{\post{n}\lk{n} \{(\tstatletter_n\adds{n}+\addf{n})^2\testf[n+1]\}-\post{n}\lk{n}(\tstat{n+1}^2\adds{n+1}\testf[n+1])}\\
		&= \post{n} \varotimes \lk{n}(\{(\tstat{n} \adds{n}+\addf{n})^2 - \tstat{n+1}^2 \adds{n+1}\} \testf[n+1]) \\
		&= \post{n} \lk{n} \varotimes \bkm{n} (\{(\tstat{n}\adds{n} + \addf{n})^2 - 2 \tstat{n+1}^2 \adds{n+1} + \tstat{n+1}^2 \adds{n+1}\}\testf[n+1]) \\
		&
		\!\begin{multlined}
		= \post{n} \lk{n}(\bkm{n} \{(\tstat{n} \adds{n}+\addf{n})^2 - 2 \bkm{n}(\tstat{n} \adds{n} + \addf{n}) \tstat{n+1}\adds{n+1}\\
		+\tstat{n+1}^2\adds{n+1}\}\testf[n+1]) 
		\end{multlined}\\
		&= \post{n} \lk{n}\{\bkm{n} (\tstat{n}\adds{n}+\addf{n}-\tstat{n+1}\adds{n+1})^2\testf[n+1] \},
	\end{align}
	which allows us to conclude, using \eqref{eq:eta}, that $\wgtsum{n+1}^{-1}\sum_{i=1}^{\N}\wgt{n+1}{i}(\tstat[i]{n+1})^2\testf[n+1](\epart{n+1}{i})$ tends in probability to
	\begin{align}
		&\!\begin{multlined}
			\frac{\post{n}\lk{n}(\tstat{n+1}^2\adds{n+1}\testf[n+1])}{\post{n}\lk{n}\1{\set{X}_{n+1}}}\\+\frac{\post{n}\lk{n}\{\bkm{n}(\tstat{n}\adds{n}+\addf{n}-\tstat{n+1}\adds{n+1})^2\testf[n+1]\}+\sqc{n}(\lk{n}\testf[n+1])}{(1+\M{n})\post{n}\lk{n}\1{\set{X}_{n+1}}}
		\end{multlined} \\
		&
		\!\begin{multlined}
			= \post{n+1}(\tstat{n+1}^2\adds{n+1}\testf[n+1])+\frac{\post{n}\lk{n}\{\bkm{n}(\tstat{n}\adds{n}+\addf{n}-\tstat{n+1}\adds{n+1})^2\testf[n+1]\}}{(1+\M{n})\post{n}\lk{n}\1{\set{X}_{n+1}}} \\
			+ \sum_{\ell=0}^{n-1}\frac{\post{\ell}\lk{\ell}\{\bkm{\ell}(\tstatletter_\ell\adds{\ell}+\addf{\ell}-\tstatletter_{\ell+1}\adds{\ell+1})^2\lk{\ell+1}\cdots\lk{n-1}\lk{n}\testf[n+1]\}}{(1+\M{n})(\post{n}\lk{n}\1{\set{X}_{n+1}})(\post{\ell}\lk{\ell}\cdots\lk{n-1}\1{\set{X}_n}) \prod_{k=\ell}^{n-1}(1+\M{k})}
		\end{multlined} 
		\\
		&
		\!\begin{multlined}
			= \post{n+1}(\tstat{n+1}^2\adds{n+1}\testf[n+1]) \\
			+ \sum_{\ell=0}^{n}\frac{\post{\ell}\lk{\ell}\{\bkm{\ell}(\tstatletter_\ell\adds{\ell}+\addf{\ell}-\tstatletter_{\ell+1}\adds{\ell+1})^2\lk{\ell+1}\cdots\lk{n}\testf[n+1]\}}{(\post{\ell}\lk{\ell}\cdots\lk{n}\1{\set{X}_{n+1}}) \prod_{k=\ell}^n(1+\M{k})}
		\end{multlined} 
		\\
		&= \post{n+1}(\tstat{n+1}^2\adds{n+1}\testf[n+1])+\sqc{n+1}(\testf[n+1]),
	\end{align}
	where we used the identity \eqref{eq:L:identity} in the last step. The proof is complete. 
\end{proof}

\begin{proof}[Proof of Theorem~\ref{thm:cltsupp}]
	We proceed by induction and suppose that the claim of the theorem holds true for some $n \in \nset$.  Pick arbitrarily $(\testf[n+1],\testfp[n+1]) \in \bmf{\alg{X}_{n+1}}^2$ and assume first that $\post{n+1}(\tstat{n+1}\adds{n+1}\testf[n+1]+\testfp[n+1])=0$ (the general case will be treated later). Then write
	$$
		\sqrt{N} \sum_{i=1}^{\N} \frac{\wgt{n+1}{i}}{\wgtsum{n+1}}\{ \tstat[i]{n+1} \testf[n+1](\epart{n+1}{i}) + \testfp[n+1](\epart{n+1}{i}) \} = \Delta_\N^1+\Delta_\N^2,
	$$
	where
	\begin{align}
		\Delta_\N^1 &\eqdef \N \wgtsum{n+1}^{-1} \frac{1}{\sqrt{N}} \sum_{i=1}^{\N} \left( \wgt{n+1}{i} \{ \tstat[i]{n+1} \testf[n+1](\epart{n+1}{i}) + \testfp[n+1](\epart{n+1}{i}) \} \right. \\
		&\hspace{30mm}  \left. - \E \left[ \wgt{n+1}{1}\{\tstat[1]{n+1}\testf[n+1](\epart{n+1}{1}) + \testfp[n+1](\epart{n+1}{1})\} \mid \partfiltbar{n} \right]  \right), \\
		\Delta_\N^2 &\eqdef \N \wgtsum{n+1}^{-1}  \sqrt{N}\E \left[ \wgt{n+1}{1}\{\tstat[1]{n+1}\testf[n+1](\epart{n+1}{1}) + \testfp[n+1](\epart{n+1}{1})\} \mid \partfiltbar{n} \right],
	\end{align}
	and by Lemma~\ref{lemma:2},
	\begin{align}
		\Delta_\N^1 &= \N \wgtsum{n+1}^{-1} \frac{1}{\sqrt{N}} \sum_{i=1}^{\N} \bigg( \wgt{n+1}{i} \{ \tstat[i]{n+1} \testf[n+1](\epart{n+1}{i}) + \testfp[n+1](\epart{n+1}{i}) \} \\
		&\quad  - (\post[N]{n}\am{n})^{-1}\sum_{\ell = 1}^{\N} \frac{\wgt{n}{\ell}}{\wgtsum{n}} \{\tstat[\ell]{n}\lk{n}\testf[n+1](\epart{n}{\ell}) + \lk{n}(\addf{n}\testf[n+1] + \testfp[n+1])(\epart{n}{\ell}) \} \bigg), \\
		\Delta_\N^2 &= \frac{\N \sqrt{N}}{\wgtsum{n+1}  \post[N]{n}\am{n}}\sum_{\ell = 1}^{\N} \frac{\wgt{n}{\ell}}{\wgtsum{n}} \{ \tstat[\ell]{n} \lk{n}\testf[n+1](\epart{n}{\ell}) + \lk{n}(\addf{n}\testf[n+1] + \testfp[n+1])(\epart{n}{\ell})  \}.
	\end{align}
	In order to establish the weak convergence of $\Delta_\N^1$, we define the triangular array
	\begin{equation}
		\arr{i} \eqdef \frac{1}{(1+\M{n})\sqrt{N}}\sum_{j=1}^{\M{n}}\arrterm{\epart{n+1}{i}}{I_{n+1}^i,\bi{n+1}{i}{j}}+\frac{\arrterm{\epart{n+1}{i}}{I_{n+1}^i,I_{n+1}^i}}{(1+\M{n})\sqrt{N}},
	\end{equation}
	$\N \in \nsetpos, \ i \in \intvect{1}{\N}$, where
	\begin{multline*}
		\arrterm{x}{i,j} \eqdef w_n(\epart{n}{i}, x)  \left(\{\tstat[j]{n}+\addf{n}(\epart{n}{j},x)\} \testf[n+1](x) + \testfp[n+1](x)\right)  \\
		- (\post[N]{n}\am{n})^{-1}\sum_{\ell = 1}^{\N} \frac{\wgt{n}{\ell}}{\wgtsum{n}} \{\tstat[\ell]{n}\lk{n}\testf[n+1](\epart{n}{\ell}) + \lk{n}(\addf{n}\testf[n+1] + \testfp[n+1])(\epart{n}{\ell}) \}, 
	\end{multline*}
	$(i,j,x)\in \intvect{1}{\N}^2\times\set{X}_{n+1}$. Note that with this definition, 
	\begin{equation}
		\Delta_\N^1 = \N \wgtsum{n+1}^{-1} \sum_{i=1}^{\N} \arr{i}.
	\end{equation}
	By Lemma \ref{lemma:2} it holds that $\E[\arr{i} \mid \partfiltbar{n}]=0$ for all $i$. Moreover, Assumption~\ref{assum:supp1} implies that $|{\arr{i}}|\le 2\supn{w_n}(\supn{\adds{n+1}}\supn{\testf[n+1]}+\supn{\testfp[n+1]})/\sqrt{\N}$. In order to find the weak limit of $\sum_{i = 1}^\N \arr{i}$ we apply Theorem~A.3 in \citet{douc:moulines:2008}; this requires checking two conditions, where the first is the convergence in probability of 
	\begin{align}
		&\!\begin{multlined}
		\sum_{i=1}^{\N} \E \left[(\arr{i})^2 \mid \partfiltbar{n} \right]
		=\frac{1}{(1+\M{n})^2}\E\bigg[\bigg(\sum_{j=1}^{\M{n}}\arrterm{\epart{n+1}{1}}{I_{n+1}^1,\bi{n+1}{1}{j}}\bigg)^2\mid \partfiltbar{n}\bigg] \\
		+\frac{2}{(1+\M{n})^2}\E\bigg[\arrterm{\epart{n+1}{1}}{I_{n+1}^1,I_{n+1}^1}\sum_{j=1}^{\M{n}}\arrterm{\epart{n+1}{1}}{I_{n+1}^1,\bi{n+1}{1}{j}} \mid \partfiltbar{n}\bigg] \\
		+\frac{1}{(1+\M{n})^2}\E \left[\arrterm[2]{\epart{n+1}{1}}{I_{n+1}^1,I_{n+1}^1} \mid \partfiltbar{n}\right] \end{multlined}\\
		&=a_\N^1+a_\N^2+a_\N^3+a_\N^4,
	\end{align}
	as $\N$ tends to infinity, where
	\begin{align}
		&a_\N^1\eqdef \frac{\M{n}}{(1+\M{n})^2}\E\left[\E\left[\arrterm[2]{\epart{n+1}{1}}{I_{n+1}^1,\bi{n+1}{1}{1}}\mid \partfilt{n+1}\right]\mid \partfiltbar{n}\right], \\
		&a_\N^2 \eqdef \frac{\M{n}(\M{n}-1)}{(1+\M{n})^2}\E\left[\E^2\left[\arrterm{\epart{n+1}{1}}{I_{n+1}^1,\bi{n+1}{1}{1}}\mid \partfilt{n+1}\right]\mid \partfiltbar{n}\right], \\
		&
		\!\begin{multlined}
		a_\N^3 \eqdef \frac{2\M{n}}{(1+\M{n})^2}\E\bigg[\arrterm{\epart{n+1}{1}}{I_{n+1}^1,I_{n+1}^1}\\
		\hspace{3cm}\times\E\left[\arrterm{\epart{n+1}{1}}{I_{n+1}^1,\bi{n+1}{1}{1}}\mid \partfilt{n+1}\right]\mid \partfiltbar{n}\bigg],
		\end{multlined} \\
		&a_\N^4 \eqdef \frac{1}{(1+\M{n})^2}\E\left[\arrterm[2]{\epart{n+1}{1}}{I_{n+1}^1,I_{n+1}^1} \mid\partfiltbar{n}\right].
	\end{align}
	We treat separately the four terms, starting with $a_\N^1$. Write, using the definition of $\tilde{\upsilon}_\N$,
	\begin{multline} \label{clt:5}
		\E\left[ \E\left[\arrterm[2]{\epart{n+1}{1}}{I_{n+1}^1,\bi{n+1}{1}{1}} \mid \partfilt{n+1} \right] \mid \partfiltbar{n} \right] = \E \bigg[ w_n^2(\epart{n}{I_{n+1}^1}, \epart{n+1}{1}) \\
		\times\E \bigg[ \bigg( \{\tstat[\bi{n+1}{1}{1}]{n}+\addf{n}(\epart{n}{\bi{n+1}{1}{1}},\epart{n+1}{1})\} \testf[n+1](\epart{n+1}{1}) + \testfp[n+1](\epart{n+1}{1})\bigg)^2 \mid \partfilt{n+1}\bigg] \mid \partfiltbar{n} \bigg] %
		\\  -(\post[N]{n}\am{n})^{-2}\left(\sum_{\ell = 1}^{\N} \frac{\wgt{n}{\ell}}{\wgtsum{n}} \{\tstat[\ell]{n}\lk{n}\testf[n+1](\epart{n}{\ell}) + \lk{n}(\addf{n}\testf[n+1] + \testfp[n+1])(\epart{n}{\ell}) \}\right)^2.
	\end{multline}
	By Theorem~\ref{thm:hoef} and Lemma~\ref{lemma:1} it holds that 
	\begin{multline}
		\sum_{\ell = 1}^{\N} \frac{\wgt{n}{\ell}}{\wgtsum{n}} \{\tstat[\ell]{n}\lk{n}\testf[n+1](\epart{n}{\ell}) + \lk{n}(\addf{n}\testf[n+1] + \testfp[n+1])(\epart{n}{\ell}) \} \\
		\convp\post{n}\{\tstatletter_{n}\adds{n}\lk{n}\testf[n+1]+ \lk{n}(\addf{n}\testf[n+1] + \testfp[n+1])\} \\
		= \post{n+1}(\tstatletter_{n+1}\adds{n+1}\testf[n+1]+\testfp[n+1])(\post{n}\lk{n}\1{\set{X}_{n+1}}) = 0,
	\end{multline}
	since $\post{n+1}(\tstatletter_{n+1}\adds{n+1}\testf[n+1]+\testfp[n+1])=0$ by assumption. Thus, we may focus on the first term on the right-hand side of \eqref{clt:5}, which can be written as
	\begin{align}
		\lefteqn{\E \left[ w_n^2(\epart{n}{I_{n+1}^1}, \epart{n+1}{1}) \sum_{j=1}^{\N} \backprob{n}{\N}(1, j) 
		\left( \vphantom{\testfp[n+1]} (\tstat[j]{n})^2\testf[n+1]^2(\epart{n+1}{1}) \right. \right.} \\
		&\hspace{1cm} + \{\addf{n}(\epart{n}{j},\epart{n+1}{1}) \testf[n+1](\epart{n+1}{1}) + \testfp[n+1](\epart{n+1}{1})\}^2 \\
		& \left. \left. \hspace{15mm} +2\tstat[j]{n} \testf[n+1](\epart{n+1}{1})\{\addf{n}(\epart{n}{j},\epart{n+1}{1})\testf[n+1](\epart{n+1}{1})+ \testfp[n+1](\epart{n+1}{1})\} \right) \mid \partfiltbar{n} \vphantom{\sum_{j=1}^{\N}} \right] \\
		&=(\post[\N]{n}\am{n})^{-1} \sum_{i=1}^{\N} \frac{\wgt{n}{i}\am{n}(\epart{n}{i})}{\wgtsum{n}} \int %\hd_{n}(\epart{n}{i},x)
		\left( \frac{\ld{n}(\epart{n}{i},x)}{\am{n}(\epart{n}{i}) \hd_{n}(\epart{n}{i},x)} \right)^2 \sum_{j=1}^{\N} \frac{\wgt{n}{j}\ld{n}(\epart{n}{j},x)}{\sum_{j'=1}^{\N}\wgt{n}{j'}\ld{n}(\epart{n}{j'},x)} \\
		&\hspace{1cm} \times \left( (\tstat[j]{n})^2\testf[n+1]^2(x)+ \{ \addf{n}(\epart{n}{j},x)\testf[n+1](x) + \testfp[n+1](x) \}^2 \right. \\
		&\hspace{2cm} \left. + 2 \tstat[j]{n}\testf[n+1](x) \{ \addf{n}(\epart{n}{j},x) \testf[n+1](x) + \testfp[n+1](x) \} \right) \, \hk_n(\epart{n}{i}, dx) \\
		&= \left(\sum_{i=1}^{\N}\frac{\wgt{n}{i}}{\wgtsum{n}}(\tstat[i]{n})^2\lk{n}(\varphi_\N \testf[n+1]^2)(\epart{n}{i})+\sum_{i=1}^{\N}\frac{\wgt{n}{i}}{\wgtsum{n}}\lk{n}\{(\addf{n}\testf[n+1]+\testfp[n+1])^2 \varphi_\N\}(\epart{n}{i}) \right. \\
		&\hspace{1cm} \left. + 2 \sum_{i=1}^{\N} \frac{\wgt{n}{i}}{\wgtsum{n}}\tstat[i]{n} \lk{n}\{(\addf{n}\testf[n+1]+\testfp[n+1])\varphi_\N \testf[n+1]\}(\epart{n}{i}) \right)(\post[\N]{n}\am{n})^{-1} , \label{eq:clt:5:first:term}
	\end{align}
	where we have defined the function 
	\begin{equation}
		\varphi_\N(x) \eqdef \frac{\sum_{i=1}^{\N}\wgt{n}{i}\ld{n}(\epart{n}{i},x) w_n(\epart{n}{i},x)}{\sum_{i'=1}^{\N} \wgt{n}{i'} \ld{n}(\epart{n}{i'},x)} ,\quad x \in \set{X}_{n + 1}.
	\end{equation}
	Note that by Theorem~\ref{thm:hoef} it holds that $\lim_{\N \to \infty} \varphi_\N(x) = \varphi(x)$, $\prob$-a.s., where $\varphi(x) \eqdef \bkm{n} w_n(x)$, $x \in \set{X}_{n + 1}$. We now examine in turn the limits of each of the three sums in \eqref{eq:clt:5:first:term}. By Lemma~\ref{lemma:3} it holds that
	\begin{equation}
		\sum_{i=1}^{\N}\frac{\wgt{n}{i}}{\wgtsum{n}}(\tstat[i]{n})^2\lk{n}(\varphi \testf[n+1]^2)(\epart{n}{i}) \convp\post{n}\{\tstat{n}^2\adds{n}\lk{n}(\bkm{n}w_n\testf[n+1]^2)\} + \sqc{n} \lk{n}(\bkm{n} w_n \testf[n+1]^2).  
	\end{equation}
	Moreover, note that  
	\begin{multline}
		\left \lvert \sum_{i=1}^{\N} \frac{\wgt{n}{i}}{\wgtsum{n}}(\tstat[i]{n})^2 \lk{n}(\varphi_\N \testf[n+1]^2)(\epart{n}{i}) -\sum_{i=1}^{\N} \frac{\wgt{n}{i}}{\wgtsum{n}}(\tstat[i]{n})^2 \lk{n}(\varphi \testf[n+1]^2)(\epart{n}{i}) \right \rvert \\
		\le \supn{\adds{n}}^2 \sum_{i=1}^{\N} \frac{\wgt{n}{i}}{\wgtsum{n}}\lk{n}(|\varphi_\N -\varphi|\testf[n+1]^2)(\epart{n}{i}) \label{eq:phiNphi};
	\end{multline}
	thus, since $|\varphi_\N(x) - \varphi(x)| \testf[n+1]^2(x) \le 2 \supn{w_n} \supn{\testf[n+1]}^2$ for all $x \in \set{X}_{n + 1}$, Lemma~14 in \citet{olsson:westerborn:2017} implies that \eqref{eq:phiNphi} tends to zero in probability as $\N$ tends to infinity. Combining the previous two results yields  
	$$
		\sum_{i = 1}^{\N} \frac{\wgt{n}{i}}{\wgtsum{n}}(\tstat[i]{n})^2 \lk{n}(\varphi_\N \testf[n+1]^2)(\epart{n}{i}) \convp \post{n} \{ \tstat{n}^2 \adds{n} \lk{n} (\bkm{n} w_n \testf[n+1]^2) \} + \sqc{n} \lk{n}(\bkm{n} w_n \testf[n+1]^2). 
	$$	
	By operating again with Theorem~\ref{thm:hoef} and Lemma~14 in \citet{olsson:westerborn:2017} the other two limits of \eqref{eq:clt:5:first:term} can be treated similarly, allowing us to establish that 
	\begin{multline}
	\sum_{i=1}^{\N} \frac{\wgt{n}{i}}{\wgtsum{n}}\lk{n}\{(\addf{n}\testf[n+1]+\testfp[n+1])^2\varphi_\N\}(\epart{n}{i}) \convp\post{n}\lk{n}\{(\addf{n}\testf[n+1]+\testfp[n+1])^2\bkm{n} w_n\},\\
	\hspace{-45mm}\sum_{i=1}^{\N}\frac{\wgt{n}{i}}{\wgtsum{n}} \tstat[i]{n} \lk{n}\{(\addf{n}\testf[n+1]+\testfp[n+1]) \varphi_\N \testf[n+1]\}(\epart{n}{i}) \\\convp \post{n}(\tstat{n} \adds{n}\lk{n}\{(\addf{n}\testf[n+1]+\testfp[n+1])\bkm{n} w_n \testf[n+1]\}).
	\end{multline}
	To sum up, it holds, as $\N$ tends to infinity, 
	\begin{multline}
		a_\N^1 \convp\frac{\M{n}(\post{n}\am{n})^{-1}}{(1+\M{n})^2}\\\times\left(\sqc{n}\lk{n}(\testf[n+1]^2\bkm{n} w_n)+\post{n}\lk{n}(  \{(\tstat{n}\adds{n}+\addf{n})\testf[n+1]+\testfp[n+1]\}^2 \bkm{n} w_n) \right).
	\end{multline}
	
	We turn to $a_\N^2$. First, define the function 
	$$ 		
		\zeta_\N(x) \eqdef \testf[n+1](x)\sum_{j=1}^{\N}\frac{\wgt{n}{j}\ld{n}(\epart{n}{j},x)}{\sum_{j'=1}^{\N} \wgt{n}{j'}\ld{n}(\epart{n}{j'},x)}\{\tstat[j]{n}+\addf{n}(\epart{n}{j},x)\}  + \testfp[n+1](x), \quad x \in \set{X}_{n + 1}, 
	$$
	and note that by Theorem~\ref{thm:hoef} and the recursion \eqref{eq:forward:recursion}, for every $x \in \set{X}_{n + 1}$, 
	\begin{align}
		\lim_{\N\to\infty} \zeta_\N(x) &= \testf[n+1](x)\bkm{n}(\tstat{n}\adds{n}+\addf{n})(x)  + \testfp[n+1](x) \\
		&= \testf[n+1](x)\tstat{n+1}\adds{n+1}(x)  + \testfp[n+1](x), \quad \mbox{$\prob$-a.s.}
	\end{align}
	With this definition, 
	\begin{align}
		\lefteqn{\E \left[\E^2 \left[\arrterm{\epart{n+1}{1}}{I_{n+1}^1,\bi{n+1}{1}{1}}\mid \partfilt{n+1}\right] \mid \partfiltbar{n}\right]} \hspace{10mm} \\
		&=\E\left[\left(\sum_{j=1}^{\N} \backprob{n}{\N}(1, j) \arrterm{\epart{n+1}{1}}{I_{n+1}^1,j}\right)^2 \mid \partfiltbar{n}\right] \\
		&=(\post[\N]{n}\am{n})^{-1} \sum_{i=1}^{\N} \frac{\wgt{n}{i}}{\wgtsum{n}} 
		\lk{n}(w_n \zeta_\N^2)(\epart{n}{i}) \\
		&\quad - (\post[\N]{n}\am{n})^{-2} \left(\sum_{\ell = 1}^{\N} \frac{\wgt{n}{\ell}}{\wgtsum{n}} \{\tstat[\ell]{n}\lk{n}\testf[n+1](\epart{n}{\ell}) + \lk{n}(\addf{n}\testf[n+1] + \testfp[n+1])(\epart{n}{\ell})\}\right)^2.	
	\end{align}
	Since, the second term tends, again, to zero in probability by assumption, Lemma~14 in \citet{olsson:westerborn:2017} implies that, as $\N$ tends to infinity, 
	$$
		a_\N^2 \convp \frac{\M{n}(\M{n}-1)}{(1+\M{n})^2}(\post{n}\am{n})^{-1}\post{n}\lk{n}\{w_n(\testf[n+1]\tstat{n+1}\adds{n+1}  + \testfp[n+1])^2\}.
	$$
	
	Next, we turn to $a_\N^3$, which is proportional to 
	\begin{align}
		\lefteqn{\E\left[\arrterm{\epart{n+1}{1}}{I_{n+1}^1,I_{n+1}^1}\E\left[\arrterm{\epart{n+1}{1}}{I_{n+1}^1,\bi{n+1}{1}{1}}\mid \partfilt{n+1}\right]\mid \partfiltbar{n}\right]} \\
		&=\E \left[ \vphantom{\sum_{\ell = 1}^{\N}} \arrterm{\epart{n+1}{1}}{I_{n+1}^1, I_{n+1}^1} w_n(\epart{n}{I_{n+1}^1}, \epart{n+1}{1}) \right. \\
		&\hspace{5mm}\quad \times \E \left[ \{\tstat[\bi{n+1}{1}{1}]{n} + \addf{n}(\epart{n}{\bi{n+1}{1}{1}}, \epart{n+1}{1})\} \testf[n+1](\epart{n+1}{1}) + \testfp[n+1](\epart{n+1}{1}) \mid \partfilt{n} \right] \\
		&\hspace{5mm}\quad \left. - (\post[N]{n}\am{n})^{-1} \sum_{\ell = 1}^{\N} \frac{\wgt{n}{\ell}}{\wgtsum{n}} \{\tstat[\ell]{n}\lk{n}\testf[n+1](\epart{n}{\ell}) + \lk{n}(\addf{n}\testf[n+1] + \testfp[n+1])(\epart{n}{\ell}) \} \mid \partfiltbar{n} \right] \\
		&=\E \left[ \arrterm{\epart{n+1}{1}}{I_{n+1}^1,I_{n+1}^1} \left( \vphantom{\sum_{\ell = 1}^{\N}} w_n(\epart{n}{I_{n+1}^1}, \epart{n+1}{1})  \zeta_\N(\epart{n+1}{1}) \right. \right. \\
		&\hspace{5mm} \quad \left. \left. - (\post[N]{n}\am{n})^{-1} \sum_{\ell = 1}^{\N} \frac{\wgt{n}{\ell}}{\wgtsum{n}} \{\tstat[\ell]{n}\lk{n}\testf[n+1](\epart{n}{\ell}) + \lk{n}(\addf{n}\testf[n+1] + \testfp[n+1])(\epart{n}{\ell}) \} \right) \mid \partfiltbar{n} \right] \\
		&
		\!\begin{multlined}
		=\E \bigg[ w_n^2(\epart{n}{I_{n+1}^1}, \epart{n+1}{1}) (\{\tstat[\I{n+1}{1}]{n} + \addf{n}(\epart{n}{\I{n+1}{1}},\epart{n+1}{1})\} \testf[n+1](\epart{n+1}{1}) \\+ \testfp[n+1](\epart{n+1}{1})) \zeta_\N(\epart{n+1}{1}) \mid \partfiltbar{n} \bigg]  \\
		- (\post[\N]{n} \am{n})^{-2} \left(\sum_{\ell = 1}^{\N} \frac{\wgt{n}{\ell}}{\wgtsum{n}} \{\tstat[\ell]{n}\lk{n}\testf[n+1](\epart{n}{\ell}) + \lk{n}(\addf{n}\testf[n+1] + \testfp[n+1])(\epart{n}{\ell})\} \right)^2 
		\end{multlined}\\
		&
		\!\begin{multlined}
		=(\post[\N]{n}\am{n})^{-1}\\\times\sum_{i=1}^{\N}\frac{\wgt{n}{i}}{\wgtsum{n}}\int w_n(\epart{n}{i},x)(\{\tstat[i]{n}+\addf{n}(\epart{n}{i},x)\} \testf[n+1](x) + \testfp[n+1](x))\zeta_\N(x) \, \lk{n}(\epart{n}{i}, dx)
		\\
		-(\post[\N]{n}\am{n})^{-2} \left( \sum_{\ell = 1}^{\N} \frac{\wgt{n}{\ell}}{\wgtsum{n}} \{\tstat[\ell]{n}\lk{n}\testf[n+1](\epart{n}{\ell}) + \lk{n}(\addf{n}\testf[n+1] + \testfp[n+1])(\epart{n}{\ell})\} \right)^2.
		\end{multlined} 
	\end{align}
	Again the second term converges to zero, and by proceeding as in \eqref{eq:phiNphi} and using Lemma~14 in \citet{olsson:westerborn:2017}, we establish that, as $\N$ tends to infinity, 
	\begin{multline}
		a_\N^3 \convp \frac{2\M{n}}{(1+\M{n})^2}(\post{n} \am{n})^{-1}\\\times \post{n} \lk{n} ( w_n \{(\tstat{n}\adds{n}+\addf{n}) \testf[n+1] + \testfp[n+1] \}(\tstat{n+1} \adds{n+1} \testf[n+1] + \testfp[n+1]) ). 
	\end{multline}
	
	Finally, the last term $a_\N^4$ remains to be analyzed. Note that
	\begin{align}
		\lefteqn{\E \left[ \arrterm[2]{\epart{n+1}{1}}{I_{n+1}^1, I_{n+1}^1} \mid \partfiltbar{n} \right]} \\
		&
		\!\begin{multlined}
		=(\post[\N]{n}\am{n})^{-1} \\\times\sum_{i=1}^{\N} \frac{\wgt{n}{i}}{\wgtsum{n}} \int 
		w_n(\epart{n}{i},x)(\{\tstat[i]{n}+\addf{n}(\epart{n}{i},x)\} \testf[n+1](x) + \testfp[n+1](x))^2 \, \lk{n}(\epart{n}{i}, dx) 
		\\
		-(\post[\N]{n}\am{n})^{-2}\left( \sum_{\ell = 1}^{\N} \frac{\wgt{n}{\ell}}{\wgtsum{n}} \{\tstat[\ell]{n}\lk{n}\testf[n+1](\epart{n}{\ell}) + \lk{n}(\addf{n}\testf[n+1] + \testfp[n+1])(\epart{n}{\ell})\} \right)^2
		\end{multlined}
		\\ &=  \left( \sum_{i=1}^{\N}\frac{\wgt{n}{i}}{\wgtsum{n}}(\tstat[i]{n})^2\lk{n}(w_n\testf[n+1]^2)(\epart{n}{i})
		+ \sum_{i=1}^{\N} \frac{\wgt{n}{i}}{\wgtsum{n}}\lk{n}\{w_n(\addf{n}\testf[n+1]+\testfp[n+1])^2\}(\epart{n}{i}) \right. 
		\\ &\hspace{2cm} \left. + 2 \sum_{i=1}^{\N}\frac{\wgt{n}{i}}{\wgtsum{n}}\tstat[i]{n}\lk{n}\{w_n(\addf{n}\testf[n+1]+\testfp[n+1]) \testf[n+1]\}(\epart{n}{i}) \right)(\post{n}\am{n})^{-1}
		\\&\hspace{10mm} -(\post[\N]{n}\am{n})^{-2}\left(\sum_{\ell = 1}^{\N} \frac{\wgt{n}{\ell}}{\wgtsum{n}} \{\tstat[\ell]{n}\lk{n}\testf[n+1](\epart{n}{\ell}) + \lk{n}(\addf{n}\testf[n+1] + \testfp[n+1])(\epart{n}{\ell}) \} \right)^2,
	\end{align}
	where, again, the limit in probability of the last term is zero by assumption. Thus, we may conclude that, as $\N$ tends to infinity,  
	\begin{multline*}
		a_\N^4 \convp \frac{(\post{n}\am{n})^{-1}}{(1+\M{n})^2}\left( \vphantom{\testfp[n+1]} \post{n} \{ \tstat{n}^2 \adds{n} \lk{n}(w_n \testf[n+1]^2) \} + \sqc{n} \lk{n}(w_n \testf[n+1]^2) \right. \\ 
		\left. + \post{n} \lk{n}\{w_n(\addf{n}\testf[n+1]+\testfp[n+1])^2\} + 2 \post{n}( \tstat{n} \adds{n} \lk{n}\{w_n(\addf{n}\testf[n+1]+\testfp[n+1]) \testf[n+1]\}) \right) \\
		= \frac{(\post{n}\am{n})^{-1}}{(1+\M{n})^2} \left(\sqc{n} \lk{n}(w_n\testf[n+1]^2) + \post{n} \lk{n}(w_n \{ (\tstat{n}\adds{n}+\addf{n})\testf[n+1]+\testfp[n+1] \}^2)\right).
	\end{multline*}
	We now finally combine previous results to obtain the limit, as $\N$ tends to infinity, 
	$$
		\sum_{i=1}^{\N}\E\left[(\arr{i})^2\mid \partfiltbar{n}\right] = a_\N^1+a_\N^2+a_\N^3+a_\N^4 \convp \delta_n^2 \langle \testf[n+1],\testfp[n+1] \rangle, 
	$$
	where 
	\begin{multline}
		\lefteqn{\delta_n^2 \langle \testf[n+1],\testfp[n+1] \rangle \eqdef \frac{\M{n}(\post{n}\am{n})^{-1}}{(1+\M{n})^2}\bigg(\sqc{n}\{\lk{n}(\testf[n+1]^2 \bkm{n} w_n)\}} \\+ \post{n} \lk{n} (\{(\tstat{n}\adds{n}+\addf{n})\testf[n+1]+\testfp[n+1]\}^2 \bkm{n} w_n) \bigg) \\
		+ \frac{\M{n}(\M{n}-1)}{(1+\M{n})^2}(\post{n} \am{n})^{-1} \post{n} \lk{n} \{ w_n(\testf[n+1]\tstat{n+1}\adds{n+1}  + \testfp[n+1])^2\} \\
		+ \frac{2\M{n}}{(1+\M{n})^2}(\post{n}\am{n})^{-1} \post{n} \lk{n} (w_n \{ (\tstat{n}\adds{n}+\addf{n})\testf[n+1]  + \testfp[n+1] \}(\tstat{n+1}\adds{n+1}\testf[n+1]  + \testfp[n+1])) \\
		+ \frac{(\post{n}\am{n})^{-1}}{(1+\M{n})^2} \left( \sqc{n}\{\lk{n}(w_n \testf[n+1]^2)\} + \post{n} \lk{n}( w_n \{ (\tstat{n}\adds{n}+\addf{n})\testf[n+1] + \testfp[n+1] \}^2) \right).
	\end{multline}
	Since, using Lemma \ref{lemma:1} twice,
	\begin{multline}
		\post{n} \lk{n}(\{(\tstat{n}\adds{n}+\addf{n})\testf[n+1]+\testfp[n+1]\}^2 \bkm{n} w_n) \\- \post{n} \lk{n} \{w_n(\testf[n+1] \tstat{n+1} \adds{n+1} + \testfp[n+1])^2 \} \\
		=\post{n}\lk{n} \{ \testf[n+1]^2\bkm{n} (\tstat{n}\adds{n}+\addf{n}-\tstat{n+1}\adds{n+1})^2 \bkm{n} w_n \}, 
	\end{multline}
	we may rewrite $\delta_n^2 \langle \testf[n+1], \testfp[n+1] \rangle$ as
	{\footnotesize\begin{multline}
		\delta_n^2 \langle \testf[n+1],\testfp[n+1] \rangle = \frac{\M{n}^2}{(1+\M{n})^2}(\post{n} \am{n})^{-1} \post{n} \lk{n} \{w_n(\testf[n+1] \tstat{n+1} \adds{n+1}  + \testfp[n+1])^2 \} \\
		+ \frac{\M{n}(\post{n}\am{n})^{-1}}{(1+\M{n})^2} \left( \sqc{n}\lk{n}(\testf[n+1]^2 \bkm{n} w_n) 	
		+ \post{n}\lk{n} \{ \testf[n+1]^2 \bkm{n}(\tstat{n} \adds{n} + \addf{n} - \tstat{n+1} \adds{n+1})^2 \bkm{n} w_n \} \right) \\
		+ \frac{2\M{n}}{(1+\M{n})^2}(\post{n} \am{n})^{-1} \post{n} \lk{n} \left(w_n\{(\tstat{n}\adds{n}+\addf{n}) \testf[n+1]  + \testfp[n+1]\}(\tstat{n+1} \adds{n+1}\testf[n+1]  + \testfp[n+1]) \right) \\
		+ \frac{(\post{n}\am{n})^{-1}}{(1+\M{n})^2} \left(\sqc{n}\lk{n}(w_n\testf[n+1]^2)+\post{n} \lk{n} (w_n \{ (\tstat{n}\adds{n}+\addf{n})\testf[n+1]+\testfp[n+1] \}^2) \right).
	\end{multline}}
	We will now simplify $\delta_n^2\langle \testf[n+1],\testfp[n+1]\rangle$ further by adding and subtracting three different terms: first, if we add and subtract $(\post{n}\am{n})^{-1}(1+\M{n})^{-2} \post{n} \lk{n} \{w_n(\testf[n+1]\tstat{n+1} \adds{n+1}  + \testfp[n+1])^2\}$, we obtain
	{\footnotesize\begin{multline}
		\delta_n^2\langle \testf[n+1], \testfp[n+1] \rangle = \frac{\M{n}^2+1}{(1+\M{n})^2}(\post{n} \am{n})^{-1} \post{n} \lk{n} \{w_n(\testf[n+1]\tstat{n+1} \adds{n+1} + \testfp[n+1])^2\} \\
		+ \frac{\M{n}(\post{n} \am{n})^{-1}}{(1+\M{n})^2} \left(\sqc{n}\lk{n}(\testf[n+1]^2 \bkm{n} w_n) +\post{n} \lk{n} \{ \testf[n+1]^2 \bkm{n}(\tstat{n} \adds{n} + \addf{n} - \tstat{n+1} \adds{n+1})^2 \bkm{n} w_n \} \right) \\
		+ \frac{2\M{n}}{(1+\M{n})^2}(\post{n}\am{n})^{-1}\post{n}\lk{n}\left(w_n \{ (\tstat{n}\adds{n}+\addf{n})\testf[n+1]  + \testfp[n+1] \}(\tstat{n+1}\adds{n+1}\testf[n+1]  + \testfp[n+1])\right) \\
		+ \frac{(\post{n}\am{n})^{-1}}{(1+\M{n})^2} \left( \sqc{n} \lk{n}(w_n \testf[n+1]^2)+\post{n} \lk{n} (w_n \{ (\tstat{n}\adds{n}+\addf{n})\testf[n+1]+\testfp[n+1] \}^2) \right. \\  
		\left. -\post{n}\lk{n}\{w_n(\testf[n+1]\tstat{n+1}\adds{n+1}  + \testfp[n+1])^2\} \right);
	\end{multline}}
	second, adding and subtracting $(\post{n}\am{n})^{-1}(1+\M{n})^{-2} \post{n} \lk{n}( \{ (\tstat{n}\adds{n}+\addf{n}) \testf[n+1] + \testfp[n+1] \}^2 \bkm{n} w_n)$ yields 
	{\footnotesize\begin{multline}
		\delta_n^2 \langle \testf[n+1], \testfp[n+1] \rangle = \frac{\M{n}^2+1}{(1+\M{n})^2}(\post{n} \am{n})^{-1} \post{n} \lk{n}\{w_n(\testf[n+1]\tstat{n+1} \adds{n+1}  + \testfp[n+1])^2\} \\
		+\frac{\M{n}(\post{n}\am{n})^{-1}}{(1+\M{n})^2} \left( \sqc{n} \lk{n}(\testf[n+1]^2\bkm{n} w_n) +\post{n} \lk{n} \{ \testf[n+1]^2 \bkm{n} (\tstat{n} \adds{n} + \addf{n} - \tstat{n+1} \adds{n+1})^2 \bkm{n} w_n \} \right) \\
		+\frac{2\M{n}}{(1+\M{n})^2}(\post{n}\am{n})^{-1}\post{n}\lk{n}\left(w_n(\{\tstat{n}\adds{n}+\addf{n}\}\testf[n+1]  + \testfp[n+1])(\tstat{n+1}\adds{n+1}\testf[n+1]  + \testfp[n+1])\right) \\
		+\frac{(\post{n}\am{n})^{-1}}{(1+\M{n})^2 }\left(\sqc{n}\lk{n}(w_n \testf[n+1]^2)+\post{n} \lk{n} ( \{ (\tstat{n}\adds{n}+\addf{n}) \testf[n+1] +\testfp[n+1] \}^2 \bkm{n} w_n ) \right. \\
		\left. -\post{n}\lk{n} \{w_n(\testf[n+1]\tstat{n+1}\adds{n+1}  + \testfp[n+1])^2\} \right) \\
		+\frac{(\post{n}\am{n})^{-1}}{(1+\M{n})^2}\post{n} \lk{n}((w_n-\bkm{n} w_n) \{ (\tstat{n}\adds{n}+\addf{n})\testf[n+1]+\testfp[n+1] \}^2),
	\end{multline}}
	which, since by \eqref{eq:reversed:kernel}, 
	\begin{multline}
		 \post{n} \lk{n} \{w_n \testf[n+1]^2 \bkm{n} (\tstat{n} \adds{n} + \addf{n} - \tstat{n+1} \adds{n+1})^2 \} \\
		 = \post{n} \lk{n} \{  \testf[n+1]^2 \bkm{n}(\tstat{n}\adds{n} + \addf{n} - \tstat{n+1} \adds{n+1})^2 \bkm{n} w_n \},
	\end{multline}
	we may rearrange into	
	{\footnotesize\begin{multline}
		\delta_n^2 \langle \testf[n+1], \testfp[n+1] \rangle = \frac{\M{n}^2+1}{(1+\M{n})^2}(\post{n} \am{n})^{-1} \post{n} \lk{n}\{w_n(\testf[n+1] \tstat{n+1} \adds{n+1} + \testfp[n+1])^2\} \\
		+ \frac{\M{n}(\post{n} \am{n})^{-1}}{(1+\M{n})^2} \left(\sqc{n} \lk{n}(\testf[n+1]^2 \bkm{n} w_n) +\post{n} \lk{n}\{\testf[n+1]^2 \bkm{n}(\tstat{n} \adds{n} + \addf{n} - \tstat{n+1} \adds{n+1})^2 \bkm{n} w_n \} \right) \\
		+ \frac{2\M{n}}{(1+\M{n})^2}(\post{n} \am{n})^{-1} \post{n} \lk{n} \left(w_n\{(\tstat{n}\adds{n}+\addf{n}) \testf[n+1]  + \testfp[n+1]\}(\tstat{n+1}\adds{n+1}\testf[n+1]  + \testfp[n+1]) \right) \\
		+ \frac{(\post{n}\am{n})^{-1}}{(1+\M{n})^2} \left( \sqc{n} \{\lk{n}(w_n \testf[n+1]^2)\} + \post{n}\lk{n} \{w_n\testf[n+1]^2 \bkm{n} (\tstat{n}\adds{n} + \addf{n} - \tstat{n+1} \adds{n+1})^2 \} \right) \\
		+ \frac{(\post{n} \am{n})^{-1}}{(1+\M{n})^2}\post{n}\lk{n}((w_n-\bkm{n} w_n)\{(\tstat{n} \adds{n} + \addf{n}) \testf[n+1]+\testfp[n+1]\}^2);
	\end{multline}}
	third, doing the same with $2\M{n}(\post{n}\am{n})^{-1}(1+\M{n})^{-2}\post{n}\lk{n}\{w_n(\testf[n+1]\tstat{n+1}\adds{n+1}  + \testfp[n+1])^2\}$ yields  
	{\footnotesize\begin{multline}
		\delta_n^2 \langle \testf[n+1], \testfp[n+1] \rangle = (\post{n} \am{n})^{-1} \post{n} \lk{n} \{ w_n(\testf[n+1] \tstat{n+1} \adds{n+1} + \testfp[n+1])^2\} \\
		+ \frac{\M{n}(\post{n}\am{n})^{-1}}{(1+\M{n})^2} \left( \sqc{n} \lk{n}(\testf[n+1]^2 \bkm{n} w_n) +\post{n} \lk{n}\{ \testf[n+1]^2\bkm{n} (\tstat{n} \adds{n} + \addf{n} - \tstat{n+1} \adds{n+1})^2 \bkm{n} w_n \} \right) \\ 
		+ \frac{(\post{n} \am{n})^{-1}}{(1+\M{n})^2} \left( \sqc{n} \lk{n}(w_n \testf[n+1]^2) + \post{n} \lk{n}\{w_n\testf[n+1]^2 \bkm{n}(\tstat{n} \adds{n}+\addf{n}-\tstat{n+1}\adds{n+1})^2 \} \right) \\ 
		+ \frac{2\M{n}}{(1+\M{n})^2}(\post{n} \am{n})^{-1} \post{n} \lk{n}\left(w_n\{(\tstat{n} \adds{n}+\addf{n}) \testf[n+1]  + \testfp[n+1]\}(\tstat{n+1}\adds{n+1}\testf[n+1]  + \testfp[n+1]) \right. \\
		\left. -w_n(\testf[n+1]\tstat{n+1}\adds{n+1}  + \testfp[n+1])^2 \right) \\
		+\frac{(\post{n} \am{n})^{-1}}{(1+\M{n})^2} \post{n} \lk{n}((w_n - \bkm{n} w_n)\{(\tstat{n} \adds{n} + \addf{n}) \testf[n+1] + \testfp[n+1]\}^2).
	\end{multline}}
	Finally, by manipulating separately the last two terms of the previous expression we obtain
	{\footnotesize\begin{multline} \label{eq:deltaAdd}
		\delta_n^2 \langle \testf[n+1], \testfp[n+1] \rangle= (\post{n} \am{n})^{-1} \post{n} \lk{n}\{w_n(\testf[n+1] \tstat{n+1} \adds{n+1} + \testfp[n+1])^2\} \\ 
		+ \frac{\M{n}(\post{n} \am{n})^{-1}}{(1+\M{n})^2} \left( \sqc{n}\lk{n}(\testf[n+1]^2 \bkm{n} w_n) + \post{n} \lk{n}\{ \testf[n+1]^2 \bkm{n}(\tstat{n} \adds{n} + \addf{n} - \tstat{n+1} \adds{n+1})^2 \bkm{n} w_n\} \right) \\
		+ \frac{(\post{n}\am{n})^{-1}}{(1+\M{n})^2} \left( \sqc{n} \lk{n}(w_n\testf[n+1]^2) + \post{n} \lk{n} \{w_n \testf[n+1]^2 \bkm{n}(\tstat{n}\adds{n} + \addf{n} - \tstat{n+1} \adds{n+1})^2 \} \right) \\
		+ \frac{2\M{n}}{(1+\M{n})^2}(\post{n} \am{n})^{-1} \post{n} \lk{n}\{w_n(\tstat{n+1} \adds{n+1} \testf[n+1]  + \testfp[n+1])(\tstat{n}\adds{n}+\addf{n} -\tstat{n+1}\adds{n+1})\testf[n+1]\} \\
		+ \frac{(\post{n} \am{n})^{-1}}{(1+\M{n})^2} \post{n} \lk{n}((w_n - \bkm{n} w_n) 
		\{ \tstat{n+1} \adds{n+1} \testf[n+1] + \testfp[n+1] + (\tstat{n} \adds{n} + \addf{n} - \tstat{n+1} \adds{n+1}) \testf[n+1] \}^2).
	\end{multline}}
	Now, by~Assumption~\ref{assum:supp1}, $|\arr{i}| \le 2\supn{w_n}(\supn{\adds{n+1}}\supn{\testf[n+1]}+\supn{\testfp[n+1]})/\sqrt{\N}$ for $i\in \intvect{1}{\N}$, which implies, for every $\epsilon > 0$, 
	\begin{multline}
		\sum_{i=1}^{\N}\E \left[ (\arr{i})^2 \1{\{|\arr{i}| \ge \epsilon\} }\mid \partfiltbar{n} \right] \\ \le 4 \supn{w_n}^2(\supn{\adds{n+1}} \supn{\testf[n+1]} + \supn{\testfp[n+1]})^2\\\times \1{\{2\supn{w_n}(\supn{\adds{n+1}} \supn{\testf[n+1]} + \supn{\testfp[n+1]}) \ge \epsilon \sqrt{\N}\}},
	\end{multline}
	where the right-hand side tends to zero as $\N$ tends to infinity. Thus, both the sufficient conditions of Theorem~A.3 in \citet{douc:moulines:2008} are satisfied, and we may conclude that for every $u \in \rset$, as $\N$ tends to infinity, 
	\begin{align}
		\E \left[ \exp \left( \operatorname{i}u \sum_{i=1}^{\N} \arr{i} \right) \mid \partfiltbar{n} \right] \convp \exp \left(- u^2 \delta_n^2 \langle \testf[n+1], \testfp[n+1] \rangle / 2 \right).
	\end{align}
	Moreover, in order to generalize to the case where $\post{n+1}(\tstat{n+1}\adds{n+1}\testf[n+1]+\testfp[n+1])$ is possibly non-zero we note that with $\testfbar[n + 1](x) \eqdef \testfp[n+1](x) - \post{n+1}(\tstat{n+1}\adds{n+1}\testf[n+1]+\testfp[n+1])$, $x \in \set{X}_{n + 1}$, it holds that $\post{n+1}(\tstat{n+1}\adds{n+1}\testf[n+1]+\testfbar[n+1]) = 0$.
	Also note that by Theorem~\ref{thm:hoef}, $\N \wgtsum{n+1}^{-1} \convp \post{n} \am{n}/(\post{n}\lk{n}\1{\set{X}_{n+1}})$ as $\N$ tends to infinity. Now, combining Theorem~A.3 in \citet{douc:moulines:2008}, the induction hypothesis, Lemma~A.5 in \citet{delmoral:moulines:olsson:verge:2016}, and Slutsky's lemma, we conclude that, as $\N \to \infty$, 
	\begin{multline}
		\sqrt{N} \bigg( \sum_{i=1}^{\N} \frac{\wgt{n+1}{i}}{\wgtsum{n+1}}\{ \tstat[i]{n+1} \testf[n+1](\epart{n+1}{i}) + \testfp[n+1](\epart{n+1}{i}) \} \\- \post{n+1}(\tstat{n+1}\adds{n+1}\testf[n+1]+\testfp[n+1]) \bigg) \\
		= \sqrt{\N} \sum_{i=1}^{\N}\frac{\wgt{n + 1}{i}}{\wgtsum{n + 1}}\{\tstat[i]{n + 1}\testf[n + 1](\epart{n + 1}{i}) + \testfbar[n + 1](\epart{n + 1}{i})\} 		\convd \asvar[]{n+1}{\testf[n+1]}{\testfp[n+1]} Z,
	\end{multline}
	where $Z$ has standard Gaussian distribution and 	
	%{ \footnotesize
	\begin{align}
		\lefteqn{\asvar[2]{n+1}{\testf[n+1]}{\testfp[n+1]}} \\
		& 
		\!\begin{multlined}
			\eqdef \frac{(\post{n}\am{n})^2}{(\post{n}\lk{n}\1{\set{X}_{n+1}})^2} \delta_n^2 \langle \testf[n+1], \testfp[n+1] - \post{n+1}(\tstat{n+1}\adds{n+1}\testf[n+1]+\testfp[n+1]) \rangle \\ 
			+ \frac{\asvar[2]{n}{\lk{n}\testf[n+1]}{\lk{n}\{\addf{n} \testf[n+1] + \testfp[n+1] - \post{n+1}(\tstat{n+1} \adds{n+1} \testf[n+1] + \testfp[n+1])\}}}{(\post{n} \lk{n} \1{\set{X}_{n+1}})^2}
		\end{multlined}
		\\ 
		& 
		\!\begin{multlined}
			= \post{n} \am{n} \frac{\post{n}\lk{n}(w_n\{\testf[n+1] \tstat{n+1} \adds{n+1}  + \testfp[n+1] - \post{n+1}(\tstat{n+1}\adds{n+1}\testf[n+1]+\testfp[n+1])\}^2)}{(\post{n}\lk{n}\1{\set{X}_{n+1}})^2} \\
			+  \frac{\M{n} \post{n} \am{n}}{(1+\M{n})^2(\post{n} \lk{n} \1{\set{X}_{n+1}})^2}\bigg(\sqc{n}\{\lk{n}(\testf[n+1]^2 \bkm{n} w_n)\} \\
			+ \post{n} \lk{n} \{\bkm{n}(\tstat{n} \adds{n} + \addf{n} - \tstat{n+1} \adds{n+1})^2 (\bkm{n} w_n) \testf[n+1]^2\}\bigg) \\
			+ \post{n} \am{n} \frac{\sqc{n}\{\lk{n}(w_n\testf[n+1]^2)\} + \post{n} \lk{n}\{ \bkm{n}(\tstat{n} \adds{n} + \addf{n} - \tstat{n+1} \adds{n+1})^2 w_n \testf[n+1]^2 \}}{(1+\M{n})^2(\post{n} \lk{n} \1{\set{X}_{n+1}})^2} \\
			+\frac{2 \M{n} \post{n} \am{n}}{(1+\M{n})^2(\post{n} \lk{n} \1{\set{X}_{n+1}})^2}\bigg(\post{n} \lk{n} (w_n \{\tstat{n+1} \adds{n+1}\testf[n+1]  + \testfp[n+1] \\
			-\post{n+1}(\tstat{n+1} \adds{n+1} \testf[n+1] + \testfp[n+1])\}(\tstat{n} \adds{n} + \addf{n} -\tstat{n+1}\adds{n+1}) \testf[n+1])\bigg)  \\
			+  \frac{\post{n} \am{n} }{(1+\M{n})^2(\post{n} \lk{n} \1{\set{X}_{n+1}})^2}\bigg(\post{n}\lk{n}((w_n - \bkm{n} w_n)\{\tstat{n+1}\adds{n+1}\testf[n+1]+\testfp[n+1]\\
			- \post{n+1}(\tstat{n+1}\adds{n+1}\testf[n+1]+\testfp[n+1])+(\tstat{n}\adds{n}+\addf{n} -\tstat{n+1}\adds{n+1})\testf[n+1]\}^2)\bigg) \\
			+ \frac{\asvar[2]{n}{\lk{n}\testf[n+1]}{\lk{n}\{\addf{n}\testf[n+1]+\testfp[n+1]-\post{n+1}(\tstat{n+1}\adds{n+1}\testf[n+1]+\testfp[n+1])\}}}{(\post{n}\lk{n}\1{\set{X}_{n+1}})^2}\label{eq:prevar}.
		\end{multlined}
	\end{align}
	%}
	The next step is to establish a non-recursive expression for the asymptotic variance. Recall the retro-prospective kernels defined in \eqref{eq:retro:prospective}; using Lemma~\ref{lemma:1} we may establish the recursive formula 
	\begin{multline}
		\lefteqn{\BFcent[]{m+1}{n}(h_n\lk{n}\testf[n+1]+\lk{n}\{\addf{n}\testf[n+1]+\testfp[n+1]-\post{n+1}(\tstat{n+1}\adds{n+1}\testf[n+1]+\testfp[n+1])\})} \\ 
		= \BF[]{m+1}{n}\left( \lk{n}\{(\adds{n}+\addf{n})\testf[n+1]+\testfp[n+1]\} - \post{n+1}(\tstat{n+1}\adds{n+1}\testf[n+1]+\testfp[n+1]) \lk{n} \1{\set{X}_{n+1}} \right. \\
		\left. - \post{0:n}\lk{n}\{(\adds{n}+\addf{n})\testf[n+1]+\testfp[n+1]\} + \post{n+1}(\tstat{n+1}\adds{n+1}\testf[n+1]+\testfp[n+1]) \post{n} \lk{n} \1{\set{X}_{n+1}}\right)
		\\
		= \BF[]{m+1}{n}\left(\lk{n}(\adds{n+1}\testf[n+1]+\testfp[n+1])-\post{n+1}(\tstat{n+1}\adds{n+1}\testf[n+1]+\testfp[n+1]) \lk{n} \1{\set{X}_{n+1}} \right. 
		\\ \left. -\post{0:n}\lk{n}(\adds{n+1}\testf[n+1]+\testfp[n+1]) + \post{n+1}(\tstat{n+1}\adds{n+1}\testf[n+1]+\testfp[n+1]) \post{n} \lk{n} \1{\set{X}_{n+1}}\right)
		\\
		= \BF[]{m+1}{n}\lk{n}\{\adds{n+1}\testf[n+1]+\testfp[n+1] - \post{n+1}(\tstat{n+1}\adds{n+1}\testf[n+1]+\testfp[n+1])\} \\
		= \BFcent[]{m+1}{n+1}(\adds{n+1}\testf[n+1]+\testfp[n+1]). \label{eq:retro:prospective:recursion}
	\end{multline}
	In addition, we note that
	$$
		\tstat{n+1} \adds{n+1} \testf[n+1] + \testfp[n+1] - \post{n+1}(\tstat{n+1}\adds{n+1}\testf[n+1]+\testfp[n+1])=\BFcent[]{n+1}{n+1}(\adds{n+1}\testf[n+1]+\testfp[n+1]), 
	$$
	and by combining the previous identities with the definition \eqref{eq:eta} of $\sqc{n}$ we may rewrite the first, incremental part of %five terms of 
	$\asvar[2]{n+1}{\testf[n+1]}{\testfp[n+1]}$ according to 
	\begin{align}
	\lefteqn{\frac{(\post{n} \am{n})^2}{(\post{n} \lk{n} \1{\set{X}_{n+1}})^2} \delta_n^2 \langle \testf[n+1], \testfp[n+1] - \post{n+1}(\tstat{n+1}\adds{n+1}\testf[n+1]+\testfp[n+1]) \rangle} \\
		&
		\!\begin{multlined}
			= \post{n} \am{n}\frac{\post{n} \lk{n} \{w_n \BFcent[2]{n+1}{n+1}(\adds{n+1} \testf[n+1] + \testfp[n+1])\}}{(\post{n} \lk{n} \1{\set{X}_{n+1}})^2} \\
			+ \frac{\M{n} \post{n} \am{n}}{1+\M{n}} \left( \sum_{\ell=0}^{n-1}\frac{\post{\ell} \lk{\ell} \{\bkm{\ell}(\tstatletter_\ell \adds{\ell}+\addf{\ell}-\tstatletter_{\ell+1}\adds{\ell+1})^2 \lk{\ell+1} \cdots \lk{n-1} \lk{n}(\bkm{n}w_n) \testf[n+1]^2\}}{(\post{\ell} \lk{\ell} \cdots \lk{n-1}\1{\set{X}_n})(\post{n} \lk{n} \1{\set{X}_{n+1}})^2 \prod_{k = \ell}^n (1+\M{k})} \right. \\
			\left. +\frac{\post{n}\lk{n}\{\bkm{n}(\tstat{n}\adds{n} + \addf{n} - \tstat{n+1} \adds{n+1})^2 (\bkm{n} w_n) \testf[n+1]^2\}}{(1+\M{n})(\post{n}\lk{n}\1{\set{X}_{n+1}})^2} \right) \\
			+ \frac{\post{n}\am{n}}{1+\M{n}} \left( \sum_{\ell=0}^{n-1}\frac{\post{\ell} \lk{\ell} \{ \bkm{\ell}(\tstatletter_\ell\adds{\ell} + \addf{\ell} - \tstatletter_{\ell+1} \adds{\ell+1})^2 \lk{\ell+1} \cdots \lk{n-1} \lk{n} w_n \testf[n+1]^2\}}{(\post{\ell} \lk{\ell} \cdots \lk{n-1}\1{\set{X}_n})(\post{n} \lk{n}\1{\set{X}_{n+1}})^2 \prod_{k=\ell}^n(1+\M{k})} \right. \\
			\left. +\frac{\post{n} \lk{n} \{ \bkm{n}(\tstat{n} \adds{n} + \addf{n} - \tstat{n+1} \adds{n+1})^2 w_n\testf[n+1]^2\}}{(1+\M{n})(\post{n}\lk{n}\1{\set{X}_{n+1}})^2} \right) \\
			+ 2 \M{n} \post{n} \am{n} \frac{\post{n}\lk{n}\{w_n\BFcent[]{n+1}{n+1}(\adds{n+1}\testf[n+1]+\testfp[n+1])(\tstat{n}\adds{n}+\addf{n} -\tstat{n+1}\adds{n+1})\testf[n+1]\}}{(1 + \M{n})^2 (\post{n}\lk{n}\1{\set{X}_{n+1}})^2} \\
			+  \frac{\post{n} \am{n}}{(1+\M{n})^2 (\post{n}\lk{n}\1{\set{X}_{n+1}})^2}\bigg(\post{n} \lk{n} ( (w_n - \bkm{n}w _n) \{ \BFcent[]{n+1}{n+1}(\adds{n+1}\testf[n+1]+\testfp[n+1])\\
			+(\tstat{n}\adds{n}+\addf{n} -\tstat{n+1}\adds{n+1})\testf[n+1] \}^2 )\bigg). 
		\end{multlined} %\label{eq:incremental:variance}
	\end{align}
	Moreover, using the induction hypothesis, we may express the last part of $\asvar[2]{n+1}{\testf[n+1]}{\testfp[n+1]}$ as 
     {\footnotesize
	\begin{align}
		\lefteqn{\frac{\asvar[2]{n}{\lk{n}\testf[n+1]}{\lk{n}(\addf{n}\testf[n+1]+\testfp[n+1]-\post{n+1}(\tstat{n+1}\adds{n+1}\testf[n+1]+\testfp[n+1]))}}{(\post{n}\lk{n}\1{\set{X}_{n+1}})^2}} \\
		&= \frac{\Xinit\{w_{-1}\BFcent[2]{0}{n+1}(\adds{n+1}\testf[n+1] + \testfp[n+1])\}}{(\Xinit \lk{0}\cdots\lk{n-1}\1{\set{X}_n})^2(\post{n}\lk{n}\1{\set{X}_{n+1}})^2}\\
		&+\sum_{m = 0}^{n - 1} \post{m} \am{m} \frac{\post{m} \lk{m} \{w_m \BFcent[2]{m+1}{n+1}(\adds{n+1}\testf[n+1] + \testfp[n+1])\}}{(\post{m} \lk{m} \cdots\lk{n-1} \1{\set{X}_n})^2(\post{n} \lk{n} \1{\set{X}_{n+1}})^2} \\
		& 
		\!\begin{multlined}
		+ \sum_{m = 0}^{n-1} \frac{\M{m} \post{m} \am{m}}{1+\M{m}}\\
		\times\sum_{\ell = 0}^{m} \frac{\post{\ell }\lk{\ell} ( \bkm{\ell}(\tstat{\ell} \af{\ell} + \addf{\ell} - \tstat{\ell+1} \af{\ell+1})^2 \lk{\ell+1} \cdots \lk{m}\{\bkm{m} w_m (\lk{m+1} \cdots \lk{n} \testf[n+1])^2\} )}{(\post{\ell} \lk{\ell} \cdots \lk{m-1} \1{\set{X}_m})(\post{m} \lk{m} \cdots \lk{n-1} \1{\set{X}_n})^2(\post{n} \lk{n} \1{\set{X}_{n+1}})^2 \prod_{k=\ell}^m(1+\M{k})} \end{multlined}\\
		& 
		\!\begin{multlined}
		+ \sum_{m=0}^{n-1} \frac{\post{m} \am{m}}{1+\M{m}}\\
		\times\sum_{\ell = 0}^m \frac{\post{\ell} \lk{\ell} ( \bkm{\ell}(\tstat{\ell} \af{\ell} + \addf{\ell} - \tstat{\ell+1} \af{\ell+1})^2 \lk{\ell+1} \cdots \lk{m}\{w_m( \lk{m+1} \cdots \lk{n} \testf[n+1] )^2 \} )}{(\post{\ell} \lk{\ell} \cdots \lk{m-1}\1{\set{X}_m})(\post{m}\lk{m}\cdots\lk{n-1}\1{\set{X}_n})^2(\post{n} \lk{n} \1{\set{X}_{n+1}})^2 \prod_{k = \ell}^m(1+\M{k})} 
		\end{multlined}\\
		& 
		\!\begin{multlined}
		+ \sum_{m = 0}^{n - 1} \frac{2 \M{m} \post{m} \am{m}}{(1+\M{m})^2(\post{m} \lk{m} \cdots \lk{n-1} \1{\set{X}_n})^2(\post{n} \lk{n} \1{\set{X}_{n+1}})^2}\bigg(\post{m}\lk{m}\{w_m \BFcent{m+1}{n+1}\\
		\times(\adds{n+1} \testf[n+1] 	+ \testfp[n+1])(\tstat{m}\adds{m} +\addf{m}-\tstat{m+1}\adds{m+1}) \lk{m+1}\cdots\lk{n}\testf[n+1] \}\bigg) 
		\end{multlined}\\
		& 
		\!\begin{multlined}
		+ \sum_{m=0}^{n - 1}  \frac{\post{m} \am{m}}{(1+\M{m})^2(\post{m}\lk{m}\cdots\lk{n-1}\1{\set{X}_n})^2(\post{n}\lk{n}\1{\set{X}_{n+1}})^2}\bigg(\post{m} \lk{m}((w_m-\bkm{m} w_m)\\
		\times\{ \BFcent{m+1}{n+1}(\adds{n+1}\testf[n+1] + \testfp[n+1])(\tstat{m}\adds{m}+\addf{m}-\tstat{m+1}\adds{m+1})\lk{m+1} \cdots \lk{n}\testf[n+1] \}^2)\bigg).
		\end{multlined} 
	\end{align}}
	By adding, term by term, the last two expressions and using the identity 
	\begin{equation} \label{eq:L:identity}
		(\post{m}\lk{m}\cdots\lk{n-1}\1{\set{X}_n})(\post{n}\lk{n}\1{\set{X}_{n+1}}) = \post{m}\lk{m}\cdots\lk{n}\1{\set{X}_{n+1}}, 
	\end{equation}
	we finally obtain 
	{\footnotesize
	\begin{align}
		\lefteqn{\asvar[2]{n+1}{\testf[n+1]}{\testfp[n+1]}} \\
		&
		=\frac{\Xinit \{w_{-1} \BFcent[2]{0}{n+1}(\adds{n+1} \testf[n+1] + \testfp[n+1]) \}}{(\Xinit \lk{0} \cdots \lk{n} \1{\set{X}_{n+1}})^2} \\
		&+ \sum_{m = 0}^n \post{m} \am{m}\frac{\post{m} \lk{m} \{w_m \BFcent[2]{m+1}{n+1}(\adds{n+1} \testf[n+1] + \testfp[n+1])\}}{(\post{m} \lk{m} \cdots \lk{n} \1{\set{X}_{n+1}})^2} \\
		& 
		\!\begin{multlined}
		+ \sum_{m = 0}^n \frac{\M{m} \post{m} \am{m}}{1+\M{m}}\\
		\times\sum_{\ell = 0}^m \frac{\post{\ell} \lk{\ell} ( \bkm{\ell} (\tstat{\ell} \af{\ell} + \addf{\ell} - \tstat{\ell+1} \af{\ell+1})^2 \lk{\ell+1} \cdots \lk{m} \{ \bkm{m} w_m (\lk{m+1} \cdots \lk{n} \testf[n+1])^2\} )}{(\post{\ell}\lk{\ell} \cdots \lk{m-1}\1{\set{X}_m})(\post{m}\lk{m}\cdots\lk{n}\1{\set{X}_{n+1}})^2 \prod_{k = \ell}^m (1+\M{k})} 
		\end{multlined}\\
		& 
		\!\begin{multlined}
		+ \sum_{m = 0}^n \frac{\post{m} \am{m}}{1+\M{m}}\\
		\times\sum_{\ell = 0}^m \frac{\post{\ell} \lk{\ell} ( \bkm{\ell}(\tstat{\ell} \af{\ell} + \addf{\ell} - \tstat{\ell+1} \af{\ell+1})^2 \lk{\ell+1} \cdots \lk{m}\{w_m (\lk{m+1} \cdots \lk{n} \testf[n+1])^2\} )}{(\post{\ell}\lk{\ell} \cdots \lk{m-1}\1{\set{X}_m})(\post{m}\lk{m}\cdots\lk{n}\1{\set{X}_{n+1}})^2 \prod_{k=\ell}^{m}(1+\M{k})} 
		\end{multlined}\\
		& 
		\!\begin{multlined}
		+  \sum_{m = 0}^n  \frac{2\M{m} \post{m} \am{m}}{(1+\M{m})^2(\post{m}\lk{m}\cdots\lk{n}\1{\set{X}_{n+1}})^2}\bigg(\post{m} \lk{m} \{w_m \BFcent{m+1}{n+1}(\adds{n+1}\testf[n+1] + \testfp[n+1])\\
		\times(\tstat{m}\adds{m}+\addf{m}-\tstat{m+1}\adds{m+1}) \lk{m+1} \cdots \lk{n} \testf[n+1] \}\bigg) 
		\end{multlined}\\
		& 
		\!\begin{multlined}
		+ \sum_{m = 0}^n  \frac{\post{m} \am{m}}{(1+\M{m})^2(\post{m}\lk{m}\cdots\lk{n}\1{\set{X}_{n+1}})^2}\bigg(\post{m} \lk{m}((w_m - \bkm{m} w_m) \{\BFcent{m+1}{n+1}(\adds{n+1}\testf[n+1] \\
		+ \testfp[n+1])(\tstat{m}\adds{m}+\addf{m}-\tstat{m+1}\adds{m+1})\lk{m+1}\cdots \lk{n} \testf[n+1]\}^2)\bigg) ,
		\end{multlined}
	\end{align}
	}
	which completes the induction step. It remains to establish the base case $n = 1$. Since the estimator at time zero is obtained by means of standard importance sampling,
	\begin{align}
		\asvar[2]{0}{\testf[0]}{\testfp[0]} &= \frac{\nu(w_{-1}^2\{\testf[0]\adds{0}  + \testfp[0]-\post{0}(\adds{0}\testf[0]+\testfp[0])\}^2)}{(\nu w_{-1})^2} \\
		&= \frac{\Xinit\{w_{-1}\BFcent[2]{0}{0}(\adds{0}\testf[0] + \testfp[0])\}}{(\Xinit \1{\set{X}_0})^2}. \label{eq:initial:variance}
	\end{align}
	In addition, using the recursive form \eqref{eq:prevar} of the asymptotic variance at time one, we get
	\begin{align}
		&\lefteqn{\asvar[2]{1}{\testf[1]}{\testfp[1]}
		= \post{0} \am{0} \frac{\post{0}\lk{0}(w_0\{\testf[1]\tstat{1}\adds{1}  + \testfp[1]-\post{1}(\tstat{1}\adds{1}\testf[1]+\testfp[1])\}^2)}{(\post{0}\lk{0}\1{\set{X}_{1}})^2}} \\
		&+ \M{0} \post{0} \am{0} \frac{\sqc{0}\{\lk{0}(\testf[1]^2\bkm{0}w_0)\}+\post{0} \lk{0}\{\testf[1]^2\bkm{0}(\tstat{0}\adds{0}+\addf{0}-\tstat{1}\adds{1})^2 \bkm{0} w_0\}}{(1+\M{0})^2(\post{0}\lk{0}\1{\set{X}_{1}})^2} \\
		&+ \post{0} \am{0} \frac{\sqc{0}\{\lk{0}(w_0\testf[1]^2)\} + \post{0}\lk{0}\{w_0\testf[1]^2\bkm{0}(\tstat{0}\adds{0}+\addf{0}-\tstat{1}\adds{1})^2 \}}{(1+\M{0})^2(\post{0}\lk{0}\1{\set{X}_{1}})^2} \\
		&
		\!\begin{multlined}
		+  \frac{2 \M{0} \post{0} \am{0}}{(1+\M{0})^2(\post{0} \lk{0} \1{\set{X}_1})^2} \bigg( \post{0} \lk{0}(w_0\{\tstat{1}\adds{1}\testf[1]  + \testfp[1] - \post{1}(\tstat{1}\adds{1}\testf[1]+\testfp[1])\}\\
		\times(\tstat{0} \adds{0} + \addf{0} -\tstat{1} \adds{1}) \testf[1])\bigg)
		\end{multlined}\\
		&
		\!\begin{multlined}
		+  \frac{\post{0} \am{0}}{(1+\M{0})^2 (\post{0}\lk{0}\1{\set{X}_{1}})^2}\bigg(\post{0}\lk{0}\{(w_0 - \bkm{0} w_0) \{ \tstat{1}\adds{1}\testf[1] + \testfp[1]\\
		-\post{1}(\tstat{1}\adds{1}\testf[1]+\testfp[1])+(\tstat{0}\adds{0}+\addf{0} -\tstat{1}\adds{1})\testf[1] \}^2 \}\bigg)
		\end{multlined}\\
		&+ \frac{\asvar[2]{0}{\lk{0}\testf[1]}{\lk{0}\{\addf{0}\testf[1]+\testfp[1]-\post{1}(\tstat{1}\adds{1}\testf[1]+\testfp[1])\}}}{(\post{0}\lk{0}\1{\set{X}_{1}})^2}.
	\end{align}
	Note that by \eqref{eq:initial:variance} and \eqref{eq:retro:prospective:recursion}, the last term is equal to $\Xinit\{w_{-1}\BFcent[2]{0}{1}(\adds{1}\testf[1] + \testfp[1])\}/(\Xinit \lk{0}\1{\set{X}_{1}})^2$, and by rewriting the previous expression using the retro-prospective kernels we obtain 
	\begin{multline}
		\lefteqn{\asvar[2]{1}{\testf[1]}{\testfp[1]}= \post{0} \am{0} \frac{\post{0}\lk{0}\{w_0\BFcent[2]{1}{1}(\adds{1} \testf[1] + \testfp[1])\}}{(\post{0} \lk{0}\1{\set{X}_{1}})^2}} \\
		+ \M{0} \post{0}\am{0} \frac{\post{0}\lk{0}\{\bkm{0}(\adds{0}+\addf{0}-\tstat{1}\adds{1})^2(\bkm{0} w_0) \testf[1]^2\}}{(1+\M{0})^2(\post{0}\lk{0}\1{\set{X}_{1}})^2} \\
		+ \post{0} \am{0} \frac{\post{0}\lk{0} \{ \bkm{0}(\adds{0} + \addf{0} - \tstat{1} \adds{1})^2 w_0\testf[1]^2 \}}{(1+\M{0})^2(\post{0}\lk{0}\1{\set{X}_{1}})^2} \\
		+ 2 \M{0} \post{0} \am{0} \frac{}{}\frac{\post{0}\lk{0} \{w_0\BFcent[2]{1}{1}(\adds{1} \testf[1] + \testfp[1])(\adds{0} + \addf{0} - \tstat{1} \adds{1}) \testf[1] \}}{(1+\M{0})^2 (\post{0}\lk{0}\1{\set{X}_{1}})^2} \\
		+ \post{0} \am{0} \frac{\post{0} \lk{0}((w_0-\bkm{0} w_0) \{\BFcent[]{1}{1}(\adds{1} \testf[1] + \testfp[1])+(\adds{0}+\addf{0} -\tstat{1}\adds{1})\testf[1] \}^2)}{(1+\M{0})^2 (\post{0}\lk{0}\1{\set{X}_1})^2} \\
		+ \frac{\Xinit\{w_{-1}\BFcent[2]{0}{1}(\adds{1}\testf[1] + \testfp[1])\}}{(\Xinit \lk{0}\1{\set{X}_{1}})^2}, 
	\end{multline}
	%\par}
	which, recalling that $\tstat{0} \adds{0} = \adds{0}$, corresponds to \eqref{asvar} for $n = 1$. The proof is complete. 
\end{proof}

%%%% UNIFORM VARIANCE BOUND

\subsection{Time linear variance bounds}

In this part we will derive an $\ordo(n)$ bound on the asymptotic variance $\sigma_n^2(\adds{n})$ in Corollary~\ref{corollary:1} in the case where the increments $(\addf{n})_{n \in \nset}$ of the additive functionals can be uniformly bounded in $n$. The analysis will be carried through under the following strong mixing assumptions, which typically require the state spaces $(\set{X}_n)_{n \in \nset}$ to be compact sets; see, \eg, \citet[Section~4]{delmoral:2004} and \citet[Section~4]{cappe:moulines:ryden:2005}. The numerical stability of AdaSmooth in the case of a deterministic selection and backward-sampling schedule will then, in the next section, be established by, first, formulating AdaSmooth equivalently as an algorithm of the same type as Algorithm~\ref{algo:adasmsupp}, but when the latter is operating on an extended path-space model; then, second, it will be shown that Assumption~\ref{assum:compact} implies that the strong mixing assumptions are satisfied also for the extended model. 

\begin{assumption2}\label{assum:2}
	\mbox{}
	\begin{itemize}
		\item[(i)] There exist constants $0 < \hklow < \hkup < \infty$ such that for all $n \in \nset $ and all $(x,x') \in \set{X}_n \times \set{X}_{n+1}$,
		%\begin{equation}
			$\hklow \le \ld{n}(x, x')\le \hkup$.
		%\end{equation}
		\item[(ii)] There exist positive constants $\mdup$  and $\amup$ and such that for all $n \in \nset$, $\supn{w_n}\le \mdup$ and $\supn{\am{n}}\le \amup$. In addition, $\supn{w_{-1}} \le \mdup$.  
	\end{itemize}
\end{assumption2}
Under Assumption~\ref{assum:2} we may, without loss of generality, assume that each reference measure $\mu_n$ is a probability measure. As a consequence, for all $n \in \nset$, $\lk{n} \1{\set{X}_{n+1}}(x_n) \geq \hklow$ for all $x_n\in\set{X}_n$. Moreover, under Assumption~\ref{assum:2} we define $\mr \eqdef 1 - \hklow / \hkup$. 

\begin{theorem}\label{thm:boundsupp}
	Let Assumption~\ref{assum:2} hold. Then for all additive functionals $(\adds{n})_{n \in \nset}$ of form \eqref{eq:adds} for which there exists a positive constant $\hbd$ such that for all $n\in \nsetpos$, $\supn{\addf{n}}\le \hbd$ and $\supn{\adds{0}+\addf{0}}\le \hbd$, %and Assumption \ref{assum:3} hold. Then
	\begin{multline}\label{asvar_bound}
		\limsup_{n \to \infty} \frac{1}{n} \sigma_n^2(\adds{n}) \le	\hbd^2 \frac{\amup \mdup }{\hklow(1-\mr)^4} \bigg( \frac{108+ \mdup\mr(3 - 2 \mr)(1 + \mr)}{2 \mr^2} \\+\frac{(3 - 2 \mr)^2}{1 - \mr} \lim_{n \to \infty} \frac{1}{n} \sum_{m=0}^{n-1} \sum_{\ell = 0}^m \prod_{k=\ell}^m (1 + \M{k})^{-1} \bigg).
	\end{multline}
\end{theorem}
\begin{proof}
	We write 
	\begin{equation}
		\sigma_n^2(\adds{n}) = A_n + B_n + C_n + D_n + E_n + F_n,
	\end{equation}
	where
	{\footnotesize
	\begin{align}
		&A_n \eqdef \frac{\Xinit(w_{-1} \BFcent[2]{0}{n}\adds{n})}{(\Xinit \lk{0} \cdots \lk{n-1} \1{\set{X}_n})^2}, \\
		&B_n \eqdef \sum_{m=0}^{n - 1} \post{m} \am{m} \frac{\post{m} \lk{m} (w_m \BFcent[2]{m+1}{n} \adds{n})}{(\post{m} \lk{m} \cdots \lk{n-1} \1{\set{X}_n})^2}, \\
		& 
		\!\begin{multlined}
		C_n \eqdef\sum_{m=0}^{n - 1} \frac{\M{m} \post{m} \am{m}}{1+\M{m}}\\
		\times\sum_{\ell = 0}^m \frac{\post{\ell} \lk{\ell} (\bkm{\ell}(\tstat{\ell} \af{\ell} + \addf{\ell} - \tstat{\ell+1} \af{\ell+1})^2 \lk{\ell+1} \cdots \lk{m} \{\bkm{m} w_m (\lk{m+1} \cdots \lk{n-1} \1{\set{X}_n})^2\})}{(\post{\ell}\lk{\ell} \cdots \lk{m-1}\1{\set{X}_m})(\post{m}\lk{m}\cdots\lk{n-1}\1{\set{X}_n}  )^2 \prod_{k = \ell}^m (1+\M{k})}, 
		\end{multlined}\\
		&
		\!\begin{multlined}
		D_n \eqdef \sum_{m=0}^{n-1} \frac{\post{m} \am{m}}{1+\M{m}}\\
		\times\sum_{\ell = 0}^m \frac{\post{\ell} \lk{\ell}(\bkm{\ell}(\tstat{\ell} \af{\ell} + \addf{\ell} - \tstat{\ell+1} \af{\ell+1})^2 \lk{\ell+1} \cdots \lk{m}\{w_m (\lk{m+1} \cdots \lk{n-1} \1{\set{X}_n})^2\})}{(\post{\ell}\lk{\ell} \cdots \lk{m-1}\1{\set{X}_m})(\post{m}\lk{m}\cdots\lk{n-1}\1{\set{X}_n})^2 \prod_{k=\ell}^m (1+\M{k})}, 
		\end{multlined}\\
		&
		\!\begin{multlined}
		E_n \eqdef 2 \sum_{m=0}^{n-1} \M{m} \post{m} \am{m}\\
		\times\frac{\post{m} \lk{m}\{w_m \BFcent{m+1}{n} \adds{n}(\tstat{m} \adds{m} + \addf{m} - \tstat{m+1} \adds{m+1}) \lk{m+1} \cdots \lk{n-1} \1{\set{X}_n}\}}{(1+\M{m})^2 (\post{m} \lk{m} \cdots \lk{n-1} \1{\set{X}_n})^2}, 
		\end{multlined}\\
		&
		\!\begin{multlined}
		F_n \eqdef \sum_{m=0}^{n-1}  \frac{\post{m} \am{m}}{(1+\M{m})^2(\post{m}\lk{m}\cdots\lk{n-1}\1{\set{X}_n})^2}\bigg(\post{m} \lk{m} ((w_m - \bkm{m} w_m) \{\BFcent{m+1}{n}\adds{n}\\
		+(\tstat{m}\adds{m}+\addf{m}-\tstat{m+1}\adds{m+1})\lk{m+1}\cdots\lk{n-1}\1{\set{X}_n}\}^2)\bigg).
		\end{multlined}
	\end{align}
	}
	In the following we will use techniques developed by \citet[Section~D]{gloaguen:lecorff:olsson:2021} to bound each of the terms $A_n$--$F_n$. For each $n \in \nsetpos$ and $k \in \nset$ such that $k < n$, let 
	\begin{equation} \label{eq:hext}
	\addfext{k}{n} : \set{X}_0 \times \cdots \times \set{X}_n \ni x_{0:n} \mapsto
	\begin{cases}
	\addf{k}(x_k, x_{k+1}) & \mbox{for $k \in\intvect{1}{n - 1}$}, \\
	\adds{0}(x_0) + \addf{0}(x_0, x_1) & \mbox{for $k = 0$},  
	\end{cases}
	\end{equation}
	denote the extensions of $\addf{k}$ and $\adds{0}+\addf{0}$ to $\set{X}_0\times\cdots\times\set{X}_{n}$. Now, note that 
	\begin{align}
		\post{0:n}\adds{n}=\frac{\post{m}\BF{m}{n}\adds{n}}{\post{m}\BF{m}{n}\1{\set{X}_{n}}};
	\end{align}
	thus, using \eqref{eq:hext} %and \eqref{eq:hext2} 
	we may write, for every $m \in \intvect{0}{n-1}$ and $x_m \in \set{X}_m$,
	\begin{align}
		\frac{\BFcent{m}{n}\adds{n}(x_m)}{\lk{m}\cdots\lk{n-1}\1{\set{X}_n}(x_m)}=\sum_{k=0}^{n-1}\bigg(\frac{\BF{m}{n}\addfext{k}{n}(x_m)}{\BF{m}{n}\1{\set{X}_n}(x_m)}-\frac{\post{m}\BF{m}{n}\addfext{k}{n}}{\post{m}\BF{m}{n}\1{\set{X}_n}}\bigg).
	\end{align}
	Applying Lemma~D.3 in \citet{gloaguen:lecorff:olsson:2021} yields 
	\begin{align}
		\supn{\BFcent{m}{n} \adds{n}}&\le \supn{\lk{m} \cdots \lk{n-1} \1{\set{X}_n}} \sum_{k=0}^{n-1} \left \lvert \frac{\delta_{x_m} \BF{m}{n} \addfext{k}{n}}{\delta_{x_m} \BF{m}{n} \1{\set{X}_n}} - \frac{\post{m} \BF{m}{n} \addfext{k}{n}}{\post{m} \BF{m}{n} \1{\set{X}_n}} \right \rvert
		\\&\le \hbd \supn{\lk{m}\cdots \lk{n - 1} \1{\set{X}_n}} \sum_{k = 0}^{n - 1} \mr^{|k-m|-1}. \label{eq:bound:retro-prospective}
	\end{align}
	
	%%%% Term A
	
	\subsubsection*{Term $A_n$}
	Using \eqref{eq:bound:retro-prospective} we obtain  
	\begin{align}
		\supn{\BFcent{0}{n}\adds{n}}\le\hbd\supn{\lk{0}\cdots\lk{n-1}\1{\set{X}_n}}\frac{1-\mr^{n}}{\mr(1-\mr)},
	\end{align}
	implying that
	\begin{equation}
		A_n\le \hbd^2\supn{\lk{0}\cdots\lk{n-1}\1{\set{X}_n}}^2\left(\frac{1-\mr^{n}}{\mr(1-\mr)}\right)^2\frac{\Xinit\left(w_{-1}\right)}{(\Xinit\lk{0}\cdots\lk{n-1}\1{\set{X}_n})^2}.
	\end{equation}
	Now, under Assumption \ref{assum:2}, for all $x\in \set{X}_0$,
	\begin{equation}
		\hklow\refM{1} \lk{1}\cdots\lk{n-1}\1{\set{X}_n} \le \lk{0} \cdots \lk{n-1} \1{\set{X}_n}(x) \le \hkup \refM{1} \lk{1} \cdots \lk{n-1} \1{\set{X}_n}, 
	\end{equation}
	which, recalling that $\Xinit\lk{0}\cdots\lk{n-1}\1{\set{X}_n}=(\Xinit\1{\set{X}_0})\post{0}\lk{0}\cdots\lk{n-1}\1{\set{X}_n}$, yields that
	\begin{multline}
		\frac{\Xinit w_{-1}\supn{\lk{0}\cdots\lk{n-1}\1{\set{X}_n}}^2}{(\Xinit \1{\set{X}_0})^2(\post{0}\lk{0}\cdots\lk{n-1}\1{\set{X}_n})^2}=\frac{\post{0}w_{-1}\supn{\lk{0}\cdots\lk{n-1}\1{\set{X}_n}}^2}{(\Xinit \1{\set{X}_0})(\post{0}\lk{0}\cdots\lk{n-1}\1{\set{X}_n})^2}\le \frac{\mdup}{\Xinit \1{\set{X}_0}}\left(\frac{\hkup}{\hklow}\right)^2\\=\frac{\mdup}{\Xinit \1{\set{X}_0}(1-\mr)^2}.
	\end{multline}
	Hence, $A_n$ is uniformly bounded and $\limsup_{n\rightarrow\infty}A_n/n =0$. 
	
	%%%% Term B
	
	\subsubsection*{Term $B_n$}
	We turn to $B_n$. Since $\post{m}\lk{m}\cdots\lk{n-1}\1{\set{X}_n} = (\post{m} \lk{m} \1{\set{X}_{m+1}})(\post{m+1}\lk{m+1}\cdots\lk{n-1}\1{\set{X}_n})$ and, by \eqref{eq:bound:retro-prospective}, for $m\in\intvect{0}{n-1}$,
	\begin{equation}\label{BF_ineq}
		\supn{\BFcent{m+1}{n}\adds{n}}\le \hbd\supn{\lk{m+1} \cdots \lk{n-1} \1{\set{X}_n}}\sum_{k=0}^{n-1}\mr^{|k-m-1|-1}, 
	\end{equation}
	it follows that
	\begin{equation}
		B_n\le \hbd^2\amup\sum_{m=0}^{n-1} \frac{\post{m}\lk{m}w_m \supn{\lk{m+1}\cdots\lk{n-1}\1{\set{X}_n}}^2}{(\post{m}\lk{m}\1{\set{X}_{m+1}})^2(\post{m+1}\lk{m+1}\cdots\lk{n-1}\1{\set{X}_n})^2} \left( \sum_{k=0}^{n-1}\mr^{|k-m-1|-1} \right)^2. 
	\end{equation}
	Under Assumption~\ref{assum:2}, for all $m\in\intvect{0}{n-1}$ and $x \in \set{X}_{m + 1}$, 
	\begin{equation}
		\hklow \mu_{m+2} \lk{m+2} \cdots \lk{n-1} \1{\set{X}_n} \le \lk{m+1} \cdots \lk{n-1} \1{\set{X}_n}(x) \le \hkup \mu_{m+2} \lk{m+2} \cdots \lk{n-1} \1{\set{X}_n},
	\end{equation}
	implying that
	\begin{equation}\label{L_ineq}
		\frac{\supn{\lk{m+1}\cdots\lk{n-1}\1{\set{X}_n}}}{\post{m+1} \lk{m+1} \cdots \lk{n-1} \1{\set{X}_n}}\le \frac{1}{(1-\mr)}.
	\end{equation}
	Consequently, 
	\begin{equation}
		\frac{\post{m} \lk{m} w_m \supn{\lk{m+1} \cdots \lk{n-1}\1{\set{X}_n}}^2}{(\post{m} \lk{m} \1{\set{X}_{m+1}})^2(\post{m+1}\lk{m+1}\cdots\lk{n-1}\1{\set{X}_n})^2} \le \frac{\mdup}{\hklow(1-\mr)^2},
	\end{equation}
	and since 
	\begin{equation}
		\sum_{m=0}^{n - 1} \left(\sum_{k=0}^{n - 1} \mr^{|k-m-1|-1} \right)^2 \le \frac{4n}{\mr^2(1-\mr)^2},
	\end{equation}
	we may conclude that 
	\begin{equation}
		\limsup_{n \to \infty}\frac{1}{n} B_n \le \hbd^2 \frac{4 \amup \mdup}{\hklow \mr^2(1-\mr)^4}.
	\end{equation}
	
	%%%% Terms C and D
	
	\subsubsection*{Terms $C_n$ and $D_n$}
	
	Due to their similarity, $C_n$ and $D_n$ are treated in the same way. In order to bound $C_n$, we proceed like, using \eqref{L_ineq},
	{\footnotesize\begin{align}
		\lefteqn{\frac{\post{\ell}\lk{\ell}(\bkm{\ell}(\tstat{\ell} \af{\ell} + \addf{\ell} - \tstat{\ell + 1}\af{\ell+1})^2 \lk{\ell+1} \cdots \lk{m} \{\bkm{m} w_m (\lk{m+1} \cdots \lk{n-1} \1{\set{X}_n})^2\})}{(\post{\ell} \lk{\ell} \cdots \lk{m-1}\1{\set{X}_m})(\post{m}\lk{m}\cdots\lk{n-1}\1{\set{X}_n})^2}} \\
		&\le \frac{\post{\ell+1}\{ \bkm{\ell}(\tstat{\ell}\af{\ell} + \addf{\ell} - \tstat{\ell+1}\af{\ell+1})^2 \} \supn{\lk{\ell+1} \cdots \lk{m-1}\1{\set{X}_m}} \mdup \supn{\lk{m+1} \cdots \lk{n-1} \1{\set{X}_n}}^2}{(\post{\ell+1} \lk{\ell+1} \cdots \lk{m-1}\1{\set{X}_m})(\post{m}\lk{m}\1{\set{X}_{m+1}}  )(\post{m+1}\lk{m+1}\cdots\lk{n-1}\1{\set{X}_n}  )^2} \\
		&\le \post{\ell+1}\{\bkm{\ell}(\tstat{\ell} \adds{\ell} + \addf{\ell} - \tstat{\ell+1} \adds{\ell+1})^2\} \frac{\mdup}{\hklow (1 - \mr)^3}.
	\end{align}}
	Then, since $\bkm{\ell}(\tstat{\ell} \adds{\ell} + \addf{\ell}) = \tstat{\ell+1} \adds{\ell+1}$ and $\tstat{\ell} \adds{\ell} = \BF{\ell}{\ell} \adds{\ell}$, we can write 
	\begin{align}
		\lefteqn{\tstat{\ell} \adds{\ell}(x_\ell) + \addf{\ell}(x_\ell, x_{\ell+1}) - \tstat{\ell+1} \adds{\ell+1}(x_{\ell+1})}\\
		&= \tstat{\ell} \adds{\ell}(x_\ell)-\bkm{\ell} \tstat{\ell} \adds{\ell}(x_{\ell+1}) + \addf{\ell}(x_\ell, x_{\ell+1}) - \bkm{\ell} \addf{\ell}(x_{\ell + 1}) \\
		&= \sum_{k = 0}^{\ell - 1}(\BF{\ell}{\ell}\addfext{k}{\ell}(x_{\ell}) - \bkm{\ell} \BF{\ell}{\ell}\addfext{k}{\ell}(x_{\ell+1}))+\addf{\ell}(x_\ell, x_{\ell+1}) - \bkm{\ell} \addf{\ell}(x_{\ell+1}),
	\end{align}
	and using again Lemma~D.3 in \citet{gloaguen:lecorff:olsson:2021} yields
	\begin{multline}\label{eq:T_ineq}
		\supn{\tstat{\ell} \adds{\ell} + \addf{\ell} - \tstat{\ell+1} \adds{\ell+1}} \le \hbd \left( \sum_{k=0}^{\ell-1}\mr^{\ell-k-1}+2 \right) \\\le \hbd \left( \frac{1}{1-\mr}+2 \right).
	\end{multline}
	This implies	
	$$
		\limsup_{n \to \infty} \frac{1}{n}(C_n+D_n) 
		\leq  \hbd^2 \frac{\amup \mdup (3 - 2 \mr)^2}{\hklow (1 - \mr)^5} \lim_{n\rightarrow\infty} \frac{1}{n} \sum_{m=0}^{n - 1} \sum_{\ell = 0}^m \prod_{k=\ell}^m (1+\M{k})^{-1}.
	$$
	Note that by monotonicity, the limit on the right-hand side of the previous inequality either exists or is infinite.   
	 
	 %%%% Term E
	 
	\subsubsection*{Term $E_n$}
		
	In order to bound $E_n$ (and later $F_n$) we may reuse \eqref{BF_ineq}, \eqref{L_ineq} and \eqref{eq:T_ineq}. Write 
	\begin{multline*}
		\left| \frac{\post{m}\lk{m}\{w_m \BFcent{m+1}{n}\adds{n}(\tstat{m}\adds{m}+\addf{m}-\tstat{m+1}\adds{m+1}) \lk{m+1}\cdots\lk{n-1}\1{\set{X}_n} \}}{(\post{m}\lk{m}\1{\set{X}_{m+1}})^2(\post{m+1}\lk{m+1}\cdots\lk{n-1}\1{\set{X}_n})^2} \right| \\
		\le \hbd^2 \frac{ \mdup(3-2\mr) \supn{\lk{m+1}\cdots\lk{n-1}\1{\set{X}_n}}^2}{\hklow (1-\mr)^3 (\post{m+1}\lk{m+1}\cdots\lk{n-1}\1{\set{X}_n})^2} \sum_{k=0}^{n-1}\mr^{|k-m-1|-1}.
	\end{multline*}
	Thus, by noting that $\M{m}/(1+\M{m})^2\le 4^{-1}$ for all $m$,
	\begin{equation}
		|E_n|\le \hbd^2\frac{\amup \mdup(3-2\mr)}{2 \hklow (1-\mr)^3}\sum_{m=0}^{n-1}\bigg(\frac{1-\mr^{m+2}}{\mr(1-\mr)}+\frac{1-\mr^{n-m-2}}{1-\mr}\bigg)
	\end{equation}
	we obtain, by Ces\`{a}ro summation,  
	\begin{align}
		\limsup_{n \to \infty} \frac{1}{n}|E_n| &\le \hbd^2 \frac{\amup \mdup(3 - 2 \mr)}{2 \hklow \mr(1 - \mr)^4}\lim_{n \to \infty}\frac{1}{n}\left( n(1+\mr) - \sum_{m=0}^{n-1}(\mr^{m+2}+\mr^{n-m-1}) \right) \\
		&= \frac{\amup \mdup (3 - 2 \mr)(1 + \mr)}{2 \hklow \mr (1 - \mr)^4}.
	\end{align}
	
	%%%% Term F
	
	\subsubsection*{Term $F_n$}
	
	In order to bound $F_n$, note that
	\begin{multline}
		\footnotesize{\lefteqn{\bigg|\frac{\post{m} \lk{m} ((w_m - \bkm{m} w_m) \{ \BFcent{m+1}{n} \adds{n}+(\tstat{m} \adds{m} + \addf{m} - \tstat{m+1} \adds{m+1}) \lk{m+1} \cdots\lk{n-1} \1{\set{X}_n}\}^2)}{(\post{m} \lk{m} \cdots \lk{n-1} \1{\set{X}_n})^2} \bigg|}}\\
		\le \hbd^2 \frac{2\mdup \supn{\lk{m+1}\cdots\lk{n-1}\1{\set{X}_n}}^2}{\hklow(\post{m+1}\lk{m+1}\cdots\lk{n-1}\1{\set{X}_n})^2} \left(\sum_{k=0}^{n - 1} \mr^{|k-m-1|-1}+(3-2\mr)/(1-\mr) \right)^2 \\
		\le \hbd^2\frac{2 \mdup}{\hklow(1-\mr)^2}\left(\frac{1-\mr^{m+2}-\mr^{n-m-1}+4\mr-2\mr^2}{\mr(1-\mr)}\right)^2.
	\end{multline}
	Thus, by observing that $1/(1+\M{m})^2 \le 1$ for all $m$ and using again Ces\`{a}ro summation,
	\begin{align}
		\limsup_{n \to \infty}\frac{1}{n}|F_n| &\le \hbd^2 \frac{2 \amup \mdup}{\hklow \mr^2(1-\mr)^4} \lim_{n \to \infty} \frac{1}{n}\sum_{m=0}^{n-1} \left( 1-\mr^{m+2}-\mr^{n-m-1}+4\mr-2\mr^2 \right)^2 \\
		& \le \hbd^2\frac{50 \amup \mdup}{\hklow\mr^2(1-\mr)^4}.
	\end{align}
	
	Finally, summing up the obtained bounds on $A_n$--$F_n$ yields 
	\begin{align}
		\limsup_{n\rightarrow\infty}&\frac{1}{n}\sigma_n^2(\adds{n})\\
		&
		\!\begin{multlined}
			\le  \hbd^2\frac{4 \amup \mdup}{\hklow\mr^2(1-\mr)^4} 
			+ \hbd^2\frac{\amup \mdup(3-2\mr)^2}{\hklow (1-\mr)^5} \lim_{n \to \infty}\frac{1}{n} \sum_{m=0}^{n-1}\sum_{\ell = 0}^{m}\prod_{k=\ell}^{m}(1+\M{k})^{-1} \\
			+ \hbd^2\frac{\amup \mdup(3-2\mr)(1+\mr)}{2\hklow \mr(1-\mr)^4}
			+ \hbd^2\frac{50\amup \mdup}{\hklow\mr^2(1-\mr)^4}
		\end{multlined} \\
		&
		\!\begin{multlined}
		=\hbd^2 \frac{\amup \mdup}{\hklow(1-\mr)^4}\bigg(\frac{108+ \mdup\mr(3-2\mr)(1+\mr)}{2\mr^2}\\\hspace{3cm}+\frac{ (3-2\mr)^2}{ 1-\mr}\lim_{n\rightarrow\infty}\frac{1}{n}\sum_{m=0}^{n-1}\sum_{\ell = 0}^{m}\prod_{k=\ell}^{m}(1+\M{k})^{-1}\bigg).
		\end{multlined}
	\end{align}
	The proof is complete. 
\end{proof}

\section{Proofs of Theorems 3.1--3.3}
\label{sec:detsel}

\subsection{Model extension}
The aim of this section is to show that the results obtained in Section~\ref{sec:convstab} can be used directly to establish the consistency and asymptotic normality of AdaSmooth in the case where the selection and backward-sampling time points are governed by deterministic rules $(\res{n})_{n\in\nset}$ and $(\M{n})_{n\in\nset}$, respectively. The idea is to consider an extended version of the model of Section~\ref{sec:introsupp} with states given by paths of a varying length determined by the selection schedule $(\res{n})_{n\in\nset}$. In the following, this construction will be carried through in detail. Recall that we defined, in Section~\ref{sec:detsched}, the resampling times $(n_m)_{m \in \nset}$, as $n_m \eqdef \min\{n \in \nset : \sum_{k=0}^n \res{k} = m + 1\}$. By convention, $n_{-1} \eqdef -1$. Then we introduce the sequence $(\Xp{m}, \Xpalg{m})_{m \in \nset}$ of measurable spaces, where $\Xp{m} \eqdef \set{X}_{n_{m-1}+1 }\times \set{X}_{n_{m-1}+2} \times \cdots \times \set{X}_{n_{m}}$ and $\Xpalg{m} \eqdef \alg{X}_{n_{m-1}+1}\varotimes \alg{X}_{n_{m-1}+2}\varotimes \cdots\varotimes \alg{X}_{n_m}$. In the following we will use boldface to indicate that a quantity is related to such a path space; \eg, we let $\xvec{m} \eqdef x_{n_{m-1}+1:n_m}$ indicate a generic element in $\Xp{m}$ and define the projection 
\begin{equation}
	\proj_m : \Xp{m} \ni \xvec{m} \mapsto x_{n_m} \in \set{X}_{n_m}. 
\end{equation}
The extended model on $(\Xp{m}, \Xpalg{m})_{m \in \nset}$ that we will consider is governed by multi-step unnormalized transition kernels $(\lkm{m})_{m \in \nset}$ induced by products of the single-step transition kernels $(\lk{n})_{n\in\nset}$ in Section~\ref{sec:model} as follows. For each $m \in \nset$, define 
\begin{align}
	\lkm{m} \boldsymbol{h}(\xvec{m}) \eqdef \lk{n_m} \varotimes \lk{n_m+1} \varotimes \cdots \varotimes \lk{n_{m+1}-1} \boldsymbol{h} (\proj_{n_m}(\xvec{m})), 
\end{align} 
$(\xvec{m},\boldsymbol{h}) \in \Xp{m} \times \bmf{\Xpalg{m + 1}}$. It follows that every $\lkm{m}$ has a density
\begin{equation}
	\ldm{m}(\xvec{m}, \xvec{m+1}) 
	\eqdef \prod_{k = n_{m}}^{n_{m+1} - 1} \ld{k}(x_k, x_{k+1}), \quad (\xvec{m}, \xvec{m+1}) \in \Xp{m} \times \Xp{m + 1}, 
\end{equation}
with respect to the product reference measure $\boldsymbol{\mu}_{m + 1} \eqdef \bigotimes_{k = n_{m} + 1}^{n_{m+1}} \mu_k$. Note that $\lkm{m}$, as well as its density $\ldm{m}$,  depends only on $x_{n_m}$ and is constant with respect to the previous states. If resampling is performed systematically, then $n_m =m$ for all $m$, and consequently, $\lkm{m} = \lk{m}$ in that case. The model is equipped with an initial distribution $\Xinitm \eqdef \Xinit \varotimes \lk{0} \varotimes \cdots \varotimes \lk{n_0-1}$
on $\Xpalg{0}$, and we will, abusing notations, denote its density 
$$
\Xinitm(\xvec{0}) = \prod_{k = 0}^{n_0 - 1} \ld{k}(x_k, x_{k+1}), \quad \xvec{0} \in \Xp{0}, 
$$
with respect to $\boldsymbol{\mu}_0$ by the same symbol.  

 So far the quantities governing our extended model. In accordance with 
\eqref{eq:filter:distribution}, we may now define the extended marginals $(\postm{m})_{m \in \nset}$, where for each $m$, 
\begin{equation}
	\postm{m} \eqdef 
	\frac{\Xinitm \lkm{0} \cdots \lkm{m-1}}{\Xinitm \lkm{0} \cdots \lkm{m-1} \1{\Xp{m}}}.  
\end{equation}
Under the convention that $\lkm{m} \cdots \lkm{n} = \operatorname{id}$ if $m > n$, $\postm{0} = \Xinitm / \Xinitm \1{\Xp{0}}$. Note that with this definition, each $\postm{m}$ is the restriction of the joint-smoothing distribution $\post{0:n_m}$ to $\Xpalg{m}$. Defining also extended joint-smoothing distributions $(\postm{0:m})_{m \in \nset}$ in accordance with \eqref{eq:joint:smoothing:distribution} yields simply that for each $m$, $\postm{0:m}$ is a probability distribution on $\Xpalg{0} \varotimes \cdots \varotimes \Xpalg{m}$ simply determined by  
$\postm{0:m} = \post{0:n_m}$. We may also define the backward kernels $(\bkmm{m})_{m \in \nset}$, where $\bkmm{m}$ is the reversed kernel of $\lkm{m}$ with respect to $\postm{m}$, given by, for $(\xvec{m+1}, \boldsymbol{h}) \in \Xp{m + 1} \times \bmf{\Xpalg{m}}$, 
\begin{align}\label{eq:backker}
	\bkmm{m} \boldsymbol{h}(\xvec{m+1})&\eqdef \frac{\int \boldsymbol{h}(\xvec{m},\xvec{m+1}) \ldm{m}(\xvec{m},\xvec{m+1}) \, \postm{m}(d\xvec{m})}{\int \ldm{m}(\xvec{m},\xvec{m+1}) \, \postm{m}(d\xvec{m})} \\
	&=\frac{\int \boldsymbol{h}(\xvec{m},\xvec{m+1})\ld{n_m}(x_{n_m}, x_{n_m+1}) \, \postm{m}(d\xvec{m})}{\int \ld{n_m}(x_{n_m}, x_{n_m+1}) \, \postm{m}(d\xvec{m})}.
\end{align}
Note that each backward kernel $\bkmm{m}$ depends, as expected, only on $x_{n_m+1} \in \set{X}_{n_m+1}$ rather than the whole path $\xvec{m + 1} \in \Xp{m + 1}$. On the basis of the extended backward kernels, we may define, for every $m \in \nset$, the Markov kernel
\begin{align}
	\tk{m} \eqdef 
	\begin{cases}
		\bkmm{m - 1} \varotimes \cdots \varotimes \bkmm{0}\quad &\mbox{for $m \in \nsetpos$}\\
		\operatorname{id} &\mbox{for $m = 0$},
	\end{cases}
\end{align}
on $\Xp{m} \times (\Xpalg{0} \varotimes \cdots \varotimes \Xpalg{m - 1})$, satisfying, by \citet[Lemma~2.2]{gloaguen:lecorff:olsson:2021}, $\post{0:n_m} = \postm{0:m} = \postm{m} \tk{m}$. In addition, in analogy with \eqref{eq:retro:prospective} we introduce extended retro-prospective kernels given by, for $(\xvec{k}, \boldsymbol{h}) \in \Xp{k} \times \bmf{\Xpalg{0} \varotimes \cdots \varotimes \Xpalg{m}}$, 
\begin{align}
	\BFm{k}{m}\boldsymbol{h}(\xvec{k}) &\eqdef \iint \boldsymbol{h}(\xvec{0:m}) \, \tk{k}(\xvec{k},d\xvec{0:k-1}) \, \lkm{k} \cdots \lkm{m-1}(\xvec{k}, d\xvec{k+1:m}), \\
	\BFcentm{k}{m}\boldsymbol{h}(\xvec{k}) &\eqdef \BFm{k}{m}(\boldsymbol{h} - \postm{0:m}\boldsymbol{h})(\xvec{k}).
\end{align}
We also introduce extended versions $(\bigadds{m})_{m \in \nset}$ of the given additive functionals $(\adds{n})_{n \in \nset}$ (in the form (3)) by letting, for $m \in \nset$, 
\begin{align}
\bigaddf{m}(\xvec{m},\xvec{m+1}) &\eqdef \sum_{k=n_m}^{n_{m+1}-1} \addf{k}(x_k, x_{k+1}), \quad (\xvec{m}, \xvec{m + 1}) \in \Xp{m} \times \Xp{m + 1}, \\
\bigadds{0}(\xvec{0}) &\eqdef \adds{n_0}(\xvec{0}), \quad \xvec{0} \in \Xp{0},   
\end{align}
and, recursively, 
\begin{align}
	\bigadds{m+1}(\xvec{0:m+1}):=\bigadds{m}(\xvec{0:m})+\bigaddf{m}(\xvec{m},\xvec{m+1}), 
\end{align}
so that $\bigadds{m+1}(\xvec{0:m+1}) = \adds{n_{m+1}}(x_{0:n_{m+1}})$. Note that $\bigaddf{m}$ depends on $\xvec{m}$ only through the first state $x_{n_m}$, while it is constant with respect to the previous states. 

Similar extensions can be made for the particle generation mechanisms in the APF. More precisely, for every $m \in \nset$, let 
$$ 
	\hkm{m} \boldsymbol{h}(\xvec{m}) \eqdef \hk_{n_m} \varotimes \hk_{n_m+1} \varotimes \cdots \varotimes \hk_{n_{m+1}-1} \boldsymbol{h} (\proj_{n_m}(\xvec{m})),
$$
$ (\xvec{m}, \boldsymbol{h}) \in \Xp{m} \times \bmf{\Xpalg{m + 1}} $, be a Markov proposal kernel on $\Xp{m} \times \Xpalg{m + 1}$, having the density 
\begin{equation}
	\hdm{m}(\xvec{m}, \xvec{m+1}) 
	\eqdef \prod_{k = n_{m}}^{n_{m+1} - 1} \hd_k(x_k, x_{k+1}), \quad (\xvec{m}, \xvec{m+1}) \in \Xp{m} \times \Xp{m + 1}, 
\end{equation}
with respect to the reference measure $\boldsymbol{\mu}_{m + 1}$. Also the initial proposal $\nu$ is extended analogously, \ie, by defining $\boldsymbol{\nu} \eqdef \nu \varotimes \hk_0 \varotimes \cdots \varotimes \hk_{n_0 - 1}$ 
on $\Xpalg{0}$, having a probability density function $\boldsymbol{\nu}(\xvec{0}) = \nu(x_0) \prod_{k = 0}^{n_0 - 1} \hd_k(x_k, x_{k+1})$, $\xvec{0} \in \Xp{0}$, abusing again notations. 

Algorithm~\ref{algo:adasmoothdet} below is obtained  by casting casting AdaSmooth with systematic resampling, Algorithm~\ref{algo:adasmsupp}, into the extended model described above. 

\begin{algorithm}[H]
	\caption{AdaSmooth with systematic resampling in the extended model.}
	\begin{algorithmic}[1]\label{algo:adasmoothdet}
		\REQUIRE $(\epartm{m}{i}, \wgtm{m}{i}, \tstatm[i]{m})_{i=1}^\N$, $\M{n_m}$.
		\STATE run $(\epartm{m + 1}{i}, \wgtm{m}{i}, \I{m+1}{i})_{i=1}^\N \leftarrow \apf((\epartm{m}{i}, \wgtm{m}{i})_{i=1}^\N)$;
		\FOR{$i=1\rightarrow \N$}
		\IF{$\M{n_m} > 0$}
		\FOR{$j = 1 \to \M{n_m}$}
		\STATE draw $\bi{m+1}{i}{j} \sim \catdist((\wgtm{m}{\ell} \ldm{m}(\epartm{m}{\ell}, \epartm{m + 1}{i}))_{\ell = 1}^\N)$;
		\ENDFOR
		\ENDIF
		\STATE set $\tstatm[i]{m+1} \leftarrow   \dfrac{\tstatm[\I{m + 1}{i}]{m}+ \bigaddf{m}(\epartm{m}{\I{m+1}{i}}, \epartm{m + 1}{i}) +\sum_{j=1}^{\M{n_m}} \big( \tstatm[\bi{m + 1}{i}{j}]{n_m}+\bigaddf{m}(\epartm{m}{\bi{m+1}{i}{j}},\boldsymbol{\xi}_{m+1}^{i}) \big)}{1 + \M{n_m}}$;
		\ENDFOR
		\RETURN $(\boldsymbol{\xi}_{m+1}^i,\wgtm{m+1}{i},\tstatm[i]{m+1})_{i=1}^\N$.
	\end{algorithmic}
\end{algorithm}
\begin{algorithm}[htb]
	\caption{APF with systematic selection in the extended model.}
	\begin{algorithmic}[1] \label{algo:apf}
		\REQUIRE $(\epartm{m}{i},\wgtm{m}{i})_{i=1}^\N$.
		\FOR{$i=1\rightarrow\N$}
		\STATE draw $\I{m+1}{i} \sim \catdist((\wgtm{m}{\ell} \amm{m}(\epartm{m}{\ell}))_{\ell=1}^\N)$;
		\STATE draw $\epartm{m + 1}{i} \sim \hkm{m}(\epartm{m}{\I{m + 1}{i}},\cdot)$;
		\STATE weight $\displaystyle \wgt{n+1}{i} \gets \frac{\ldm{m}(\epartm{m}{\I{m+1}{i}}, \epartm{m + 1}{i})}{\amm{m}(\epartm{m}{\I{m + 1}{i}})\boldsymbol{\hd}_m(\epartm{m}{\I{m+1}{i}}, \epartm{m + 1}{i})}$;
		\ENDFOR
		\RETURN $(\epartm{m + 1}{i}, \wgtm{m+1}{i}, \I{m + 1}{i})_{i=1}^\N$.
	\end{algorithmic}
\end{algorithm}

%%%% PROPOSITION: EQUAL DISTRIBUTION

\begin{proposition} \label{prop:original:vs:extended}
Let $(\rho_n)_{n \in \nset}$ and $(\M{n})_{n \in \nset}$ be a selection and backward-sampling schedule satisfying Assumption~\ref{assum:seqdet} and let $(n_m)_{m \in \nset}$ be the induced selection times. Furthermore, let $(\epart{n_m}{i}, \tstat[i]{n_m}, \wgt{n_m}{i})_{i = 1}^\N$, $m \in \nset$, be a subsequence of weighted samples generated by Algorithm~\ref{algo:new} (AdaSmooth) for the original model and let $(\epartm{m}{i}, \tstatm[i]{m}, \wgtm{m}{i})_{i = 1}^\N$, $m \in \nset$, be weighted samples generated by Algorithm~\ref{algo:adasmoothdet} (AdaSmooth with systematic selection) for the extended model. Then for every $m \in \nset$, 
$$
(\proj_m(\epartm{m}{i}), \tstatm[i]{m}, \wgtm{m}{i})_{i = 1}^\N \stackrel{\mathcal{D}}{=} (\epart{n_m}{i}, \tstat[i]{n_m}, \wgt{n_m}{i})_{i = 1}^\N. 
$$
 \end{proposition}

\begin{proof}
The proof consists of simply inspecting that the distribution of the outputs of the two algorithms coincide. We proceed by induction. Standing at time $n_m$, suppose that we have generated a sample $(\epart{n_m}{i}, \tstat[i]{n_m}, \wgt{n_m}{i})_{i = 1}^\N$ by applying AdaSmooth to the original model, and we assume that the claim holds true for this sample. First, let us examine the output of AdaSmooth at time $n_{m + 1}$. Since $\rho_{n_m} = 1$ by definition, selection is activated when forming the sample at time $n_m + 1$, but not after that (since $\rho_k = 0$ for all $k \in \intvect{n_m + 1}{n_{m + 1} - 1}$). This means that each particle path $\epart{n_m + 1:n_{m + 1}}{i}$ will be drawn from 
\begin{equation} \label{eq:path:distribution-1}
\epart{n_m + 1:n_{m + 1}}{i} \sim \hk_{n_m} \varotimes \cdots \varotimes  \hk_{n_m - 1}(\epart{n_m}{\I{n_{m}+1}{i}}, \cdot) 
\end{equation}
and assigned the weight 
\begin{equation} \label{eq:weights-1}
\wgt{n_{m + 1}}{i} = \frac{\ld{n_m}(\epart{n_m}{\I{n_{m}+1}{i}}, \epart{n_m+1}{i}) \prod_{k=n_{m}+1}^{n_{m+1}-1}\ld{k}(\epart{k}{i},\epart{k+1}{i})}{\am{n}(\epart{n_m}{\I{n_{m}+1}{i}}) \hd_{n_m}(\epart{n_m}{\I{n_{m}+1}{i}}, \epart{n_m+1}{i}) \prod_{k=n_{m}+1}^{n_{m+1}-1}\hd_k(\epart{k}{i},\epart{k+1}{i})}, 
\end{equation}
where 
\begin{equation} \label{eq:I:dist-1}
\I{n_{m}+1}{i} \sim \catdist((\am{n_m}(\epart{n_m}{\ell}) \wgt{n_m}{\ell})_{\ell = 1}^\N).
\end{equation}
 According to the updating rule (8) (modified to allow also for $\M{m} > 1$ backward samples), each statistic $\tstat[i]{n_m + 1}$ is assigned the value 
\begin{multline}
\tstat[i]{n_{m + 1}} = (1 + \M{n_m})^{-1} \left( \tstat[\I{n_{m}+1}{i}]{n_m} + \adds{n_m}(\epart{n_m}{\I{n_{m}+1}{i}}, \epart{n_m+1}{i})\right. \\
\left. + \sum_{j = 1}^{\M{n_m}} \left( \tstat[\bi{n_m + 1}{i}{j}]{n_m} + \addf{n_m}(\epart{n_m}{\bi{n_m + 1}{i}{j}}, \epart{n_m + 1}{i}) \right) \right)
+ \sum_{k = n_{m}+1}^{n_{m+1} - 1} \adds{k}(\epart{k}{i}, \epart{k+1}{i}) \\
=  (1 + \M{n_m})^{-1} \left( \tstat[\I{n_{m+1}}{i}]{n_m} + \adds{n_m}(\epart{n_m}{\I{n_{m+1}}{i}}, \epart{n_m+1}{i}) + \sum_{k=n_{m}+1}^{n_{m+1}-1} \adds{k}(\epart{k}{i},\epart{k+1}{i}) \right. \\
\left. + \sum_{j = 1}^{\M{n_m}} \left( \tstat[\bi{n+1}{i}{j}]{n_m} + \addf{n_m}(\epart{n_m}{\bi{n_m + 1}{i}{j}}, \epart{n_m + 1}{i}) + \sum_{k=n_{m}+1}^{n_{m+1}-1} \adds{k}(\epart{k}{i},\epart{k+1}{i}) \right) \right), \label{eq:tau:statistics-1}
\end{multline}
where, in the case $\M{n_m} > 0$, the indices $(\bi{n_m + 1}{i}{j})_{j = 1}^{\M{n_m}}$ are conditionally independent and identically distributed according to 
\begin{equation} \label{eq:backward:distribution-1}
\bi{n_m + 1}{i}{j} \sim \catdist((\wgt{n_m}{\ell} \ld{n}(\epart{n_m}{\ell},\epart{n_m + 1}{i})_{\ell = 1}^\N).
\end{equation} 

Now, on the other hand, subjecting the sample $(\epartm{m}{i}, \tstatm[i]{m}, \wgtm{m}{i})_{i = 1}^\N$ to one iteration of AdaSmooth with systematic resampling, Algorithm~\ref{algo:adasmoothdet}, yields path particles $\epartm{m + 1}{i} = \epart{n_m + 1:n_{m+ 1}}{i}$, $i \in \intvect{1}{\N}$, with distribution  
\begin{equation} \label{eq:eq:path:distribution-2}
\epartm{m + 1}{i} = \epart{n_m + 1:n_{m+ 1}}{i} \sim 
\hkm{m}(\epartm{m}{i}, \cdot) = \hk_{n_m} \varotimes \cdots \varotimes \hk_{n_m - 1}(\proj_m(\epartm{m}{\I{m+1}{i}}), \cdot)
\end{equation}
and associated weights 
\begin{align} 
&\wgtm{m}{i} = \frac{\ldm{m}(\epartm{m}{\I{m+1}{i}}, \epartm{m + 1}{i})}{\amm{m}(\epartm{m}{\I{m+1}{i}}) \boldsymbol{\hd}_m(\epartm{m}{\I{m+1}{i}}, \epartm{m + 1}{i})} \\ 
&= \frac{\ld{n_m}(\proj_m(\epartm{m}{\I{m+1}{i}}), \epart{n_m+1}{i}) \prod_{k=n_{m}+1}^{n_{m+1}-1}\ld{k}(\epart{k}{i},\epart{k+1}{i})}{\am{n_m}(\proj_m(\epartm{m}{\I{m+1}{i}})) \hd_{n_m}(\proj_m(\epartm{m}{\I{m+1}{i}}), \epart{n_m+1}{i}) \prod_{k=n_{m}+1}^{n_{m+1}-1}\hd_k(\epart{k}{i},\epart{k+1}{i})}, \label{eq:weights-2}
\end{align}
where 
\begin{equation} \label{eq:I:dist-2}
\I{m+1}{i} \sim \catdist((\amm{m}(\epartm{m}{\ell}) \wgtm{m}{\ell})_{\ell = 1}^\N) = \catdist((\am{n_m}(\proj_m(\epartm{m}{\ell})) \wgt{n_m}{\ell})_{\ell = 1}^\N). 
\end{equation}
Finally, the statistics $(\tstatm[i]{m})_{i = 1}^\N$ are updated according to 
\begin{align}
\tstatm[i]{m+1} &= (1 + \M{n_m})^{-1} \\&\quad\times\left( \tstatm[\I{m+1}{i}]{m} + \bigaddf{m}(\boldsymbol{\xi}_m^{\I{m+1}{i}},\boldsymbol{\xi}_{m+1}^{i})+\sum_{j=1}^{\M{n_m}} \tstatm[\bi{m+1}{i}{j}]{m} + \bigaddf{m}(\boldsymbol{\xi}_m^{\bi{m+1}{i}{j}},\boldsymbol{\xi}_{m+1}^{i}) \right) \\
&
\!\begin{multlined}
= (1 + \M{n_m})^{-1} \left( \tstatm[\I{m+1}{i}]{m} + \adds{n_m}(\proj_m(\epartm{m}{\I{m+1}{i}}), \epart{n_m+1}{i}) + \sum_{k=n_{m}+1}^{n_{m+1}-1} \adds{k}(\epart{k}{i},\epart{k+1}{i}) \right. \\
\left. + \sum_{j = 1}^{\M{n_m}} \left( \tstatm[\bi{m+1}{i}{j}]{m} + \addf{n_m}(\proj_m(\epartm{m}{\bi{m + 1}{i}{j}}), \epart{n_m + 1}{i}) + \sum_{k=n_{m}+1}^{n_{m+1}-1} \adds{k}(\epart{k}{i},\epart{k+1}{i}) \right) \right), 
\end{multlined} \label{eq:tau:statistics-2}
\end{align}
where indices $(\bi{m + 1}{i}{j})_{j = 1}^{\M{n_m}}$ are drawn from $\catdist((\wgtm{m}{\ell} \ldm{m}(\epartm{m}{\ell},\epartm{m + 1}{i})_{\ell = 1}^\N)$ when $\M{m} > 0$; however, by noting that  
\begin{align}
	\lefteqn{\frac{\wgtm{m}{j}\ldm{m}(\boldsymbol{\xi}_m^j,\boldsymbol{\xi}_{m+1}^i)}{\sum_{j' = 1}^{\N} \wgtm{m}{j'} \ldm{m}(\epartm{m}{j'},\boldsymbol{\xi}_{m+1}^i)}} \hspace{10mm}\\
	&=\frac{\wgt{n_m}{j}\ld{n_{m}}(\proj_m(\epartm{m}{j}),\epart{n_{m}+1}{i})\prod_{k'=n_{m}+1}^{n_{m+1}-1}\ld{k'}(\epart{k'}{i},\epart{k'+1}{i})}{\sum_{j'=1}^{\N}\wgt{n_m}{j'}\ld{n_{m}}(\proj_m(\epartm{m}{j'}),\epart{n_{m}+1}{i}) \prod_{k'=n_{m}+1}^{n_{m+1}-1}\ld{k'}(\epart{k'}{i},\epart{k'+1}{i})} \\
	&= \frac{\wgt{n_m}{j}\ld{n_{m}}(\proj_m(\epartm{m}{j}),\epart{n_{m}+1}{i})}{\sum_{j'=1}^{\N}\wgt{n_m}{j'}\ld{n_{m}}(\proj_m(\epartm{m}{j'}),\epart{n_{m}+1}{i})}, 
\end{align}
we conclude that 
\begin{equation} \label{eq:backward:distribution-2}
\catdist((\wgtm{m}{\ell} \ldm{m}(\epartm{m}{\ell},\epartm{m + 1}{i})_{\ell = 1}^\N) = \catdist((\wgt{n_m}{\ell} \ld{n_{m}}(\proj_m(\epartm{m}{j}),\epart{n_{m}+1}{i})_{\ell = 1}^\N).
\end{equation} 

Finally, by comparing \eqref{eq:path:distribution-1} and \eqref{eq:eq:path:distribution-2}, \eqref{eq:weights-1} and \eqref{eq:weights-2}, \eqref{eq:I:dist-1} and \eqref{eq:I:dist-2}, \eqref{eq:tau:statistics-1} and \eqref{eq:tau:statistics-2}, \eqref{eq:backward:distribution-1} and \eqref{eq:backward:distribution-2}, we conclude that under the induction hypothesis,  
$$
(\proj_{m + 1}(\epartm{m + 1}{i}), \tstatm[i]{m + 1}, \wgtm{m + 1}{i})_{i = 1}^\N \stackrel{\mathcal{D}}{=} (\epart{n_{m + 1}}{i}, \tstat[i]{n_{m + 1}}, \wgt{n_{m + 1}}{i})_{i = 1}^\N. 
$$
The base case $m = 0$ is checked similarly. This completes the proof. 
\end{proof}

Thus, in the deterministic case we may reinterpret the model in the aforementioned way and assume systematic resampling. Since the convergence analysis of Algorithm~\ref{algo:adasmsupp} is valid for general models and state spaces, is also applies to the extended model, providing immediately the strong consistency and asymptotic normality of AdaSmooth in the deterministic case. This will be discussed in detail in the following sections, where each $\M{n}$ is again, in accordance with Algorithm~\ref{algo:new}, restricted to be an indicator function (being either zero or one).

\subsection{Proof of Theorem~3.1}\label{sec:proof:3.1}

\begin{proof}
First, we note that Assumption~\ref{assum:1} implies Assumption~\ref{assum:supp1} for the extended model. Thus, we may apply Theorem~\ref{thm:hoef} to Algorithm~\ref{algo:adasmoothdet}. 
We establish Theorem~\ref{thm:as} for an arbitrarily chosen $n \in \nset$; even though this $n$ is generally not a resampling time, we may assume without loss of generality that $n = n_m$ for some $m \in \nset$ (since it does not matter for the distribution of the particle cloud at a give time point whether resampling is performed in the subsequent iteration of the algorithm). Now, since $\post{0:n} h_n = \postm{0:m} \boldsymbol{h}_m$, Proposition~\ref{prop:original:vs:extended} and Corollary~\ref{cor:hoeffding} imply that for every $m \in \nset$, there exist positive constants $c_m$ and $\tilde{c}_m$ such that for all $\epsilon > 0$ and $\N \in \nsetpos$, 
\begin{multline}
	\displaystyle \prob \left(\left \lvert \sum_{i=1}^{\N} \frac{\wgt{n}{i}}{\wgtsum{n}} \tstat[i]{n} - \post{0:n} \adds{n} \right \rvert \ge \epsilon \right) \\= \prob \left(\left \lvert \sum_{i=1}^{\N} \frac{\wgtm{m}{i}}{\wgtsumm{n}} \tstatm[i]{m} - \postm{0:m} \boldsymbol{h}_m \right \rvert \ge \epsilon \right)
	\leq c_m \exp \left( - \tilde{c}_m \N \epsilon^2 \right), 
\end{multline}
where $(\tstat[i]{n}, \wgt{n}{i})_{i = 1}^\N$ and $(\tstatm[i]{m}, \wgtm{m}{i})_{i = 1}^\N$ are produced by $n$ and $m$ iterations of Algorithm~\ref{algo:new} and Algorithm~\ref{algo:adasmoothdet}, respectively. 
\end{proof}

\subsection{Proof of Theorem 3.2}\label{sec:proof:3.2} 
\begin{proof}
	Again, we may assume without loss of generality that $n = n_m$ for some $m \in \nset$. Then, since Assumption~\ref{assum:1} implies Assumption~\ref{assum:supp1}, applying Proposition~\ref{prop:original:vs:extended} and Corollary~\ref{corollary:1} to the extended model yields, for every $m \in \nset$,
	\begin{equation}
		\sqrt{\N} \left( \sum_{i=1}^{\N}\frac{\wgt{n}{i}}{\wgtsum{n}} \tstat[i]{n}  -\postm{0:n} \adds{n} \right) \stackrel{\mathcal{D}}{=} \sqrt{\N} \left( \sum_{i=1}^{\N}\frac{\wgtm{m}{i}}{\wgtsumm{m}}\tstatm[i]{m} - \postm{0:m} \bigadds{m} \right) \convd \boldsymbol{\sigma}_m(\bigadds{m}) Z,
	\end{equation}
	where $Z$ has standard Gaussian distribution and the asymptotic variance $\boldsymbol{\sigma}_m^2(\bigadds{m}) = \sigma_n^2(\adds{n})$ is obtained by casting the extended model into the Corollary \ref{corollary:1}, \ie, 
	{\footnotesize
	\begin{equation}
		\begin{split} \label{eq:asvarmult}
			\lefteqn{\boldsymbol{\sigma}_m^2(\bigadds{m}) = \frac{\Xinitm(\boldsymbol{w}_{-1} \BFcentm{0}{n}^2 \bigadds{m})}{(\Xinitm \lkm{0} \cdots \lkm{m-1 }\1{\Xp{m}})^2} + \sum_{k = 0}^{m - 1} \postm{k} \amm{k} \frac{\postm{k}\lkm{k}( \boldsymbol{w}_k \BFcentm{k+1}{m}^2 \bigadds{m})}{(\postm{k}\lkm{k}\cdots\lkm{m-1}\1{\Xp{m}})^2}} \\
			& + \sum_{k=0}^{m-1} \frac{\M{n_k} \postm{k} \amm{k}}{1+\M{n_k}} \sum_{\ell = 0}^{k} \frac{\postm{\ell} \lkm{\ell} ( \bkm{\ell}(\tk{\ell}\bigadds{\ell} + \bigaddf{\ell} - \tk{\ell+1} \bigadds{\ell+1})^2 \lkm{\ell+1} \cdots \lkm{k}\{\bkm{k}\boldsymbol{w}_k (\lkm{k+1} \cdots \lkm{m-1} \1{\Xp{m}} )^2\} )}{(\postm{\ell}\lkm{\ell} \cdots \lkm{k-1}\1{\Xp{k}})(\postm{k}\lkm{k} \cdots \lkm{m-1} \1{\Xp{m}}  )^2 \prod_{j=\ell}^k(1+\M{n_j})} \\
			& + \sum_{k=0}^{m-1} \frac{\postm{k} \amm{k}}{1+\M{n_k}}\sum_{\ell = 0}^{k} \frac{\postm{\ell} \lkm{\ell}\{ \bkm{\ell}(\tk{\ell}\bigadds{\ell} + \bigaddf{\ell} - \tk{\ell+1}\bigadds{\ell+1})^2 \lkm{\ell+1} \cdots \lkm{k}(\boldsymbol{w}_k \{ \lkm{k+1} \cdots \lkm{m-1} \1{\Xp{m}} \} ^2) \}}{(\postm{\ell}\lkm{\ell} \cdots \lkm{k-1}\1{\Xp{k}})(\postm{k}\lkm{k}\cdots\lkm{m-1}\1{\Xp{m}}  )^2 \prod_{j=\ell}^k(1+ \M{n_j})} \\
			&+ 2 \sum_{k=0}^{m-1} \M{n_k} \postm{k} \amm{k} \frac{\postm{k}\lkm{k}\{\boldsymbol{w}_k \BFcentm{k+1}{m}\bigadds{m}(\tk{k}\bigadds{k}+\bigaddf{k}-\tk{k+1}\bigadds{k+1}) \lkm{k+1}\cdots\lkm{m-1}\1{\Xp{m}} \}}{(1+\M{n_k})^2 (\postm{k}\lkm{k}\cdots\lkm{m-1}\1{\Xp{m}})^2} \\
			& +\sum_{k=0}^{m-1} \postm{k} \amm{k} \frac{\postm{k}\lkm{k}((\boldsymbol{w}_k-\bkm{k}\boldsymbol{w}_k) \{\BFcentm{k+1}{m}\bigadds{m}+(\tk{k}\bigadds{k}+\bigaddf{k}-\tk{k+1}\bigadds{k+1})\lkm{k+1}\cdots\lkm{m-1}\1{\Xp{m}}\}^2)}{(1+\M{n_k})^2(\postm{k}\lkm{k}\cdots\lkm{m-1}\1{\Xp{m}})^2}.
		\end{split}%\par
	\end{equation}
		} 
This completes the proof. 
\end{proof}

\subsection{Proof of Theorem~3.3}\label{sec:proof:3.3} 
\begin{proof}
We suppose that the Assumptions~\ref{assum:dist}--\ref{assum:compact} hold and show that these imply that the assumptions of Theorem~\ref{thm:boundsupp} (\ie, Assumption~\ref{assum:2}) are satisfied for the extended model. Then Theorem~\ref{thm:boundsupp} provides an $\ordo(n)$ bound on $\boldsymbol{\sigma}_m^2(\bigadds{m}) = \sigma_n^2(\adds{n})$ (recall that $n = n_m$ by assumption). More specifically, recall that by Assumption~\ref{assum:dist} there exists $\maxd\in\nsetpos$ such that $n_{m}-n_{m-1}\le \maxd$ for all $m\in\nset$; thus, for all $m \in \nset$, using also Assumption~\ref{assum:compact}, 
\begin{equation}
	\ldm{m}(\xvec{m},\xvec{m+1})=\prod_{k=n_{m}}^{n_{m+1}-1}\ld{k}(x_k,x_{k+1})\le \hkup^{n_{m+1}-n_m}\le \hkup_\maxd
 \eqdef \hkup^\maxd \vee \hkup 
\end{equation}
and 
\begin{equation}
	\ldm{m}(\xvec{m},\xvec{m+1}) \ge \hklow^{n_{m+1}-n_m}\ge \hklow_\maxd \eqdef \hklow^\maxd \wedge \hklow. 
\end{equation}
This checks the first condition, and we may define $\mr_\maxd \eqdef 1-\hklow_\maxd/\hkup_\maxd$. 

To check the second condition of Assumption~\ref{assum:2}, recall from Assumption~\ref{assum:compact} that all the single-step weight functions are bounded by $\mdup$; thus, 
\begin{equation}
	\supn{\boldsymbol{w}_m} \le \supn{w_{n_m} \langle1\rangle} \prod_{k=n_m+1}^{n_{m+1}-1} \supn{w_{k} \langle0\rangle}\le \mdup^{n_{m+1}-n_m}\le \mdup_\maxd \eqdef \mdup^\maxd \vee \mdup, 
\end{equation}
and, similarly, $\supn{\boldsymbol{w}_{-1}} \le \mdup_\maxd$. Finally, since also uniform boundedness of $(\am{n})_{n \in \nset}$ implies trivially the uniform boundedness of $(\amm{m})_{m \in \nset}$, we conclude that Assumption~\ref{assum:2} holds also for the extended model. 

We also note that for additive functionals $(\adds{n})_{n \in \nset}$ whose terms are bounded in such a way as it is stated in Theorem~\ref{thm:bound}, also the terms of the induced functionals $(\bigadds{m})$ are trivially bounded; indeed, for all $m \in\nset$,
\begin{equation}
	\supn{\bigaddf{m}}\le\sum_{k=n_m}^{n_{m+1}-1}\supn{\addf{k}}\le \maxd \hbd\quad \text{and}\quad \supn{\bigadds{0}+\bigaddf{0}}\le2 \maxd \hbd.
\end{equation}

We now apply Theorem~\ref{thm:boundsupp}. On the basis of the given subsequence $(n_m)_{m \in \nset}$, we define another subsequence $(n_m')_{m\in\nset}$ as
\begin{equation}
	n_{m}' \eqdef \argmax_{n_{m-1}<k\le n_m}\frac{1}{k}\sigma_k^2(\adds{k}),\quad m\in\nset,
\end{equation}
so that
\begin{equation}
	\limsup_{n \to \infty}\frac{1}{n}\sigma_{n}^2(\adds{n}) = \limsup_{m \to \infty} \frac{1}{n_m'} \sigma_{n_m'}^2(\adds{n_m'}).
\end{equation}
	Now, note that the asymptotic variance $\sigma_{n_m'}^2(\adds{n_m'})$ corresponds to the deterministic selection schedule $(n_0, n_1,\dots, n_{m - 1}, n_m')$ comprising $m$ selection operations before time $n_m'$. Then we can let $\boldsymbol{\sigma}_m^2(\bigadds{m}') \eqdef \sigma_{n_m'}^2(\adds{n_m'})$ be given by the asymptotic variance \eqref{eq:asvarmult}, but where $n_m$ is replaced by $n_m'$. It follows that	
	\begin{equation}
		\limsup_{n\rightarrow\infty}\frac{1}{n}\sigma_{n}^2(\adds{n})= \limsup_{m\rightarrow\infty}\frac{1}{n_m'}\boldsymbol{\sigma}_m^2(\bigadds{m}')\le \limsup_{m\rightarrow\infty}\frac{1}{m}\boldsymbol{\sigma}_m^2(\bigadds{m}').
	\end{equation}
	Now, since Assumption~\ref{assum:2} holds for the extended model, Theorem~\ref{thm:boundsupp} implies that
	\begin{multline}
		\limsup_{m \to \infty} \frac{1}{m} \boldsymbol{\sigma}_m^2(\bigadds{m}')
		\leq \maxd^2 \hbd^2 \frac{4 \amup \mdup_\maxd}{\hklow_\maxd(1-\mr_\maxd)^4} 
		\left( \vphantom{\sum_{j=0}^{m-1}} \frac{108+ \mdup\mr_\maxd(3-2\mr_\maxd)(1+\mr_\maxd)}{2\mr_\maxd^2} 
		+ \right. \\
		\left. \frac{ (3-2\mr_\maxd)^2}{ 1-\mr_\maxd} \lim_{m \to \infty}\frac{1}{m} \sum_{j=0}^{m-1} \sum_{\ell = 0}^j \prod_{k=\ell}^j(1+\M{n_k})^{-1} \right).
	\end{multline}
	Finally, since $m$ is the number of selection operations before time $n_m'$, this bound corresponds to the one in Theorem~\ref{thm:bound}.
\end{proof}

%%%% Proof of Lemma~3.5 

\section{Proof of Lemma~3.5}\label{sec:proof:3.5} 
\label{sec:adaptres}

Lemma~\ref{lemma:ess} is an immediate consequence of the following result, which extend a similar result obtained by \citet{douc:moulines:2008} for adaptive sequential importance sampling with resampling to the more general adaptive APF considered in the present paper. Our proof follows similar lines. 

\begin{lemma} \label{lemma:essmeas}
	Let Assumption~\ref{assum:1} and Assumption~\ref{assum:adaptiveRes} hold. Moreover, let $(\epart{n}{i},\wgt{n}{i})_{i=1}^\N$, $n\in\nset$, be weighted samples generated by Algorithm~\ref{algo:sisr} on the basis of the selection schedule $(\res[\N]{n})_{n\in\nset}$ in Assumption~\ref{assum:adaptiveRes}. Then for every $n\in\nset$ there exist finite measures $\sqmeas{n}$ and $\unorm{n}$ on $\alg{X}_n$ such that for all $\testf[n] \in \bmf{\alg{X}_n}$,
	\begin{equation}\label{eq:convmeas}
		\N\sum_{i=1}^{\N}\left(\frac{\wgt{n}{i}}{\wgtsum{n}}\right)^2\testf[n](\epart{n}{i})\convp\sqmeas{n}\testf[n], \quad \frac{1}{\N}\sum_{i=1}^{\N} \wgt{n}{i} \testf[n](\epart{n}{i})\convp\unorm{n}\testf[n]
	\end{equation}
	and $\post{n} \testf[n] = \unorm{n} \testf[n] / \unorm{n} \1{\set{X}_n}$. 
	The measures $(\sqmeas{n})_{n \in \nset}$ and $(\unorm{n})_{n \in \nset}$ satisfy the recursions
	\begin{equation}\label{eq:recsq}
		\sqmeas{n+1}\testf[n+1] = \frac{\sqmeas{n}\lk{n}(w_n\langle0\rangle\testf[n+1])}{(\post{n}\lk{n}\1{\set{X}_{n+1}})^2}(1-\res[\alpha,\maxd]{n})+\post{n}\am{n}\frac{\post{n}\lk{n}(w_n\langle1\rangle\testf[n+1])}{(\post{n}\lk{n}\1{\set{X}_{n+1}})^2}\res[\alpha,\maxd]{n},
	\end{equation}
	and
	\begin{equation}\label{eq:recun} 
	\unorm{n+1}\testf[n+1] = \unorm{n}\lk{n}\testf[n+1](1-\res[\alpha,\maxd]{n})+(\post{n}\am{n})^{-1}\post{n}\lk{n}\testf[n+1]\res[\alpha,\maxd]{n}, 
	\end{equation}
	where $\res[\alpha,\maxd]{n} \in \{0, 1\}$ is the limit in probability of $\res[\N]{n}$ as $\N \to \infty$ and 
	\begin{align}
		\res[\alpha,\maxd]{n+1} = 1-\1{\{(\sqmeas{n+1}\1{\set{X}_{n+1}})^{-1}\ge\alpha\}}\1{\{d_n+1<\maxd\}}, \quad d_{n+1} = (1-\res[\alpha,\maxd]{n+1})(1+d_n).
	\end{align}
	These recursions are initialized by 
	\begin{equation}
		\sqmeas{0}\testf[0] = \frac{\nu(w_{-1}^2\testf[0])}{(\nu w_{-1})^2},\quad \unorm{0}\testf[0] = \nu (w_{-1}\testf[0]) 
	\end{equation}
	and 
	\begin{equation}
		\res[\alpha,\maxd]{0} = \1{\{(\sqmeas{0}\1{\set{X}_0})^{-1}<\alpha\}}, \quad d_0 = 1-\res[\alpha,\maxd]{0}. 
	\end{equation}
\end{lemma}

%%%% Proof 

\begin{proof}
We proceed by induction and assume that the limits \eqref{eq:convmeas} hold true for some $n \in \nset$ and that $\post{n}\testf[n]=\unorm{n}\testf[n]/\unorm{n}\1{\set{X}_n}$ for all $\testf[n]\in \bmf{\alg{X}_n}$. In addition, we assume that there exists $d_{n - 1} \in \nset$ such that 
$$
d_{n - 1}^\N \convp d_{n - 1},
$$
as $\N$ tends to infinity. We then establish the limits of  
\begin{multline} \label{eq:gamma:limit}
\N \sum_{i=1}^\N\left(\frac{\wgt{n+1}{i}}{\wgtsum{n+1}}\right)^2\testf[n+1](\epart{n+1}{i}) \stackrel{\mathcal{D}}{=} \N \sum_{i=1}^\N\left(\frac{\wgttil{n+1}{i}}{\wgtsumtil{n+1}}\right)^2\testf[n+1](\eparttil{n+1}{i})(1-\res[\N]{n}) \\ 
	+\N \sum_{i=1}^\N\left(\frac{\wgtbar{n+1}{i}}{\wgtsumbar{n+1}}\right)^2\testf[n+1](\epartbar{n+1}{i})\res[\N]{n}, 
\end{multline}
and 
\begin{multline} \label{eq:varphi:limit}
\frac{1}{\N} \sum_{i=1}^{\N} \wgt{n+1}{i} \testf[n+1](\epart{n+1}{i})\\ \stackrel{\mathcal{D}}{=} \frac{1}{\N} \sum_{i=1}^{\N} \wgttil{n+1}{i} \testf[n+1](\eparttil{n+1}{i})(1-\res[\N]{n}) + \frac{1}{\N} \sum_{i=1}^{\N} \wgtbar{n+1}{i} \testf[n+1](\epartbar{n+1}{i})\res[\N]{n}
\end{multline}
as $\N$ tends to infinity, where $(\eparttil{n + 1}{i}, \wgttil{n + 1}{i})_{i = 1}^\N$ and $(\epartbar{n + 1}{i}, \wgtbar{n + 1}{i})_{i = 1}^\N$ are weighted samples obtained by propagating $(\epart{n}{i}, \wgt{n}{i})_{i = 1}^\N$ by pure mutation (Lines~\ref{line:noselection} and \ref{line:mutation} in Algorithm~\ref{algo:sisr}) and by selection plus mutation (Lines~\ref{line:selection} and \ref{line:mutation}), respectively, and $\res[\N]{n}=1-\1{\{\ess_n \ge \alpha \N\}}\1{\{d_{n-1}^\N+1<\maxd\}}$. By the induction hypothesis, 
\begin{equation}
	\frac{1}{\N} \ess_n = \left(\N\sum_{i=1}^\N\left(\frac{\wgt{n}{i}}{\wgtsum{n}}\right)^2\right)^{-1}\convp (\sqmeas{n}\1{\set{X}_n})^{-1}
\end{equation}
and, consequently, 
$$
\res[\N]{n} \convp \res[\alpha,\maxd]{n} \eqdef 1 - \1{\{(\sqmeas{n}\1{\set{X}_n})^{-1}\ge\alpha\}}\1{\{d_{n-1}+1<\maxd\}}
$$
and 
$$
d_n^\N = (1 - \res[\N]{n}) (1 + d_{n - 1}^\N) \convp d_n \eqdef (1 - \res[\alpha,\maxd]{n})(1 + d_{n - 1}). 
$$

When examining separately the different terms, corresponding to the cases where selection is triggered and not triggered, of the decompositions \eqref{eq:gamma:limit} and \eqref{eq:varphi:limit} above, Lemma~A.1 in \citet{douc:moulines:2008} will be instrumental.

\begin{proof}[Case 1: propagation without selection]
We determine the limit measures at time $n+1$ assuming only mutation is carried out (Line~\ref{line:noselection} of Algorithm \ref{algo:sisr}). Each particle $\epart{n}{i}$ is propagated by sampling $\eparttil{n+1}{i}$ from the proposal density $\hd_n(\epart{n}{i},\cdot)$ and assigning this draw the weight $\wgttil{n+1}{i}=\wgt{n}{i}w_n\langle0\rangle(\epart{n}{i},\eparttil{n+1}{i})$, where $w_n\langle0\rangle(\epart{n}{i},\eparttil{n+1}{i})=\ld{n}(\epart{n}{i},\eparttil{n+1}{i})/\hd_n(\epart{n}{i},\eparttil{n+1}{i})$ is bounded by assumption. We define the triangular array
\begin{equation}
	\arr{i} \eqdef \wgtsum{n}^{-2}\N(\wgttil{n+1}{i})^2\testf[n+1](\eparttil{n+1}{i}),\quad \N\in\nsetpos,
\end{equation}
and find the limit of $\sum_{i = 1}^\N \arr{i}$ using Lemma A.1 in \citet{douc:moulines:2008}. This lemma has two conditions that need to be checked. First, we consider
\begin{align}
	\sum_{i=1}^\N\E[\arr{i}\mid\partfiltbar{n}]&=\N \wgtsum{n}^{-2}\sum_{i=1}^\N(\wgt{n}{i})^2 \int 
	w_n^2\langle0\rangle(\epart{n}{i}, x)\testf[n+1](x) \, \hk_n(\epart{n}{i}, dx) \\ 
	&=\N \sum_{i=1}^\N\left(\frac{\wgt{n}{i}}{\wgtsum{n}}\right)^2\lk{n}(\epart{n}{i},w_n\langle0\rangle\testf[n+1])\convp \sqmeas{n}\lk{n}(w_n\langle0\rangle\testf[n+1]),
\end{align}
where we used the induction hypothesis and the fact that $\lk{n}(w_n\langle0\rangle\testf[n+1])\in \bmf{\alg{X}_n}$. This limit establishes the first condition. Next, we need to check that for all $\epsilon>0$, the limit of
\begin{multline}
	\sum_{i=1}^\N \E \left[ |\arr{i}|\1{|\arr{i}|\ge \epsilon}\mid\partfiltbar{n} \right] \\ 
	\le \N \sum_{i=1}^\N\left(\frac{\wgt{n}{i}}{\wgtsum{n}}\right)^2\supn{w_n\langle0\rangle}^2\supn{\testf[n+1]}\\\times\1{\{(\N/\wgtsum{n})^{2} \max_{-1\le k\le n}\prod_{m=k}^{n}\supn{w_m\langle0\rangle}^2\supn{\testf[n+1]}\ge\epsilon\N\}}
\end{multline}
in probability is zero, which is indeed the case as the indicator tends to zero, provided that $\wgtsum{n}/\N$ tends to $\unorm{n}\1{\set{X}_n}>0$. Thus, both the conditions of the lemma are satisfied, implying that $\sum_{i=1}^\N\arr{i}$ tends to $\sqmeas{n}\lk{n}(w_n\langle0\rangle\testf[n+1])$ in probability. Next, we define another triangular array
\begin{equation}
	\arru{i}\vcentcolon=\N^{-1}\wgttil{n+1}{i}\testf[n+1](\eparttil{n+1}{i}), \quad \N \in \nset.
\end{equation}
In order to identify the limit of $\sum_{i=1}^{\N}\arru{i}$ we check again the two conditions of Lemma~A.1 in \citet{douc:moulines:2008}. First, by the induction hypothesis, 
\begin{align}
	\sum_{i=1}^\N \E \left[ \arru{i} \mid \partfiltbar{n} \right]&= \frac{1}{\N}  \sum_{i=1}^\N \wgt{n}{i} \int w_n\langle0\rangle(\epart{n}{i},x) \testf[n+1](x) \, \hk_n(\epart{n}{i}, x)  
	\\&= \frac{1}{\N} \sum_{i=1}^\N\wgt{n}{i}\lk{n}(\epart{n}{i},\testf[n+1])\convp \unorm{n}\lk{n}\testf[n+1].
\end{align}
Moreover, since for all $\epsilon>0$,
\begin{multline}
	\sum_{i=1}^\N \E\left[ |\arru{i}| \1{\{|\arru{i}|\ge \epsilon\}} \mid \partfiltbar{n} \right] \\ 
	\le \frac{1}{\N} \sum_{i=1}^\N \wgt{n}{i} \supn{w_n \langle0\rangle} \supn{\testf[n+1]}\1{\{\max_{-1\le k\le n}\prod_{m=k}^n \supn{w_m\langle0\rangle}\supn{\testf[n+1]}\ge\epsilon\N\}}\convp 0,
\end{multline}
it holds that 
$$
\sum_{i=1}^\N \arru{i} = \frac{1}{\N} \sum_{i = 1}^\N \wgttil{n + 1}{i} \testf[n+1](\eparttil{n + 1}{i}) \convp \tilde{\varphi}_{n + 1} \testf[n+1] \eqdef \unorm{n}\lk{n}\testf[n+1]
$$
Furthermore,  
\begin{equation}
	\frac{\wgtsumtil{n+1}}{\wgtsum{n}}\convp\frac{\unorm{n}\lk{n}\1{\set{X}_{n+1}}}{\unorm{n}\1{\set{X}_n}}=\post{n}\lk{n}\1{\set{X}_{n+1}},
\end{equation}
and by combining the previous limits we conclude that
\begin{multline}
	\N\sum_{i=1}^\N \left( \frac{\wgttil{n+1}{i}}{\wgtsumtil{n+1}} \right)^2 \testf[n+1](\eparttil{n+1}{i}) = (\wgtsumtil{n+1}/\wgtsum{n})^{-2}\sum_{i=1}^\N\arr{i}\\\convp \tilde{\gamma}_{n + 1} \testf[n+1] \eqdef \frac{\sqmeas{n}\lk{n}(w_n\langle0\rangle\testf[n+1])}{(\post{n}\lk{n}\1{\set{X}_{n+1}})^2}.
\end{multline}
\end{proof}
 
\begin{proof}[Case~2: propagation with selection]
We now determine the measure $\sqmeas{n+1}$ when mutation is preceded by selection. In this case, each index $\I{n+1}{i}$ is drawn from $\catdist((\wgt{n}{\ell}\am{n}(\epart{n}{i}))_{\ell=1}^\N)$, whereupon the resampled particle $\epart{n}{\I{n+1}{i}}$ is propagated by drawing $\epartbar{n+1}{i}$ from the density $\hd_n(\epart{n}{\I{n+1}{i}}, \cdot)$. Finally, the particle is assigned the weight $\wgtbar{n+1}{i}=w_n\langle1\rangle(\epart{n}{\I{n+1}{i}},\epartbar{n+1}{i})=\ld{n}(\epart{n}{\I{n+1}{i}},\epartbar{n+1}{i})/(\am{n}(\epart{n}{\I{n+1}{i}})\hd_n(\epart{n}{\I{n+1}{i}},\epartbar{n+1}{i}))$. In order to repeat the arguments of Case~1, we define the triangular array
\begin{equation}
	\arr{i} \eqdef \N^{-1}(\wgtbar{n+1}{i})^2\testf[n+1](\epartbar{n+1}{i}), \quad \N \in \nset, 
\end{equation}
and consider
\begin{align}
	\lefteqn{\sum_{i=1}^\N \E \left [\arr{i}\mid\partfiltbar{n} \right]} \\ 
	&=(\post[\N]{n}\am{n})^{-1} \sum_{i=1}^\N \frac{\wgt{n}{i} \am{n}(\epart{n}{i})}{\wgtsum{n}} \int w_n^2\langle1\rangle(\epart{n}{i}, x) \testf[n+1](x) \,\hk_n(\epart{n}{i}, dx) \\
	&=(\post[\N]{n}\am{n})^{-1}\sum_{i=1}^\N\frac{\wgt{n}{i}}{\wgtsum{n}}\lk{n}(\epart{n}{i},w_n\langle1\rangle\testf[n+1])\convp (\post{n}\am{n})^{-1}\post{n}\lk{n}(w_n\langle1\rangle\testf[n+1]),
\end{align}
where the limit follows by the induction hypothesis since $\lk{n}(w_n\langle1\rangle\testf[n+1])\in \bmf{\alg{X}_n}$. The second condition is checked easily using the bound
\begin{equation}
	\sum_{i=1}^\N \E \left[ |\arr{i}| \1{\{|\arr{i}| \ge \epsilon\}} \mid \partfiltbar{n} \right] \\
	\le \supn{w_n\langle1\rangle}^2\supn{\testf[n+1]}\1{\{\supn{w_n\langle1\rangle}^2\supn{\testf[n+1]} \geq \epsilon\N\}}\convp 0, 
\end{equation}
which holds for every $\epsilon > 0$. Thus, $\sum_{i=1}^\N\arr{i}$ tends to $(\post{n}\am{n})^{-1}\post{n}\lk{n}(w_n\langle1\rangle\testf[n+1])$ in probability as $\N$ tends to infinity. 

Next, we introduce the array
\begin{equation}
	\arru{i} \eqdef \N^{-1}\wgtbar{n+1}{i}\testf[n+1](\epartbar{n+1}{i}), \quad \N \in \nset, 
\end{equation}
and use the same approach as before. First,
\begin{align}
	\lefteqn{\sum_{i=1}^\N \E \left[ \arru{i} \mid \partfiltbar{n} \right]}\\
	&=(\post[\N]{n}\am{n})^{-1} \sum_{i=1}^\N \frac{\wgt{n}{i}\am{n}(\epart{n}{i})}{\wgtsum{n}} \int w_n\langle1\rangle(\epart{n}{i},x) \testf[n+1](x)\, \hk_n(\epart{n}{i}, dx) \\
	&=(\post[\N]{n}\am{n})^{-1}\sum_{i=1}^\N\frac{\wgt{n}{i}}{\wgtsum{n}}\lk{n}(\epart{n}{i},\testf[n+1])\convp(\post{n}\am{n})^{-1}\post{n}\lk{n}\testf[n+1]
\end{align}
and, second, for every $\epsilon > 0$, 
\begin{equation}
	\sum_{i=1}^\N \E \left[ |\arru{i}| \1{\{|\arru{i}| \ge \epsilon\}}\mid\partfiltbar{n} \right] \\
	\le \supn{w_n\langle1\rangle}\supn{\testf[n+1]}\1{\{\supn{w_n\langle1\rangle}\supn{\testf[n+1]}\ge\epsilon\N\}}\convp 0.
\end{equation}
Thus, Lemma~A.1 in \citet{douc:moulines:2008} applies, implying that 
$$
\frac{1}{\N} \sum_{i=1}^{\N} \wgtbar{n+1}{i}\testf[n+1](\epartbar{n+1}{i}) \convp \bar{\varphi}_{n + 1} \testf[n+1] \eqdef (\post{n}\am{n})^{-1}\post{n}\lk{n}\testf[n+1]
$$ 
as $\N$ tends to infinity. Finally, combining the previous limits yields 
\begin{multline}
	\N\sum_{i=1}^\N\left(\frac{\wgtbar{n+1}{i}}{\wgtsumbar{n+1}}\right)^2\testf[n+1](\epartbar{n+1}{i})=\left(\wgtsumbar{n+1}/\N\right)^{-2}\sum_{i=1}^\N\arr{i}\\\convp \bar{\gamma}_{n + 1} \testf[n+1] \eqdef \post{n} \am{n} \frac{\post{n}\lk{n}(w_n\langle1\rangle\testf[n+1])}{(\post{n}\lk{n}\1{\set{X}_{n+1}})^2}.
\end{multline}

\end{proof}

Having established the limits in Case~1 and Case~2, \eqref{eq:gamma:limit} and \eqref{eq:varphi:limit} yield 
\begin{multline} \label{eq:gamma:update:one:step}
\N \sum_{i=1}^\N\left(\frac{\wgt{n+1}{i}}{\wgtsum{n+1}}\right)^2\testf[n+1](\epart{n+1}{i})\\ \convp \gamma_{n + 1} \testf[n+1] \eqdef \tilde{\gamma}_{n + 1} \testf[n+1] (1 - \res[\alpha,\maxd]{n})  + \bar{\gamma}_{n + 1} \testf[n+1] \res[\alpha,\maxd]{n} \\
= \frac{\sqmeas{n}\lk{n}(w_n\langle0\rangle\testf[n+1])}{(\post{n}\lk{n}\1{\set{X}_{n+1}})^2}(1-\res[\alpha,\maxd]{n})+\post{n}\am{n}\frac{\post{n}\lk{n}(w_n\langle1\rangle\testf[n+1])}{(\post{n}\lk{n}\1{\set{X}_{n+1}})^2}\res[\alpha,\maxd]{n} 
\end{multline}
and 
\begin{multline} \label{eq:varphi:update:one:step}
\frac{1}{\N} \sum_{i=1}^{\N} \wgt{n+1}{i} \testf[n+1](\epart{n+1}{i})\\ \convp \unorm{n+1} \testf[n+1] \eqdef \tilde{\varphi}_{n + 1} \testf[n+1] (1-\res[\alpha,\maxd]{n}) + \bar{\varphi}_{n + 1} \testf[n+1] \res[\alpha,\maxd]{n} \\
= \unorm{n}\lk{n}\testf[n+1](1-\res[\alpha,\maxd]{n})+(\post{n}\am{n})^{-1}\post{n}\lk{n}\testf[n+1]\res[\alpha,\maxd]{n}. 
\end{multline}
In addition, since $\res[\alpha,\maxd]{n} \in \{0, 1\}$ and $\varphi_n / \varphi_n \1{\set{X}_n} = \post{n}$, it is easily seen that    
\begin{equation} \label{eq:varphi:vs:post}
\frac{\varphi_{n + 1}}{\varphi_{n + 1} \1{\set{X}_n + 1}} = \frac{\post{n}\lk{n} }{\post{n}\lk{n} \1{\set{X}_{n + 1}}} = \post{n + 1}.  
\end{equation}

It remains to check the base case $n = 1$. Recall that $(\epart{0}{i})_{i=1}^\N$ are sampled from $\nu^{\varotimes\N}$; thus, by the law of large numbers,
\begin{multline}
	\N \sum_{i=1}^\N\left(\frac{\wgt{0}{i}}{\wgtsum{0}}\right)^2\testf[0](\epart{0}{i}) \\=\left(\frac{1}{\N} \sum_{i=1}^\N w_{-1}(\epart{0}{i})\right)^{-2} \frac{1}{\N} \sum_{i=1}^\N w_{-1}^2(\epart{0}{i})\testf[0](\epart{0}{i})
	\convp \sqmeas{0}\testf[0] = \frac{\nu (w_{-1}^2\testf[0])}{(\nu w_{-1})^2}.
\end{multline}
Similarly, 
\begin{align}
	\frac{1}{\N} \sum_{i=1}^\N \wgt{0}{i} \testf[0](\epart{0}{i}) \convp \unorm{0} \testf[0] = \nu(w_{-1}\testf[0]) = \Xinit \testf[0]. 
\end{align}
and, consequently, $\post{0} \testf[0] = \unorm{0} \testf[0] / \unorm{0} \1{\set{X}_0}$. Moreover, as a consequence,  
$$
	\res[\N]{0} = \1{\{ \ess_0 < \alpha \N\}} \convp \res[\alpha,\maxd]{0} \eqdef \1{\{(\sqmeas{0}\1{\set{X}_0})^{-1}<\alpha\}} 
$$
and 
$$
	d_0^\N =  1 - \res[\N]{0} \convp 1 - \res[\alpha,\maxd]{0}.  
$$
Thus, using \eqref{eq:gamma:update:one:step} and \eqref{eq:varphi:update:one:step}, 
$$
	\gamma_1 \testf[1] = \frac{\sqmeas{0}\lk{0}(w_0\langle0\rangle\testf[1])}{(\post{0}\lk{0}\1{\set{X}_1})^2}(1-\res[\alpha,\maxd]{0})+\post{0} \am{0} \frac{\post{0} \lk{0}(w_0\langle1\rangle\testf[1])}{(\post{0}\lk{0}\1{\set{X}_1})^2}\res[\alpha,\maxd]{0}
$$
and 
$$
	\varphi_1 \testf[1] = \unorm{0}\lk{0}\testf[1](1-\res[\alpha,\maxd]{0})+(\post{0}\am{0})^{-1}\post{0}\lk{0}\testf[1]\res[\alpha,\maxd]{0},   
$$
and by \eqref{eq:varphi:vs:post}, $\varphi_1 / \varphi_1 \1{\set{X}_1} = \post{1}$. This completes the proof. 
\end{proof}

\begin{remark}
The attentive reader has probably noticed that we in the previous have assumed that $\alpha \ne (\sqmeas{n} \1{\set{X}_n})^{-1}$ for all $n \in \nset$, as $(\sqmeas{n} \1{\set{X}_n})^{-1}$ is the limit of $\ess_n / \N$ and $\alpha$ is a discontinuity point of the indicator functions $\1{\{ \cdot \geq \alpha \}}$. This technicality, which is not an issue in practice, can be coped with by randomizing the threshold $\alpha$; see \citet[Section~5.2]{delmoral:doucet:jasra:2012} for details.
\end{remark}

\section{Proof of Proposition 3.4}\label{sec:propbound}
\begin{proof}
	For all $j\in\nset$, let $\varepsilon_{n_j} =1$ if $j=k\dd-1$ for $k\in \nsetpos$, and $\varepsilon_{n_j} = 0$ otherwise. This is the backward-sampling schedule that maximizes the expression $\sum_{m=0}^{r_n-1}\sum_{\ell = 0}^{m}\prod_{j=\ell}^{m}(1+\varepsilon_{n_j})^{-1}$ for any $n$, under the constraint $\dd_j\le \dd$, $j\ge-1$, and we will hence consider the limit in this case. Now, rewrite any $m\in \nset$ as $m=a_m\dd - 1 +b_m$, where $a_m \eqdef \floor{(m+1)/\dd}$ and $b_m \eqdef m+1-a_m\dd$. Then we have, for $k\in \intvect{0}{a_m-1}$ and $ \ell \in \intvect{k\dd}{(k+1)\dd-1} $,
	\begin{equation}\label{eq:prodm}
		\prod_{j=\ell}^{m}(1+\varepsilon_{n_j})^{-1}={2}^{-(a_m-k)}.
	\end{equation}
	Thus, for each $k\in \intvect{0}{a_m-1}$, there are exactly $\dd$ values of $\ell$ such that \eqref{eq:prodm} holds. It also holds that $\prod_{j=\ell}^{m}(1+\varepsilon_{n_j})^{-1}=1$ if $\ell \in \intvect{a_m\dd}{m}$, which happens for $b_m$ distinct values of $\ell$. Summing up,
	\begin{multline}
		\sum_{\ell = 0}^{m}\prod_{j=\ell}^{m}(1+\varepsilon_{n_j})^{-1}= \dd \sum_{k=0}^{a_m-1}2^{-(a_m-k)} + b_m=\dd \sum_{k=1}^{a_m}2^{-k} + b_m\\=\dd (1-2^{-a_m})+b_m.
	\end{multline}
	The first term converges as $m\rightarrow\infty$, since $a_m\rightarrow\infty$ as well. Taking the Ces\`{a}ro mean of the first term yields $\lim_{m\rightarrow\infty}\dd (1-2^{-a_m})=\dd$.
	The sequence $(b_m)_{m \in \nset}$ is not convergent since it is periodic; however, taking the Ces\`{a}ro mean yields    
	\begin{align}
		\lim_{n\rightarrow\infty}\frac{1}{r_n}\floor{r_n/\dd}\sum_{i=0}^{\dd-1}i\le\lim_{n\rightarrow\infty}\frac{1}{r_n}\sum_{m=0}^{r_n-1}b_m\le \lim_{n\rightarrow\infty}\frac{1}{r_n}\ceil{r_n/\dd}\sum_{i=0}^{\dd-1}i.
	\end{align}
	Since $\lim_{n\rightarrow\infty}r_n^{-1}\floor{r_n/\dd}=\lim_{n\rightarrow\infty}r_n^{-1}\ceil{r_n/\dd}=\dd^{-1}$ and $\sum_{i=0}^{\dd-1}i=\dd(\dd-1)/2$ we finally obtain
	\begin{align*}
		\lim_{n\rightarrow\infty}\frac{1}{r_n}\sum_{m=0}^{r_n-1}(\dd (1-2^{-a_m})+b_m)=\dd+\frac{\dd-1}{2}=\frac{3\dd-1}{2}, 
	\end{align*}	
	which completes the proof.
\end{proof}
\end{appendix}


\begin{thebibliography}{26}
% BibTex style file: imsart-nameyear.bst, 2017-11-03
% Default style options (sort=1,type=nameyear).
% Used options (sort=1,type=nameyear).

\bibitem[\protect\citeauthoryear{Capp\'{e}}{2011}]{cappe:2009}
\begin{barticle}[author]
\bauthor{\bsnm{Capp\'{e}},~\bfnm{O.}\binits{O.}}
(\byear{2011}).
\btitle{Online {EM} Algorithm for Hidden {M}arkov Models}.
\bjournal{J. Comput. Graph. Statist.}
\bvolume{20}
\bpages{728--749}.
\end{barticle}
\endbibitem

\bibitem[\protect\citeauthoryear{Capp\'{e}, Moulines and
  Ryd\'{e}n}{2005}]{cappe:moulines:ryden:2005}
\begin{bbook}[author]
\bauthor{\bsnm{Capp\'{e}},~\bfnm{O.}\binits{O.}},
  \bauthor{\bsnm{Moulines},~\bfnm{E.}\binits{E.}} \AND
  \bauthor{\bsnm{Ryd\'{e}n},~\bfnm{T.}\binits{T.}}
(\byear{2005}).
\btitle{Inference in {H}idden {M}arkov Models}.
\bpublisher{Springer}.
\end{bbook}
\endbibitem

\bibitem[\protect\citeauthoryear{Chopin and
  Papaspiliopoulos}{2020}]{chopin:papaspiliopoulos:2020}
\begin{bbook}[author]
\bauthor{\bsnm{Chopin},~\bfnm{N.}\binits{N.}} \AND
  \bauthor{\bsnm{Papaspiliopoulos},~\bfnm{O.}\binits{O.}}
(\byear{2020}).
\btitle{An introduction to sequential Monte Carlo methods}.
\bpublisher{Springer}.
\end{bbook}
\endbibitem

\bibitem[\protect\citeauthoryear{{Del Moral}}{2004}]{delmoral:2004}
\begin{bbook}[author]
\bauthor{\bsnm{{Del Moral}},~\bfnm{P.}\binits{P.}}
(\byear{2004}).
\btitle{{F}eynman-Kac {F}ormulae. {G}enealogical and {I}nteracting {P}article
  {S}ystems with {A}pplications}.
\bpublisher{Springer}.
\end{bbook}
\endbibitem

\bibitem[\protect\citeauthoryear{{Del Moral}, Doucet and
  Singh}{2010}]{delmoral:doucet:singh:2010}
\begin{barticle}[author]
\bauthor{\bsnm{{Del Moral}},~\bfnm{P.}\binits{P.}},
  \bauthor{\bsnm{Doucet},~\bfnm{A.}\binits{A.}} \AND
  \bauthor{\bsnm{Singh},~\bfnm{S.~S.}\binits{S.~S.}}
(\byear{2010}).
\btitle{A backward interpretation of {F}eynman-{K}ac formulae}.
\bjournal{ESAIM: Mathematical Modelling and Numerical Analysis}
\bvolume{44}
\bpages{947--975}.
\end{barticle}
\endbibitem

\bibitem[\protect\citeauthoryear{Del~Moral, Doucet and
  Jasra}{2012}]{delmoral:doucet:jasra:2012}
\begin{barticle}[author]
\bauthor{\bsnm{Del~Moral},~\bfnm{Pierre}\binits{P.}},
  \bauthor{\bsnm{Doucet},~\bfnm{Arnaud}\binits{A.}} \AND
  \bauthor{\bsnm{Jasra},~\bfnm{Ajay}\binits{A.}}
(\byear{2012}).
\btitle{On adaptive resampling strategies for sequential {M}onte {C}arlo
  methods}.
\bjournal{Bernoulli}
\bvolume{18}
\bpages{252--278}.
\end{barticle}
\endbibitem

\bibitem[\protect\citeauthoryear{{Del Moral}
  et~al.}{2016}]{delmoral:moulines:olsson:verge:2016}
\begin{barticle}[author]
\bauthor{\bsnm{{Del Moral}},~\bfnm{P.}\binits{P.}},
  \bauthor{\bsnm{Moulines},~\bfnm{E.}\binits{E.}},
  \bauthor{\bsnm{Olsson},~\bfnm{J.}\binits{J.}} \AND
  \bauthor{\bsnm{Verg\'e},~\bfnm{C.}\binits{C.}}
(\byear{2016}).
\btitle{Convergence properties of weighted particle islands with application to
  the double bootstrap algorithm}.
\bjournal{Stochastic Systems}
\bvolume{2}
\bpages{367--419}.
\end{barticle}
\endbibitem

\bibitem[\protect\citeauthoryear{Douc and Moulines}{2008}]{douc:moulines:2008}
\begin{barticle}[author]
\bauthor{\bsnm{Douc},~\bfnm{R.}\binits{R.}} \AND
  \bauthor{\bsnm{Moulines},~\bfnm{E.}\binits{E.}}
(\byear{2008}).
\btitle{Limit theorems for weighted samples with applications to sequential
  {M}onte {C}arlo methods}.
\bjournal{Ann. Statist.}
\bvolume{36}
\bpages{2344--2376}.
\end{barticle}
\endbibitem

\bibitem[\protect\citeauthoryear{Douc, Moulines and
  Stoffer}{2014}]{douc:moulines:stoffer:2014}
\begin{bbook}[author]
\bauthor{\bsnm{Douc},~\bfnm{Randal}\binits{R.}},
  \bauthor{\bsnm{Moulines},~\bfnm{Eric}\binits{E.}} \AND
  \bauthor{\bsnm{Stoffer},~\bfnm{David}\binits{D.}}
(\byear{2014}).
\btitle{Nonlinear time series: Theory, methods and applications with R
  examples}.
\bpublisher{CRC press}.
\end{bbook}
\endbibitem

\bibitem[\protect\citeauthoryear{Douc
  et~al.}{2011}]{douc:garivier:moulines:olsson:2009}
\begin{barticle}[author]
\bauthor{\bsnm{Douc},~\bfnm{R.}\binits{R.}},
  \bauthor{\bsnm{Garivier},~\bfnm{A.}\binits{A.}},
  \bauthor{\bsnm{Moulines},~\bfnm{E.}\binits{E.}} \AND
  \bauthor{\bsnm{Olsson},~\bfnm{J.}\binits{J.}}
(\byear{2011}).
\btitle{Sequential {M}onte {C}arlo smoothing for general state space hidden
  {M}arkov models}.
\bjournal{Ann. Appl. Probab.}
\bvolume{21}
\bpages{1201--2145}.
\end{barticle}
\endbibitem

\bibitem[\protect\citeauthoryear{Doucet, Godsill and
  Andrieu}{2000}]{doucet:godsill:andrieu:2000}
\begin{barticle}[author]
\bauthor{\bsnm{Doucet},~\bfnm{A.}\binits{A.}},
  \bauthor{\bsnm{Godsill},~\bfnm{S.}\binits{S.}} \AND
  \bauthor{\bsnm{Andrieu},~\bfnm{C.}\binits{C.}}
(\byear{2000}).
\btitle{On sequential {M}onte-{C}arlo sampling methods for {B}ayesian
  filtering}.
\bjournal{Stat. Comput.}
\bvolume{10}
\bpages{197--208}.
\end{barticle}
\endbibitem

\bibitem[\protect\citeauthoryear{Gloaguen, Le~Corff and
  Olsson}{2021}]{gloaguen:lecorff:olsson:2021}
\begin{barticle}[author]
\bauthor{\bsnm{Gloaguen},~\bfnm{Pierre}\binits{P.}},
  \bauthor{\bsnm{Le~Corff},~\bfnm{Sylvain}\binits{S.}} \AND
  \bauthor{\bsnm{Olsson},~\bfnm{Jimmy}\binits{J.}}
(\byear{2021}).
\btitle{A pseudo-marginal sequential {M}onte {C}arlo online smoothing
  algorithm}.
\bnote{\texttt{https://arxiv.org/abs/1908.07254}}.
\end{barticle}
\endbibitem

\bibitem[\protect\citeauthoryear{Godsill, Doucet and
  West}{2004}]{godsill:doucet:west:2004}
\begin{barticle}[author]
\bauthor{\bsnm{Godsill},~\bfnm{S.~J.}\binits{S.~J.}},
  \bauthor{\bsnm{Doucet},~\bfnm{A.}\binits{A.}} \AND
  \bauthor{\bsnm{West},~\bfnm{M.}\binits{M.}}
(\byear{2004}).
\btitle{Monte {C}arlo smoothing for non-linear time series}.
\bjournal{J. Am. Statist. Assoc.}
\bvolume{50}
\bpages{438--449}.
\end{barticle}
\endbibitem

\bibitem[\protect\citeauthoryear{Gordon, Salmond and
  Smith}{1993}]{gordon:salmond:smith:1993}
\begin{barticle}[author]
\bauthor{\bsnm{Gordon},~\bfnm{N.}\binits{N.}},
  \bauthor{\bsnm{Salmond},~\bfnm{D.}\binits{D.}} \AND
  \bauthor{\bsnm{Smith},~\bfnm{A.~F.}\binits{A.~F.}}
(\byear{1993}).
\btitle{Novel approach to nonlinear/non-{G}aussian {B}ayesian state
  estimation}.
\bjournal{IEE Proc. F, Radar Signal Process.}
\bvolume{140}
\bpages{107--113}.
\end{barticle}
\endbibitem

\bibitem[\protect\citeauthoryear{Hull and White}{1987}]{hull:white:1987}
\begin{barticle}[author]
\bauthor{\bsnm{Hull},~\bfnm{J.}\binits{J.}} \AND
  \bauthor{\bsnm{White},~\bfnm{A.}\binits{A.}}
(\byear{1987}).
\btitle{The pricing of options on assets with stochastic volatilities}.
\bjournal{J. Finance}
\bvolume{42}
\bpages{281--300}.
\end{barticle}
\endbibitem

\bibitem[\protect\citeauthoryear{Kantas
  et~al.}{2015}]{kantas:doucet:singh:chopin:2015}
\begin{barticle}[author]
\bauthor{\bsnm{Kantas},~\bfnm{Nikolas}\binits{N.}},
  \bauthor{\bsnm{Doucet},~\bfnm{Arnaud}\binits{A.}},
  \bauthor{\bsnm{Singh},~\bfnm{Sumeetpal~S}\binits{S.~S.}},
  \bauthor{\bsnm{Maciejowski},~\bfnm{Jan}\binits{J.}},
  \bauthor{\bsnm{Chopin},~\bfnm{Nicolas}\binits{N.}} \betal{et~al.}
(\byear{2015}).
\btitle{On particle methods for parameter estimation in state-space models}.
\bjournal{Statistical science}
\bvolume{30}
\bpages{328--351}.
\end{barticle}
\endbibitem

\bibitem[\protect\citeauthoryear{Kitagawa and Sato}{2001}]{kitagawa:sato:2001}
\begin{bincollection}[author]
\bauthor{\bsnm{Kitagawa},~\bfnm{G.}\binits{G.}} \AND
  \bauthor{\bsnm{Sato},~\bfnm{S.}\binits{S.}}
(\byear{2001}).
\btitle{Monte {C}arlo smoothing and self-organising state-space model}.
In \bbooktitle{Sequential Monte Carlo methods in practice}.
\bseries{Stat. Eng. Inf. Sci.}
\bpages{177--195}.
\bpublisher{Springer}, \baddress{New York}.
\bmrnumber{MR1847792}
\end{bincollection}
\endbibitem

\bibitem[\protect\citeauthoryear{Koskela et~al.}{2020}]{koskela2020}
\begin{barticle}[author]
\bauthor{\bsnm{Koskela},~\bfnm{Jere}\binits{J.}},
  \bauthor{\bsnm{Jenkins},~\bfnm{Paul~A.}\binits{P.~A.}},
  \bauthor{\bsnm{Johansen},~\bfnm{Adam~M.}\binits{A.~M.}} \AND
  \bauthor{\bsnm{Span\`{o}},~\bfnm{Dario}\binits{D.}}
(\byear{2020}).
\btitle{Asymptotic genealogies of interacting particle systems with an
  application to sequential Monte Carlo}.
\bjournal{Ann. Statist.}
\bvolume{48}
\bpages{560--583}.
\bdoi{10.1214/19-AOS1823}
\end{barticle}
\endbibitem

\bibitem[\protect\citeauthoryear{{Le~Gland} and
  Mevel}{1997}]{legland:mevel:1997}
\begin{binproceedings}[author]
\bauthor{\bsnm{{Le~Gland}},~\bfnm{F.}\binits{F.}} \AND
  \bauthor{\bsnm{Mevel},~\bfnm{L.}\binits{L.}}
(\byear{1997}).
\btitle{Recursive estimation in {HMM}s}.
In \bbooktitle{Proc. IEEE Conf. Decis. Control}
\bpages{3468--3473}.
\end{binproceedings}
\endbibitem

\bibitem[\protect\citeauthoryear{Liu}{1996}]{liu:1996}
\begin{barticle}[author]
\bauthor{\bsnm{Liu},~\bfnm{J.~S.}\binits{J.~S.}}
(\byear{1996}).
\btitle{Metropolized independent sampling with comparisons to rejection
  sampling and importance sampling}.
\bjournal{Stat. Comput.}
\bvolume{6}
\bpages{113--119}.
\end{barticle}
\endbibitem

\bibitem[\protect\citeauthoryear{Mongillo and
  Den\`{e}ve}{2008}]{mongillo:deneve:2008}
\begin{barticle}[author]
\bauthor{\bsnm{Mongillo},~\bfnm{G.}\binits{G.}} \AND
  \bauthor{\bsnm{Den\`{e}ve},~\bfnm{S.}\binits{S.}}
(\byear{2008}).
\btitle{Online Learning with Hidden {M}arkov Models}.
\bjournal{Neural Computation}
\bvolume{20}
\bpages{1706--1716}.
\bdoi{10.1162/neco.2008.10-06-351}
\end{barticle}
\endbibitem

\bibitem[\protect\citeauthoryear{Olsson and Douc}{2019}]{olsson:douc:2017}
\begin{barticle}[author]
\bauthor{\bsnm{Olsson},~\bfnm{J.}\binits{J.}} \AND
  \bauthor{\bsnm{Douc},~\bfnm{R.}\binits{R.}}
(\byear{2019}).
\btitle{Numerically stable online estimation of variance in particle filters}.
\bjournal{Bernoulli}
\bvolume{25}
\bpages{1504--1535}.
\end{barticle}
\endbibitem

\bibitem[\protect\citeauthoryear{Olsson and
  Westerborn}{2017}]{olsson:westerborn:2017}
\begin{barticle}[author]
\bauthor{\bsnm{Olsson},~\bfnm{J.}\binits{J.}} \AND
  \bauthor{\bsnm{Westerborn},~\bfnm{J.}\binits{J.}}
(\byear{2017}).
\btitle{Efficient particle-based online smoothing in general hidden {M}arkov
  models: The {P}a{RIS} algorithm}.
\bjournal{Bernoulli}
\bvolume{23}
\bpages{1951--1996}.
\end{barticle}
\endbibitem

\bibitem[\protect\citeauthoryear{Olsson
  et~al.}{2008}]{olsson:cappe:douc:moulines:2006}
\begin{barticle}[author]
\bauthor{\bsnm{Olsson},~\bfnm{J.}\binits{J.}},
  \bauthor{\bsnm{Capp\'e},~\bfnm{O.}\binits{O.}},
  \bauthor{\bsnm{Douc},~\bfnm{R.}\binits{R.}} \AND
  \bauthor{\bsnm{Moulines},~\bfnm{E.}\binits{E.}}
(\byear{2008}).
\btitle{Sequential {M}onte {C}arlo smoothing with application to parameter
  estimation in non-linear state space models}.
\bjournal{Bernoulli}
\bvolume{14}
\bpages{155--179}.
\end{barticle}
\endbibitem

\bibitem[\protect\citeauthoryear{Pitt and Shephard}{1999}]{pitt:shephard:1999}
\begin{barticle}[author]
\bauthor{\bsnm{Pitt},~\bfnm{M.~K.}\binits{M.~K.}} \AND
  \bauthor{\bsnm{Shephard},~\bfnm{N.}\binits{N.}}
(\byear{1999}).
\btitle{Filtering Via Simulation: Auxiliary Particle Filters}.
\bjournal{J. Am. Statist. Assoc.}
\bvolume{94}
\bpages{590--599}.
\end{barticle}
\endbibitem

\bibitem[\protect\citeauthoryear{Poyiadjis, Doucet and
  Singh}{2011}]{poyiadjis:doucet:singh:2011}
\begin{barticle}[author]
\bauthor{\bsnm{Poyiadjis},~\bfnm{G.}\binits{G.}},
  \bauthor{\bsnm{Doucet},~\bfnm{A.}\binits{A.}} \AND
  \bauthor{\bsnm{Singh},~\bfnm{S.~S.}\binits{S.~S.}}
(\byear{2011}).
\btitle{Particle approximations of the score and observed information matrix in
  state space models with application to parameter estimation}.
\bjournal{Biometrika}
\bvolume{98}
\bpages{65--80}.
\end{barticle}
\endbibitem

\end{thebibliography}
\end{document}